%% file: reconfig.tex
\documentclass[a4paper,USenglish]{lipics-v2021}

\newif\iflong
\longtrue %

\newcommand{\tr}[2]{\iflong{}\S#1\else{}\cite[\S#2]{ext}\fi}

\newcommand{\ncorr}{A}
\newcommand{\nprproof}{B}
\newcommand{\npoproof}{C}

\usepackage{import}
\usepackage{mathrsfs}
\usepackage{amsthm}
\usepackage{xspace}
\usepackage{latexsym}
\usepackage{amssymb}
\usepackage{graphicx}
\usepackage{amsthm}
\usepackage{xcolor}
\usepackage{amsmath}
\usepackage{hyperref}
\usepackage{afterpage}
\usepackage{wrapfig}
\usepackage{setspace}

\usepackage{enumerate}
\usepackage{float}
\usepackage{adjustbox}

\usepackage[vlined,linesnumbered,noresetcount]{algorithm2e}

\usepackage{stmaryrd}
\SetSymbolFont{stmry}{bold}{U}{stmry}{m}{n}

\usepackage{pgfplots}

\usepackage{tikz-cd}
\usetikzlibrary{backgrounds}

\usetikzlibrary{arrows, positioning, shapes.geometric}
\usetikzlibrary{decorations.pathreplacing}
\usetikzlibrary{arrows,decorations.markings}
\usetikzlibrary{calc,decorations.pathmorphing,shapes}
\newcommand{\Cross}{\mathbin{\tikz [x=1.4ex,y=1.4ex,line width=.2ex, red] \draw (0,0) -- (1,1) (0,1) -- (1,0);}}%

\newcommand{\ag}[1]{\todo[color=red!30]{AG:\@ #1}}
\renewcommand{\ag}[1]{}

\usepackage{cite}

\usepackage{xargs}
\usepackage[textsize=tiny]{todonotes}
\input{macro}

\makeatletter
\newcommand{\removelatexerror}{\let\@latex@error\@gobble}
\makeatother

\nolinenumbers

\authorrunning{M. Bravo et al.}
\titlerunning{Vertical Atomic Broadcast and Passive Replication}


\funding{This work was partially supported by the projects BYZANTIUM, DECO and
  PRODIGY funded by MCIN/AEI, and REDONDA funded by the CHIST-ERA network.}

\Copyright{Manuel Bravo, Gregory Chockler, Alexey Gotsman, Alejandro Naser-Pastoriza, and Christian Roldán}

\iflong
{}
\else
\relatedversiondetails[cite={ext}]{Extended Version}{https://arxiv.org/abs/......}
\fi

\iflong
\title{Vertical Atomic Broadcast and Passive Replication\\ (Extended Version)}
\else
\title{Vertical Atomic Broadcast and Passive Replication}
\fi

\author{Manuel Bravo}{Informal Systems, Madrid, Spain}{}{}{}
\author{Gregory Chockler}{University of Surrey, Guildford, UK}{}{}{}
\author{Alexey Gotsman}{IMDEA Software Institute, Madrid, Spain}{}{}{}
\author{Alejandro Naser-Pastoriza}{IMDEA Software Institute, Madrid, Spain \and
        Universidad Politécnica de Madrid, Spain}{}{}{}
\author{Christian Roldán}{IMDEA Software Institute, Madrid, Spain}{}{}{}

\EventEditors{Dan Alistarh}
\EventNoEds{1}
\EventLongTitle{38th International Symposium on Distributed Computing (DISC 2024)}
\EventShortTitle{DISC 2024}
\EventAcronym{DISC}
\EventYear{2024}
\EventDate{October 28--November 1, 2024}
\EventLocation{Madrid, Spain}
\EventLogo{}
\SeriesVolume{319}
\ArticleNo{7}

\iflong
\makeatletter
\def\@DOIPrefix{}
\makeatother
\fi

\begin{document}

\maketitle

\begin{abstract} 
  Atomic broadcast is a reliable communication abstraction ensuring that all
  processes deliver the same set of messages in a common global order. It is a
  fundamental building block for implementing fault-tolerant services using
  either active (aka state-machine) or passive (aka primary-backup)
  replication. We consider the problem of implementing reconfigurable atomic
  broadcast, which further allows users to dynamically alter the set of
  participating processes, e.g., in response to failures or changes in the
  load. We give a complete safety and liveness specification of this
  communication abstraction and propose a new protocol implementing it, called
  Vertical Atomic Broadcast, which uses an auxiliary service to facilitate
  reconfiguration. In contrast to prior proposals, our protocol significantly
  reduces system downtime when reconfiguring from a functional configuration by
  allowing it to continue processing messages while agreement on the next
  configuration is in progress.  Furthermore, we show that this advantage can be
  maintained even when our protocol is modified to support a stronger variant of
  atomic broadcast required for passive replication.
\end{abstract}

\iflong
\smallskip
\smallskip
\smallskip
\smallskip
\smallskip
\smallskip
\fi

\input{intro}

\input{model}

\input{spec}

\input{protocol}

\input{po}

\input{related}

\bibliography{biblio}

\iflong
\clearpage
\appendix

\makeatletter
    \def\@evenhead{\large\sffamily\bfseries
                   \llap{\hbox to0.5\oddsidemargin{\thepage\hss}}APPENDIX\hfil}%
    \def\@oddhead{\large\sffamily\bfseries APPENDIX\hfil
                  \rlap{\hbox to0.5\oddsidemargin{\hss\thepage}}}%
\makeatother

\input{corr}

\input{pr-proof}

\input{app-po-proof}

\fi

\end{document}

%% file: macro.tex
\newtheorem{assumption}{Assumption}
\newtheorem{invariant}{Invariant}

\SetKwProg{Function}{function}{:}{}
\SetKwBlock{SubAlgoBlock}{}{end}
\newcommand{\SubAlgo}[2]{#1 \SubAlgoBlock{#2}}

\newcommand{\reconf}{{\tt{reconfigure}}}
\newcommand{\reconfreq}{{\tt{reconfig\_req}}}
\newcommand{\reconfresp}{{\tt{reconfig\_resp}}}

\newcommand{\RESULT}{{\tt RESULT}}
\newcommand{\EXECUTE}{{\tt EXECUTE}}
\newcommand{\ACCEPT}{{\tt ACCEPT}}
\newcommand{\ACCEPTACK}{{\tt ACCEPT\_ACK}}

\newcommand{\NEWSTATEACK}{{\tt NEW\_STATE\_ACK}}
\newcommand{\NEWSTATE}{{\tt NEW\_STATE}}
\newcommand{\COMMIT}{{\tt COMMIT}}
\newcommand{\PROBE}{{\tt PROBE}}
\newcommand{\PROBEACK}{{\tt PROBE\_ACK}}
\newcommand{\NEWCONFIG}{{\tt NEW\_CONFIG}}

\newcommand{\FORWARD}{{\tt FORWARD}}

\newcommand{\onreceive}{\textbf{when received}\xspace}
\newcommand{\send}{\textbf{send}\xspace}
\newcommand{\from}{\textbf{from}\xspace}

\newcommand{\To}{\textbf{to}\xspace}
\newcommand{\precond}{\textbf{pre:}\xspace}

\newcommand{\assign}[2]{\ensuremath{#1} \ensuremath{\leftarrow} \ensuremath{#2}}

\newcommand{\Consistency}{Prefix Consistency}

\newcommand{\State}{\mathcal{S}}
\newcommand{\Command}{\mathcal{C}}
\newcommand{\cmd}{c}
\newcommand{\execute}{{\tt{execute}}}
\newcommand{\id}{\mathit{id}}
\newcommand{\client}{{\it origin}}

\newcommand{\pnext}{\mathsf{next}}
\newcommand{\initlength}{\mathsf{init\_len}}
\newcommand{\msg}{\mathsf{msg}}
\newcommand{\status}{\mathsf{status}}

\newcommand{\lastdelivered}{\mathsf{last\_delivered}}

\newcommand{\members}{\mathsf{members}}

\newcommand{\pmembers}{\vm}

\newcommand{\epoch}{\mathsf{epoch}}

\newcommand{\pepoch}{e}
\newcommand{\newepoch}{\mathsf{new\_epoch}}

\newcommand{\repoch}{{\eid_{\rm new}}}

\newcommand{\rpmembers}{{\pmembers_{\rm new}}}
\newcommand{\leadervar}{\mathsf{leader}}

\newcommand{\curleader}{\mathsf{cur\_leader}}
\newcommand{\curepoch}{\mathsf{cur\_epoch}}

\newcommand{\vmsg}{\mathit{msg}}
\newcommand{\vm}{M}

\newcommand{\false}{\textsc{false}}
\newcommand{\true}{\textsc{true}}

\newcommand{\LEADER}[0]{\textsc{leader}}
\newcommand{\FOLLOWER}{\textsc{follower}}
\newcommand{\FRESH}{\textsc{fresh}}

\newcommand{\where}{\text{~where~}}
\newcommand{\suchthat}{\text{~such~that~}}

\allowdisplaybreaks

\DeclareMathAlphabet{\mathsfit}{T1}{\sfdefault}{\mddefault}{\sldefault}
\SetMathAlphabet{\mathsfit}{bold}{T1}{\sfdefault}{\bfdefault}{\sldefault}

\newcommandx{\pid}[2][1=i,2=p,usedefault=@]{{{#2_{#1}}}}
\newcommandx{\aleader}[2][1={}, 2=l,usedefault=@]{{#2}#1}

\newcommandx{\vid}[2][1={}, 2=m,usedefault=@]{{{#2#1}}}
\newcommandx{\act}[1][1=act,usedefault=@]{{\bf {#1}}}

\newcommandx{\aid}[2][1={}, 2=a,usedefault=@]{{{#2#1}}}
\newcommandx{\hrel}[1][1=\hid,usedefault=@]{\prec_{#1}}

\newcommandx{\effid}[2][1={}, 2=\delta,usedefault=@]{{{#2#1}}}

\newcommandx{\oid}[2][1={}, 2=o,usedefault=@]{{{#2#1}}}
\newcommandx{\hid}[2][1={}, 2=h,usedefault=@]{{{#2#1}}}
\newcommandx{\eid}[2][1={}, 2=e,usedefault=@]{{{#2#1}}}
\newcommandx{\rid}[2][1={}, 2=r,usedefault=@]{{{#2#1}}}
\newcommandx{\confid}[2][1={}, 2=C,usedefault=@]{{{#2#1}}}

\newcommand{\msgi}{{msg}}

\newcommandx{\prefixmsg}[2][1=k, 2=\msg, usedefault=@]{#2{\downharpoonleft_{#1}}}
\newcommandx{\prefix}[2]{#1{\downharpoonleft_{#2}}}

\newcommand{\length}{\mathsf{length}}

\newcommand{\natset}{\mathbb{N}}
\newcommand{\msgset}{\mathsf{Msg}}
\newcommand{\procset}{\mathcal{P}}
\newcommand{\epochset}{\natset}
\newcommand{\confset}{\mathsf{Config}}

\newcommand{\epochOf}{{\rm epochOf}}
\newcommandx{\aconf}[3][1=\eid, 2=\vm, 3=\pid, usedefault=@]{\langle#1, #2, #3\rangle}

\newcommand{\intro}{{\tt{introduction}}}
\newcommand{\bcast}{{\tt{broadcast}}}

\newcommand{\deliv}{{\tt{deliver}}}

\newcommand{\invref}[1]{Invariant~\ref{#1}}

\newcommand{\asmref}[1]{Assumption~\ref{#1}}

\newcommand{\thrmref}[1]{Theorem~\ref{#1}}
\newcommand{\figref}[1]{Figure~\ref{#1}}
\newcommand{\equref}[1]{(\ref{#1})}

\newcommand{\primary}{{\rm leader}}

\newcommand{\app}{\mathit{apply}}
\newcommandx{\rval}[2][1=\effid, 2=\rid, usedefault=@]{\langle#1, #2\rangle}

\newcommand{\confchng}{{\tt{conf\_changed}}}
\newcommand{\pstate}{\Sigma}

\newcommand{\tstate}{\Theta}
\newcommand{\return}{\textbf{return}\xspace}
\newcommandx{\totrel}[1][1=\Delta, usedefault=@]{<_{#1}}

\newcommandx{\myeq}[1][1=]{\mathrel{\stackrel{\makebox[0pt]{\mbox{\normalfont\tiny #1}}}{\equiv}}}
\newcommandx{\myimplies}[1][1=]{\mathrel{\stackrel{\makebox[0pt]{\mbox{\normalfont\tiny #1}}}{\implies}}}

\newcommand{\confchgnot}{Basic Configuration Change Properties}

\renewcommand{\_}{\texttt{\textunderscore}}

\newcounter{sarrow}

%% file: intro.tex
\section{Introduction}

Replication is a widely used technique for ensuring fault tolerance of
distributed services. Two common replication approaches are \emph{active} (aka
state-machine) replication~\cite{smr} and \emph{passive} (aka primary-backup)
replication~\cite{budhiraja1993primary}. In active replication, a service is
defined by a deterministic state machine and is executed on several replicas,
each maintaining a copy of the machine. The replicas are kept in sync using
\emph{atomic broadcast}~\cite{to-survey}, which ensures that client commands are
delivered in the same order to all replicas; this can be implemented using,
e.g., Multi-Paxos~\cite{paxos}.

In contrast, in passive replication commands are executed by a single replica
(the \emph{leader} or \emph{primary}), which propagates the state updates induced
by the commands to the other replicas (\emph{followers} or \emph{backups}). This
approach allows replicating services with non-deterministic operations, e.g.,
those depending on timeouts or interrupts. But as shown in~\cite{ken-book,
  zookeeper, junqueira2013barriers}, implementing it requires propagating
updates from the leader to the followers using a stronger primitive than the
classical atomic broadcast. This is because in passive replication, a state
update is incremental with respect to the state it was generated in. Hence, to
ensure consistency between replicas, each update must be applied by a follower
to the same state in which it was generated by the leader. Junqueira et
al. formalized the corresponding guarantees by the notion of \emph{primary-order
  atomic broadcast (POabcast)}~\cite{zab,junqueira2013barriers}, which can be
implemented by protocols such as Zab~\cite{zab}, viewstamped
replication~\cite{vr} or Raft~\cite{raft}.

The above implementations of atomic or primary-order atomic broadcast require
replicating data among $2f+1$ replicas to tolerate $f$ %
failures. This is expensive: in principle, storing the data at $f+1$ replicas is 
enough for it survive $f$ failures. Since with only $f+1$ replicas even a single
replica failure will block the system, to recover we need to \emph{reconfigure}
it, i.e., change its membership to replace failed replicas with fresh
ones. Unfortunately, processes concurrently deciding to reconfigure the system
need to be able to agree on the next configuration; this reduces to solving
consensus, which again requires $2f+1$ replicas~\cite{lower-bound}. The way out
of this conundrum is to use a separate \emph{configuration service} with $2f+1$
replicas to perform consensus on configurations. In this way we use $2f+1$
replicas to only store configuration metadata and $f+1$ replicas to store the
actual data. This \emph{vertical approach}, layering replication on top of a
configuration service, was originally proposed in RAMBO~\cite{rambo}
for atomic registers and 
in \emph{Vertical Paxos}~\cite{vertical-paxos} for single-shot consensus. Since
then it has been used by many practical storage
systems~\cite{corfu,bigtable,farm,kafka}. These often use reconfiguration not
only to deal with failures, but also to make changes to a functional
configuration: e.g., to move replicas from highly loaded machines to lightly
loaded ones, or to change the number of machines replicating the
service~\cite{smart,kafka-book,matchmaker}.

Unfortunately, while the space of atomic broadcast protocols with $2f+1$
replicas has been extensively explored, the design of such protocols in 
vertical settings is poorly understood.  Even though one can obtain a vertical
solution for atomic broadcast by reducing it to Vertical Paxos, this would make
it hard to ensure the additional properties required for passive replication.
\ag{We could also mention that Rex paper Grisha found, if it's relevant.}
Furthermore, both Vertical Paxos and similar protocols~\cite{ken-book} stop the
system as the very first step of reconfiguration, which increases the downtime
when reconfiguring from a functional configuration. Due to the absence of a
theoretically grounded and efficient atomic broadcast protocol for vertical
settings, the designs used in industry are often ad hoc and buggy. For example,
until recently the vertical-style protocol used in Kafka, a widely used
streaming platform, contained a number of bugs in its failure
handling~\cite{kafka}. In this paper we make several contributions to improve
this situation.

First, we give a complete safety and liveness specification of
\emph{reconfigurable atomic broadcast}, sufficient for active replication 
(\S\ref{sec:spec}). We then propose its implementation in a vertical system with
$f + 1$ replicas and an external configuration service, which we call
\emph{Vertical Atomic Broadcast (VAB)} (\S\ref{sec:vp_protocol}). In contrast to 
prior vertical protocols~\cite{vertical-paxos,ken-book}, our implementation
allows the latest functional configuration to continue processing messages while
agreement on the next configuration is in progress. This reduces the downtime
when reconfiguring from a functional configuration from $4$ message delays in
the prior solutions to $0$.
We rigorously prove that the protocol correctly implements the reconfigurable
atomic broadcast specification, including both safety and liveness.

We next consider the case of passive replication (which we review in
\S\ref{sec:po}). We propose \emph{speculative primary-order atomic broadcast
  (SPOabcast)}, which we show to be sufficient for implementing passive
replication in a reconfigurable system (\S\ref{sec:srpob}). A key novel aspect
of SPOabcast is that SPOabcast is able to completely eliminate the downtime
induced by a \emph{Primary Integrity} property of the existing
POabcast~\cite{zab,junqueira2013barriers}. This property requires the leader of
a new configuration to suspend normal operation until an agreement is reached on
which messages broadcast in the previous configurations should survive in the
new one: in passive replication, these messages determine the initial service
state at the leader. Instead, SPOabcast allows the leader to \emph{speculatively}
deliver a tentative set of past messages before the agreement on them has been
reached, and then to immediately resume normal broadcasting. SPOabcast
guarantees that, if a process delivers a message $m_2$ broadcast by the new
leader, then prior to this the process will also deliver every message $m_1$ the
leader speculatively delivered before broadcasting $m_2$. This helps ensure that
the process applies the update in $m_2$ to the same state in which the leader
generated it, as required for the correctness of passive replication.

We show that SPOabcast can be implemented by modifying our Vertical Atomic
Broadcast protocol. The use of speculative delivery allows the resulting
protocol to preserve VAB's downtime of $0$ when reconfiguring from a functional
configuration. It thus allows using Vertical Atomic Broadcast to replicate
services with non-deterministic operations.

Overall, we believe that our specifications, protocols and correctness proofs
provide insights into the principles underlying existing reconfigurable systems,
and can serve as a blueprint for building future ones.

%% file: model.tex
\section{System Model}
\label{sec:model}

We consider an asynchronous message-passing system consisting of an (infinite) universe of
processes $\procset$ which may fail by \emph{crashing}, i.e., permanently
stopping execution. A process is \emph{correct} if it never crashes, and
\emph{faulty} otherwise. Processes are connected by reliable FIFO channels:
messages are delivered in FIFO order, and messages between non-faulty processes
are guaranteed to be eventually delivered. The system moves through a sequence
of \emph{configurations}. A configuration $\confid$ is a triple $\aconf$ that
consists of an epoch $\eid\in\epochset$ identifying the configuration, a
finite set of processes $\vm \subseteq \procset$ that belong to the configuration, and a
distinguished \emph{leader} process $\pid\in\vm$. We denote the set of
configurations by $\confset$. In contrast to static systems, we do not impose a
fixed global bound on the number of faulty processes, but formulate our
availability assumptions relative to specific configurations (\S\ref{sec:spec}).

\emph{Reconfiguration} is the process of changing the system
configuration. 
We assume that configurations are stored in 
an external \emph{configuration service (CS)}, which is 
reliable and wait-free.
The configuration service provides three atomic
operations. An operation {\tt
  compare\_and\_swap}$(\eid, \langle \eid', \vm, p_l\rangle)$ succeeds iff the
epoch of the last stored configuration is $\eid$; in this case it stores the
provided configuration with a higher epoch $\eid' > \eid$.
Operations {\tt get\_last\_epoch}$()$ and {\tt get\_members}$(e)$ respectively
return the last epoch and the members associated with a given epoch $\eid$.

In practice, a configuration service can be implemented under partial synchrony
using Paxos-like replication over $2f+1$ processes out of which at most $f$ can
fail~\cite{lower-bound} (as is done in systems such as
Zookeeper~\cite{zookeeper}). Our protocols use the service as a black box, and
as a result, do not require any further environment assumptions about
timeliness~\cite{dls} or failure detection~\cite{cht96-2}.

%% file: spec.tex
\section{Specification}
\label{sec:spec}

In this section we introduce \emph{reconfigurable atomic broadcast}, a variant of
atomic broadcast~\cite{to-survey} that allows reconfiguration. The broadcast
service allows a process to send an \emph{application message} $\vid$ from a set
$\msgset$ using a call $\bcast(\vid)$. Messages are delivered using a
notification $\deliv(\vid)$. Any process
may initiate system reconfiguration using a call $\reconf()$. If
successful, this returns the new configuration $\confid$ arising as a
result; otherwise it returns $\bot$. Each process participating in the new
configuration then gets a notification $\confchng(\confid)$, informing it about
$\confid$.
In practice, $\reconf$ would take as a parameter a
description of the desired reconfiguration. For simplicity we abstract from this
in our specification, which states broadcast correctness for any results of
reconfigurations.

We record the interactions between the broadcast and its users via
\emph{histories} $\hid$ -- sequences of \emph{actions} $\aid$ of
one of the following forms:
$$
\begin{array}{@{}l@{}}
\bcast_i(\vid), \quad
\deliv_i(\vid), \quad
\confchng_i(\confid), 
\\[1pt]
\reconfreq_i, \quad
\reconfresp_i(\confid), \quad
\intro_i(\confid),
\end{array}
$$
where $\pid\in\procset$, $\vid\in\msgset$ and $\confid\in \confset$. Each action
is parameterized by a process $\pid$ where it occurs (omitted when 
irrelevant). The first three actions respectively record invocations
of $\bcast$, $\deliv$ and $\confchng$. The next pair of actions record calls to
and returns from the $\reconf$ function. Finally, the $\intro$ action records
the moment when this function stores the new configuration in the configuration
service.

For a history $h$ we let $h_k$ be the $k$-th action in $h$, and we write
$\aid \in \hid$ if $\aid$ occurs in $\hid$. We also write $\_$ for an irrelevant
value.
We only consider %
histories where calls to and returns from $\reconf$
match, and a process may perform at most one $\intro$ action during the
execution of $\reconf$. \ag{There are probably more well-formedness conditions,
  although don't know how detailed we want to be.} For simplicity we assume that
all application messages broadcast in a single execution are unique:
\iflong
\begin{equation}\label{prop:env-1}
\else
\begin{equation*}
\fi
\forall \vid, k, l.\ h_k = \bcast(\vid) \ \wedge\  h_l = \bcast(\vid) 
\implies  k = l.
\iflong
\end{equation}
\else
\end{equation*}
\fi
For a history $\hid$, a partial function
$\epochOf: \mathbb{N} \rightharpoonup \epochset$ returns the epoch of the
action in $\hid$ with a given index. This is the epoch of the latest preceding
$\confchng$ at the same process:
\begin{align*}
 \epochOf(k) = \eid 
  \iff  {} &    (\exists i, l, a.\, h_k = a_i \ \wedge\ h_{l} = 
            \confchng_i(\aconf[@][\_][\_]) \ \wedge\ l < k\ \wedge\ {}\\
& \phantom{(} \forall l'.\ l < l' < k \implies h_{l'} \not= \confchng_i(\aconf[\_][\_][\_])).
\end{align*}
When $\epochOf(k) = e$, we say that the action $h_k$ occurs in $e$.

\begin{figure}[t]
\small
		\begin{enumerate}
			\item \textbf{\confchgnot.} \label{prop:basicconf}
			\begin{enumerate} 
				\item \label{prop:wf-1}
				Any epoch $\eid$ is associated with unique membership and leader:
				\[
                                   \begin{array}{@{}l@{}l@{}}
				\forall \eid, i, j, \vm_1, \vm_2.\
                                     \confchng(\aconf[@][\vm_1]) \in \hid \ \wedge\
                                     \confchng(\aconf[@][\vm_2][{\pid[j]}]) \in \hid
                                     \implies {}
                                     \\[1pt] \pid = \pid[j]  \ \wedge\ \vm_1 = \vm_2
                                     \end{array}
				\]
				
				\item \label{prop:wf-3}
				If a process $\pid$ joins a configuration $\confid =
				\aconf[\_][@][\_]$, then $\pid$ is a member of $\vm$: 
				\[
				\forall i, \vm.\
				\confchng_i(\_, \vm, \_) \in h \implies \pid\in \vm
				\]

				\item \label{prop:wf-4}
				Processes join configurations with monotonically increasing epochs: 
				\[
                                  \begin{array}{@{}l@{}}
				\forall \eid_1, \eid_2, i, k, l. \ h_k =
                                    \confchng_i(\eid_1,\_, \_) \ \wedge\
                                    h_l = \confchng_i(\eid_2, \_, \_) \ \wedge\
                                    k < l \implies {}
                                    \\[1pt] \eid_1<\eid_2
                                    \end{array}
				\]
				
				\item \label{prop:wf-5}
				Any configuration a process joins is introduced;
			 a configuration is introduced at most once:
				\[
				\begin{array}{@{}l@{}}
				\forall \confid.\ 
				(\confchng(\confid) \in \hid \implies 
				\intro(\confid) \in \hid) \ \wedge\  {} \\[1pt] 
				(\forall k, l.\ h_k = \intro(\confid) 
                                  \ \wedge\
                                  h_l =\intro(\confid) \implies k=l)
				\end{array}
				\]
			\end{enumerate}
			\item \textbf{Integrity.}
			\label{prop:integrity}
			A process delivers a given application message $\vid$ at most once,
				and only if $\vid$ was previously broadcast:
                        \[
                       \begin{array}{@{}l@{}}      
			\forall \vid, i, k,l.\ h_k=\deliv_i(\vid)  \ \wedge\
			h_l=\deliv_i(\vid) 
                          \implies {} \\[1pt] k = l \ \wedge\  \exists j.\
                          h_j = \bcast(\vid) \ \wedge\  j < k
                         \end{array}
                         \]

                      \item \textbf{Total Order.}
			\label{prop:total}
                        If some process delivers $\vid[_1]$ before $\vid[_2]$,
                        then any process that delivers $\vid[_2]$ must also deliver $\vid[_1]$ before
                        this:
                        \[
                          \begin{array}{@{}l@{}}
                            \forall \vid[_1], \vid[_2], i, j, k, l, l'. \
                            h_k = \deliv_i(\vid[_1]) \ \wedge\
                            h_l = \deliv_i(\vid[_2]) \ \wedge\ k < l \ \wedge\  {}
                            \\[1pt] 
                            h_{l'}=\deliv_j(\vid[_2]) \implies 
                            \exists k'.\ h_{k'} = \deliv_j(\vid_1) \ \wedge\ k' < l'
                          \end{array}
                        \]

		 \item \textbf{Agreement.}
		\label{prop:agreement}
		If $\pid$ delivers $\vid[_1]$ and $\pid[j]$
                delivers $\vid[_2]$, then either $\pid$ delivers $\vid[_2]$ or
                $\pid[j]$ delivers $\vid[_1]$: 
		\[
                  \begin{array}{@{}l@{}}
		\forall \vid[_1], \vid[_2], i, j.\ 
		\deliv_i(\vid[_1])\in\hid \ \wedge\ \deliv_j(\vid[_2])\in\hid 
                    \implies  {} \\[2pt]
                    (\deliv_i(\vid[_2]) \in \hid \ \vee\ \deliv_j(\vid_1) \in \hid)
                 \end{array}
                 \]

			\item \textbf{Liveness.} \label{prop:liveness} Consider an execution
			with finitely many reconfiguration requests ($\reconfreq$), and let $r$ be the
			last reconfiguration request to be invoked. Suppose that $r$ is
			invoked by a correct process and no other reconfiguration call
			takes steps after $r$ is invoked. Then $r$ terminates, having
                        introduced a configuration $\confid = \langle e, M, p_i \rangle$:
                        $\reconfresp(C)$. Furthermore, if all processes in $M$ are correct, then:
			\begin{enumerate}
				\item \label{prop:liveness-a} all processes in
                                  $M$ deliver $\confchng(\confid)$; 
				\item \label{prop:liveness-b} if $\pid \in M$
                                  broadcasts %
                                  $\vid$ while
                                  in $e$, then all processes in $M$ eventually deliver $\vid$;
                              \item \label{prop:liveness-c}
                                if a process delivers %
                                $\vid$, then all processes in $M$ eventually deliver $\vid$. 
			\end{enumerate}
                      \end{enumerate}
	\caption{Properties of reconfigurable atomic broadcast over a history $h$.}
	\label{fig:rpob-properties}
\end{figure}

\emph{Reconfigurable atomic broadcast} is defined by the properties over
histories $h$ listed in
\figref{fig:rpob-properties}. Properties~\ref{prop:basicconf}
and~\ref{prop:integrity} are self-explanatory. Property~\ref{prop:total} ensures
that processes cannot deliver messages in contradictory orders.
Property~\ref{prop:agreement} disallows executions where sequences of messages
delivered at different processes diverge. 

The liveness requirements of reconfigurable atomic broadcast are given by
Property~\ref{prop:liveness}. Property~\ref{prop:liveness}a asserts a
termination guarantee for reconfiguration requests. As shown by Spiegelman and
Keidar~\cite{on-liveness}, wait-free termination is impossible to support even
for reconfigurable read/write registers, which are weaker than atomic
broadcast. Hence, the guarantee given by Property~\ref{prop:liveness}a is
similar to obstruction-freedom~\cite{obstruct-free}. %
Let us say that a configuration $\confid$ is \emph{activated} when all its members get
$\confchng(C)$ notifications.  Property~\ref{prop:liveness}a asserts that, in a
run with finitely many reconfigurations, the last reconfiguration request
invoked by a correct process and executing in isolation must eventually succeed
to introduce a configuration $C$, which must then become activated if all its
members are correct.
Properties~\ref{prop:liveness}b-c state liveness guarantees
for the configuration $C$ similar to those of the classical atomic broadcast.
Property~\ref{prop:liveness}b asserts that any message broadcast by a (correct)
member of $C$ eventually gets delivered to all members of $C$. Property~\ref
{prop:liveness}c additionally ensures that the members of $C$ eventually deliver
all messages delivered by any process in any configuration.

As in prior work~\cite{dynastore,ken-book,spiegelman2017dynamic},
the liveness of our protocols is premised on the following assumption, which limits the power of the environment 
to crash configuration members.
\begin{assumption}[Availability]
\label{asm:liveness-condition}
Let $\confid = \aconf[\eid][\vm][\_]$ be an introduced configuration, i.e., such
that $\intro(\confid) \in \hid$. Then at least one member of $\vm$ does not crash
before another configuration $\confid' = \aconf[\eid'][\_][\_]$ with $\eid'> \eid$ is activated.
\end{assumption}
Our protocols use the period of time when some member of $\vm$ is guaranteed not
to crash to copy its state to the members of a new configuration.

Finally, we note that in the case of a single static configuration, our
specification in \figref{fig:rpob-properties} corresponds to the classical
notion of atomic broadcast~\cite{to-survey}.

%% file: protocol.tex
\section{The Vertical Atomic Broadcast Protocol}
\label{sec:vp_protocol}

In Figures~\ref{vp_failure_free} and \ref{vp_reconfiguration} we present a
protocol implementing the specification of \S\ref{sec:spec}, which we call
\emph{Vertical Atomic Broadcast (VAB)} by analogy with Vertical 
Paxos~\cite{vertical-paxos}. For now the reader should ignore the code in blue.
At any given time, a process executing the protocol
participates in a single configuration, whose epoch is stored in a variable
$\epoch$. The membership of the configuration is stored in a variable
$\members$. Every member of a given configuration is either the leader or a
\emph{follower}. A $\status$ variable at a process records whether it is a
\LEADER, a \FOLLOWER, or is in a special \FRESH\ state used for new processes. A
$\leadervar$ variable stores the leader of the current configuration. We assume
that the system starts in an initial active configuration with epoch $0$.

\begin{figure}[t]
	\begin{tabular}{@{}l@{\!\!\!\!\!\!}|@{\qquad}l@{}}
        \scalebox{0.96}{%
	\begin{minipage}[t]{7cm}
	\vspace*{-12pt}
	\begin{algorithm}[H]
		\DontPrintSemicolon
		\setcounter{AlgoLine}{0}
		$\assign{\epoch}{0} \in \mathbb{Z}$\;
		$\assign{\newepoch}{0} \in \mathbb{Z}$\;
		$\assign{\pnext}{0} \in \mathbb{Z}$\;
		$\assign{\initlength}{-1} \in \mathbb{Z}$\;
		$\assign{\lastdelivered}{-1} \in \mathbb{Z}$\;
		$\members \in 2^\mathcal{P}$\;
		${\leadervar}\in \mathcal{P}$\;
		$\msg[\,] \in \mathbb{N} \rightarrow \msgset \cup \{\bot\}$\;
		$\status \in \{\LEADER, \FOLLOWER, \FRESH \}$\;
		\smallskip

		\Function{${\bcast}(m)$\label{vp_broadcast}}{
			\send $\FORWARD(\vid)$ \To $\leadervar$\; \label{vp_broadcast_body}
			
		}
	
			\smallskip

        \onreceive ${\FORWARD}(\vid)$ \from $\pid[j]$\label{vp_forward}\\
	\SubAlgo{{\color{blue} // \textbf{function ${\bcast}(m)$:\label{po_broadcast}}}}{
		\precond $\pid = \leadervar$\;\label{vp_forward_pre}
		$\assign{\msg[\pnext]}{m}$\; \label{vp_broadcast_assign}
		\send $\ACCEPT(\epoch, \pnext, m)$ \qquad\qquad \To $\members \setminus \{p_i\}$\;\label{vp_send_accept}
		$\assign{\pnext}{\pnext + 1}$\;\label{vp_broadcast_next}
		
	}

	\smallskip

	\end{algorithm}
	\vspace*{-5pt}
\end{minipage}}
&
        \scalebox{0.96}{%
	\begin{minipage}[t]{7.5cm}
	\vspace*{-12pt}
	\begin{algorithm}[H]
		\DontPrintSemicolon

		\SubAlgo{\onreceive ${\ACCEPT}(e, k, m)$ \from $p_j$\label{vp_accept}}{
			\precond $\status = \FOLLOWER \wedge \epoch = e$\;\label{vp_accept_pre}
			$\assign{\msg[k]}{m}$\; \label{vp_accept_assignment}
			\send $\ACCEPTACK(e, k)$ \To $p_j$\;\label{vp_send_accept_ack}
		}
		
                \smallskip
		
		\SubAlgo{\onreceive $\ACCEPTACK(e, k)$ \qquad\qquad \textbf{from all}
			$\members \setminus \{p_i\}$\label{vp_quorum_of_accept_acks}}{
			\precond $\status = \LEADER \wedge \epoch = e$\; \label{vp_pre_quorum_of_accept_acks}
			\send $\COMMIT(e, k)$ \To $\members$\;\label{vp_send_commit}
		}
		
		\smallskip
		
		\SubAlgo{\onreceive $\COMMIT(e, k)$\label{vp_commit}}{
			\precond $\status \in \{\LEADER, \FOLLOWER\} \wedge \ \newline{\;\;\;\;\;\;\,} \epoch = e \wedge k = \lastdelivered + 1$\;\label{vp_commit_pre}
			$\assign{\lastdelivered}{k}$\;\label{vp_commit_last_delivered}
			$\deliv(\msg[k])$\;\label{vp_deliver}
		}	
\end{algorithm}
\vspace*{-5pt}
\end{minipage}}
\end{tabular}
	\caption{Vertical Atomic Broadcast at a process $p_i$: normal operation.}
	\label{vp_failure_free}
\end{figure}

\subparagraph{Normal operation.} When a process receives a call
${\tt broadcast}(\vid)$, it forwards $\vid$ to the leader of its current
configuration (line \ref{vp_broadcast}). Upon receiving $\vid$ (line
\ref{vp_forward}), the leader adds it to an array $\msg$; a $\pnext$ variable
points to the first free slot in the array (initially $0$). The leader then
sends $m$ to the followers in an $\ACCEPT(\eid, k, \vid)$ message, which carries
the leader's epoch $\eid$, the position $k$ of $m$ in the $\msg$ array, and the
message $m$ itself.

A process acts on the $\ACCEPT$ message (line \ref{vp_accept_pre}) only if it
participates in the corresponding epoch. It stores $m$ in its local copy of the
$\msg$ array and sends an $\ACCEPTACK(\eid, k)$ message to the leader of
$\eid$.
The application message at position $k$ is \emph{committed} if
the leader of $\eid$ receives $\ACCEPTACK$ messages for epoch
$\eid$ and position $k$ from all followers of its configuration (line
\ref{vp_quorum_of_accept_acks}). 
In this case the leader notifies all
the members of its configuration that the application message can be safely
delivered via a $\COMMIT$ message. A process delivers application messages in
the order in which they appear in its $\msg$ array, with $\lastdelivered$
storing the last delivered position (line \ref{vp_commit}).

\begin{figure}[t]
	\begin{tabular}{@{}l@{\!\!\!\!}|@{\ \ }l@{}}
        \scalebox{0.96}{%
		\begin{minipage}[t]{8.25cm}
			\vspace*{-12pt}
			\begin{algorithm*}[H]
		\DontPrintSemicolon
				\Function{{$\reconf$}{$()$}\label{vp_reconfigure}}{
					\textbf{var} $\pepoch, \pmembers, \repoch, \rpmembers$\;
					$\assign{\pepoch}{\tt get\_last\_epoch}()\ \textbf{at CS}$\;\label{vp_reconfigure_get_last}
					$\assign{\repoch}{\eid+1}$\; \label{vp_reconfigure_inc}
					\Repeat{\textbf{\upshape received}
                                          $\PROBEACK(${\upshape
                                            \true}$,\repoch)$ \qquad \textbf{\upshape from some} $\pid[j]$\label{vp_probe_ack_true}}{
						\If{$\eid \ge 0$}{
							$\assign{\pmembers}{{\tt get\_members}(\pepoch)}\ \textbf{at CS}$\;
							\send $\PROBE(\repoch, \pepoch)$ \To $\pmembers$\;\label{vp_send_probe}
							\textbf{wait until
                                                          received}
                                                        $\phantom{~~}\PROBEACK(\_, \repoch)$
                                                        \textbf{\qquad\qquad\qquad
                                                          \phantom{\hspace{6pt}}from a process in $\pmembers$}\label{vp_probe_ack_false}\;
							$\assign{\pepoch}{\pepoch - 1}$\;
						}
					}
					$\assign{\rpmembers}{\tt compute\_membership}()$\;\label{vp_compute_membership}
					\uIf{${\tt compare\_and\_swap}(\repoch {-} 1,\langle \repoch,
						\rpmembers, \pid[j] \rangle)$ \qquad
                                                \textbf{\upshape at CS}  {\tt /*}
                                                {\intro}$(\langle\repoch,\rpmembers, \pid[j]
                                                \rangle)$ {\tt */}\label{vp_compare_and_swap}}{
						\send $\NEWCONFIG(\repoch, \rpmembers)$ \To $\pid[j]$\;\label{vp_send_new_config}
						\textbf{return} $\langle \repoch, \rpmembers, \pid[j] \rangle$\;
					}
					\Else{\textbf{return} $\bot$}
				}
			
		\smallskip
                                
			\SubAlgo{\onreceive $\PROBE(\repoch, \eid)$ \from $\pid[j]$\label{vp_probe}}{
				\precond $\repoch \geq
				\newepoch$\;\label{vp_probe_pre}
				$\assign{\newepoch}{\repoch}$\; \label{vp_probe_set_newepoch}
				\uIf{$\epoch\geq \eid$}{\send $\PROBEACK(\true, \repoch)$ \To
					$\pid[j]$}\label{vp_send_probe_ack}
				\Else{\send $\PROBEACK(\false, \repoch)$ \To
					$\pid[j]$}\label{vp_send_probe_ack_false}
			}
			\end{algorithm*}
			\vspace*{-6pt}
                      \end{minipage}} &
        \scalebox{0.96}{%
		\begin{minipage}[t]{6.5cm}
			\vspace*{-12pt}
			\begin{algorithm*}[H]
                          		\DontPrintSemicolon
                        \SubAlgo{\onreceive $\NEWCONFIG(\eid, \vm)$ \from $p_j$\label{vp_new_config}}{
					\precond $\newepoch = e$\;\label{vp_new_config_pre}
					$\assign{\status}{\LEADER}$\;\label{vp_quorum_of_new_state_acks_status}
					$\assign{\epoch}{\eid}$\; \label{vp_new_config_epoch}
					$\assign{\members}{\vm}$\;
					$\assign{\leadervar}{\pid}$\; \label{vp_setleader}
					$\assign{\pnext}{\max\{k \mid \msg[k] \neq \bot\}+1}$\;\label{vp_quorum_of_new_state_acks_next_opt}
					$\assign{\initlength}{\pnext-1}$\;
					$\confchng(\eid,\vm, \pid)$\;\label{vp_ready_confchg}
					{\color{blue}{// $\confchng(\eid,\vm, \pid,$
                                            \phantom{\quad\ \ \ }$\msg[\lastdelivered{+}1..\initlength])$}}\label{po_newconfig_sdeliv}\\
					\send $\NEWSTATE(\eid, \msg, \vm)$
                                        \qquad \To $\members \setminus \{p_i\}$\;\label{vp_send_new_state}
				}

		\smallskip
				
				\SubAlgo{\onreceive $\NEWSTATE(e,
                                  \vmsg, \vm)$ \from $p_j$\label{vp_new_state}}{
					\precond $\newepoch \le \eid$\;\label{vp_new_state_pre}
					$\assign{\status}{\FOLLOWER}$\;
					$\assign{\epoch}{e}$\;
					$\assign{\newepoch}{e}$\;
					$\assign{\msg}{\vmsg}$\; \label{vp_newstate_msg_assignment}
					$\assign{\leadervar}{\pid[j]}$\;
					$\confchng(\eid,\vm, \pid[j])$\;\label{vp_new_state_confchg}
					{\color{blue}{// $\confchng(\eid,\vm, \pid[j], \bot)$}}\;\label{vp_new_state_confchg2}
					\send $\NEWSTATEACK(e)$ \To $p_j$\;\label{vp_send_new_state_ack}
				}
				
		\smallskip
				
				\SubAlgo{\onreceive $\NEWSTATEACK(e)$ \from{}
                                  \textbf{all}
                                  $\members \setminus \{p_i\}$\label{vp_quorum_of_new_state_acks}}{
					\precond $\newepoch = \epoch = e$\; 
					\lFor{$k = 1..\initlength$}{\qquad \qquad \textbf{send $\texttt{COMMIT}(e, k)$ \To $\members$}}\label{vp_send_commit_2}						
				}
			\end{algorithm*}	
			\vspace*{-6pt}
		\end{minipage}}
	\end{tabular}
	\caption{Vertical Atomic Broadcast at a process $p_i$: reconfiguration}
	\label{vp_reconfiguration}
\end{figure}

\subparagraph{Reconfiguration: probing.} Any process can initiate a
reconfiguration, e.g., to add new processes or to replace failed ones.
Reconfiguration aims to preserve the following invariant, key to proving the
protocol correctness.
\begin{invariant}
\label{inv:prefix-of-higher-epoch}
Assume that the leader of an epoch $\eid$ sends $\COMMIT(\eid, k)$ while having
$\msg[k] = \vid$. Whenever any process $\pid$ has $\epoch = \eid' > \eid$, it
also has $\msg[k] = \vid$.
\end{invariant}
The invariant ensures that any application message committed in an epoch $\eid$
will persist at the same position in all future epochs $\eid'$. This is used to
establish that the protocol delivers application messages in the same order at
all processes.

To ensure \invref{inv:prefix-of-higher-epoch}, a process performing a
reconfiguration first \textit{probes} the previous configurations to find
a process
whose state contains all messages that could have been committed in
previous epochs, which will serve as the new leader.
The new leader then
transfers its state to the followers of the new configuration. We say that a
process is \emph{initialized} at an epoch $\eid$ when it completes the state
transfer from the leader of $\eid$; it is at this moment that the process
assigns $\eid$ to its $\epoch$ variable, used to guard the transitions at
lines~\ref{vp_accept},~\ref{vp_quorum_of_accept_acks}, \ref{vp_commit}.
Our protocol guarantees that a configuration with epoch $\eid$ can become activated
only after all its members have been initialized at $\eid$. Probing is
complicated by the fact that there may be a series of failed reconfiguration
attempts, where the new leader fails before initializing all its followers. For
this reason, probing may require traversing epochs from the current one down,
skipping epochs that have not been activated.

In more detail, a process $p_r$ initiates a reconfiguration by calling
${\tt reconfigure}$ (line~\ref{vp_reconfigure}). The process picks an epoch
number $\repoch$ higher than the current epoch stored in the configuration
service and then starts the probing phase. The process $p_r$ keeps track of the
epoch being probed in $\pepoch$ and the membership of this epoch in
$\pmembers$. The process initializes these variables when it obtains the
information about the current epoch from the configuration service.
To probe an epoch $e$, the process sends a
$\PROBE(\repoch, \pepoch)$ message to the members of its configuration, asking
them to join the new epoch $\repoch$ (line \ref{vp_send_probe}). Upon receiving
this message (line \ref{vp_probe}), a process first checks that the proposed
epoch $\repoch$ is $\ge$ the highest epoch it has ever been asked
to join, which is stored in $\newepoch$ (we always have $\epoch \le
\newepoch$). In this case, the process sets $\newepoch$ to $\repoch$. Then, if
the process was initialized at an epoch $\ge$ the epoch $\pepoch$ being probed,
\ag{We may want to explain this if there's space.}  it replies
with $\PROBEACK(\true, \repoch)$; otherwise, it replies with
$\PROBEACK(\false, \repoch)$.

If $p_r$ receives at least one $\PROBEACK(\false, \repoch)$ from a member of
$\pepoch$ (line \ref{vp_probe_ack_false}), $p_r$ can conclude that $\pepoch$ has
not been activated, since one of its processes was not initialized by the leader
of this epoch.  The process $p_r$ can also be sure that $\pepoch$ will never
become activated,
since it has switched at least one of its members to the new
epoch. In this case, $p_r$ starts probing the preceding epoch $\pepoch-1$. Since
no application message could have been committed in $\pepoch$, picking a new leader
from an earlier epoch will not lose any committed messages and thus will not
violate \invref{inv:prefix-of-higher-epoch}.  If $p_r$ receives some
$\PROBEACK(\true, \repoch)$ messages, then it ends probing: any process $p_j$
that replied in this way can be selected as the new leader (in particular, $p_r$
is free to maintain the old leader if this is one of the processes that
replied).

\subparagraph{Reconfiguration: initialization.}
Once the probing finds a new leader $p_j$ (line~\ref{vp_probe_ack_true}), the
process $p_r$ computes the membership of the new configuration using a function
{\tt compute\_membership} (line~\ref{vp_compute_membership}). We do not
prescribe any particular implementation for this function, except that the new
membership must contain the new leader $p_j$. In practice, the function would
take into account the desired changes to be made by the reconfiguration. Once
the new configuration is computed, $p_r$ attempts to store it in the
configuration service using a {\tt compare\_and\_swap} operation. This succeeds
if and only if the current epoch in the configuration service is still the epoch
from which $p_r$ started probing, which implies that no concurrent
reconfiguration occurred during probing. In this case $p_r$ sends a $\NEWCONFIG$
message with the new configuration to the new leader and returns the new
configuration to the caller of {\tt reconfigure}; otherwise, it returns $\bot$.
A successful {\tt compare\_and\_swap} also generates an $\intro_r$ action for
the new configuration, which is used in the broadcast specification (\S\ref{sec:spec}).

When the new leader receives the $\NEWCONFIG$ message (line
\ref{vp_new_config}), it sets $\status$ to $\LEADER$, $\epoch$ to the new epoch,
and stores the information about the new configuration in $\members$ and
$\leadervar$. The leader also sets $\pnext$ to the first free slot in the $\msg$
array and saves its initial length in a variable $\initlength$. The leader then
invokes $\confchng$ for the new configuration. In order to finish the
reconfiguration, the leader needs to transfer its state to the other members of
the configuration. To this end, the leader sends a $\NEWSTATE$ message to them,
which contains the new epoch and a copy of its $\msg$ array (line
\ref{vp_send_new_state}; a practical implementation would optimize this by
sending to each process only the state it is missing). Upon receiving a
$\NEWSTATE$ message (line \ref{vp_new_state}), a process overwrites its $\msg$
array with the one provided by the leader, sets its $\status$ to $\FOLLOWER$,
$\epoch$ to the new epoch, and $\leadervar$ to the new leader. The process also
invokes $\confchng$ for the new configuration.  It then acknowledges its
initialization to the leader with an $\NEWSTATEACK$ message.
Upon receiving $\NEWSTATEACK$ messages from all
followers (line \ref{vp_quorum_of_new_state_acks}), the new leader sends
$\COMMIT$s for all application messages from the previous epoch, delimited by
$\initlength$.  These messages can be safely delivered, since they are now
stored by all members of epoch $e$.

\subparagraph{Example.} \figref{fig:message-flow-reconf} gives an example
illustrating the message flow of reconfiguration.  Assume that the initial
configuration ${\tt 1}$ consists of processes $\pid[1]$, $\pid[2]$ and $\pid[3]$.
Following a failure of $\pid[3]$, a process $\pid[r]$ initiates reconfiguration
to move the system to a new configuration ${\tt 2}$. 
To this end, $\pid[r]$ sends $\PROBE({\tt 2, 1})$ to the
members of configuration ${\tt 1}$. Both processes $\pid[1]$ and $\pid[2]$
respond to $\pid[r]$ with $\PROBEACK(\true, {\tt 2})$. 
The process $\pid[r]$ computes the membership of the new configuration,
replacing $\pid[3]$ by a fresh process $\pid[4]$, and stores the new
configuration in the configuration service, with $\pid[2]$ as the new
leader. Next, $\pid[r]$ sends a $\NEWCONFIG$ message to
$\pid[2]$.

\begin{figure*}[t!]
  \renewcommand{\baselinestretch}{0.8}
  \resizebox{1.02\textwidth}{!}{%
	\centering
	\begin{tikzpicture}[decoration={markings,mark=at position 1 with
		{\arrow[scale=1.5,>=stealth]{>}}},postaction={decorate}]
	
	\draw[thick, dashed] (4,0) --(18.5,0);
	\draw[thick, dashed] (0,1) --(0.3,1);
	\draw[thick, dashed] (0,2) --(5.1,2);
	\draw[thick, dashed] (0,3) --(18.5,3);

	\draw[thick, dashed] (0,4) --(2.6,4);
	\draw[thick, dashed] (4.3,4) --(12.1,4);
	\draw[thick, dashed] (13.6,4) --(18.5,4);
	\draw[thick, dashed] (11,-1) --(18.5,-1);

	\node [left] at (0,4.1) {$\pid[r]$};
	\node [left] at (0,3.1) {$\pid[1]$};
	\node [left] at (0,2.1) {$\pid[2]$};
	\node [left] at (0,1.1) {$\pid[3]$};
	\node [left] at (4,0.1) {$\pid[4]$};
	\node [left] at (11,-1) {$\pid[5]$};
	
	\node (CC3) [text width=2.5cm] at (1.8,1.5) {\footnotesize{probe the current configuration}};	
	
	\node (CC1) [text width=1.8cm] at (3.58,4) {\footnotesize{compute membership}};

	\node (CC7) [text width=2.5cm] at (6.8,-.47) {\footnotesize{probe the current configuration}};

	\node (CC8) [text width=1.6cm] at (4.9,1.15) {\footnotesize{notify the new leader and move to {\tt 2}}};
	
    \node (CS) [text width=50mm] at (8.5,6.05) {\footnotesize{\ \ initially: $\langle 1, \{\pid[1], \pid[2], \pid[3]\}, \pid[3] \rangle$}};

    \node (CS) [
    draw,thick,rounded corners,fill=yellow!20,inner sep=.2cm, text width=36.5mm]
    at (8.1,5.3) {\textbf{~~~~~Configuration \\~~~~~~~~~~~service}};

	\node [left] (s1c) at (0.7,3.9)  {};
	\node [left] (s2c) at (4,2.9)  {};
	\node [left] (s3c) at (5.8,2.9)  {};

	\draw[postaction={decorate}, black] (0.4,3.9) -- (.6,3) node[sloped,above] {};
	
	\draw[postaction={decorate}, black] (0.5,3.9) -- (.95,2) node[pos=0.43,sloped,above, fill=white,inner sep=0] {\footnotesize	$\PROBE(\tt{2, 1})$};
	configuration
	\draw[postaction={decorate}, black] (1.35,2) -- (2.45,4)  node[pos=0.47,sloped,above, fill=white, inner sep=0] {\footnotesize	$\PROBEACK({\tt T, 2})$};
	\draw[postaction={decorate}, black] (2.03,3) -- (2.58,4) node[sloped,above] {};
	\draw[postaction={decorate}, black] (4.4,4) -- (4.8,2) node[pos=0.49,sloped,above,fill=white] {\footnotesize	$\NEWCONFIG$};		
	
	\draw[postaction={decorate}, black] (5,4) -- (5.3,3) node[pos=0.5,sloped,above] {};
	\draw[postaction={decorate}, black] (5.1,4) -- (6.3,0) node[pos=0.25,sloped,above,fill=white, inner sep=0] {\footnotesize	$\PROBE(\tt{3, 2})$};;
	
	\node (CC6) [text width=3cm] at (9,2.49) {\footnotesize{probe the preceding configuration}};	
			
	\draw[postaction={decorate}, black] (7,4) -- (8.5,3) node[pos=0.75,sloped,above] {\footnotesize	$\PROBE(\tt{3, 1})$};;
	
	\draw[postaction={decorate}, black] (6.5,0) -- (6.8,4)  node[pos=0.45,sloped,above] {\footnotesize	$\PROBEACK({\tt F, 3})$};
	
	\draw[postaction={decorate}, black] (9.3,3) -- (12,4) node[pos=0.32,sloped,above,fill=white, inner sep=0] {\footnotesize	$\PROBEACK({\tt T, 3})$};

	\node (CC2) [text width=1.8cm] at (12.93,4) {\footnotesize{compute membership}};
	\node (CC4) [text width=1.6cm] at (14.7,2.15) {\footnotesize{notify the new leader and move to {\tt 3}}};
	\node (CC5) [text width=1.8cm] at (16.6,-1.65) {\footnotesize{transfer state and move to {\tt 3}}};
	
	\node (CC6) [text width=1.6cm] at (18.6,-1.52) {\footnotesize{deliver messages}};
		
	\draw[postaction={decorate}, black] (13.6,4) -- (15.2,3) node[pos=0.7,sloped,above] {\footnotesize	$\NEWCONFIG$};

	\draw[postaction={decorate}, black] (15.5,3) -- (16.1,0) node[pos=0.65,sloped,above] {};
	\draw[postaction={decorate}, black] (15.6,3) -- (16.4,-1) node[pos=0.25,sloped,above]{\footnotesize	$\NEWSTATE$};
	
	\node [left] (s1c) at (6.6,1.9)  {};
	
	\node (W1) [above] at (0.35,0.76)
	{{$\Cross$}};
	
	\node (W2) [above] at (5.1,1.76)
	{{$\Cross$}};
	
	\draw[postaction={decorate}, black, bend left=20] (CC1) to[left]  node[pos=0.3, above=.4cm] {\footnotesize store($\langle${\tt 2}, \{$\pid[1], \pid[2], \pid[4]$\}, $\pid[2] \rangle$)$\qquad$} (CS.west);
	
	\draw[postaction={decorate}, black, bend right=20] (CC2) to[left]  node[pos=0.3, above=.39cm] {\footnotesize $\quad$ store($\langle${\tt 3}, \{$\pid[1], \pid[4], \pid[5]$\}, $\pid[1] \rangle$)} (CS.east);
	
	\draw[postaction={decorate}, black] (16.6,-1) -- (17,3) node[pos=0.6,sloped,above]{\footnotesize $\NEWSTATEACK$};
	
	\draw[postaction={decorate}, black] (16.85,0) -- (17.15,3) node[pos=0.6,sloped,above] {};

	\draw[postaction={decorate}, black] (17.6,3) -- (17.97,0) node[pos=0.3,sloped,above] {};
	\draw[postaction={decorate}, black] (17.7,3) -- (18.2,-1) node[pos=0.3,sloped,above]{\footnotesize	$\COMMIT({\tt 3,1})$};
	
	\end{tikzpicture}
	}
\caption{The behavior of the protocol during reconfiguration.}\label{fig:message-flow-reconf}
\end{figure*}
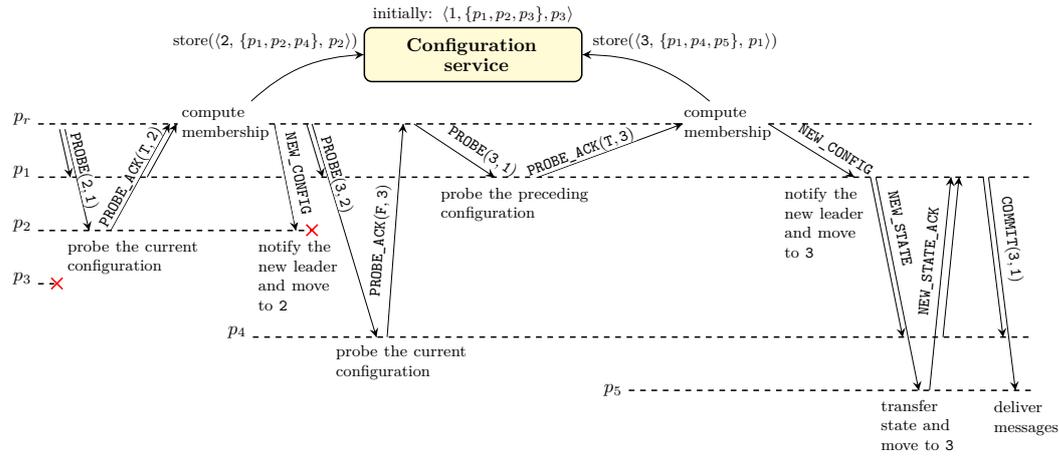

Assume that after receiving this message $\pid[2]$ fails, prompting $\pid[r]$ to
initiate yet another reconfiguration to move the system to a configuration
${\tt 3}$.
To this end,
$\pid[r]$ sends $\PROBE({\tt 3, 2})$ to the members of configuration ${\tt 2}$,
and $\pid[4]$ responds with $\PROBEACK(\false, {\tt 3})$. The process $\pid[r]$
concludes that epoch ${\tt 2}$ has not been activated and starts probing the
preceding epoch ${\tt 1}$: it sends $\PROBE({\tt 3, 1})$ and gets a
reply $\PROBEACK(\true, {\tt 3})$ from $\pid[1]$, which is selected as the new
leader. The process $\pid[r]$ computes the new set of members, replacing
$\pid[3]$ by a fresh process $\pid[5]$, stores the new configuration in the
configuration service, and sends a $\NEWCONFIG$ message to 
the new leader $\pid[1]$. This process invokes the $\confchng$ upcall for the new configuration
and sends its state to the followers in a $\NEWSTATE$ message.  The followers
store the state, invoke $\confchng$ upcalls and reply with $\NEWSTATEACK$s. Upon
receiving these, $\pid[1]$ sends $\COMMIT$s for all application messages in its state.

\subparagraph{Steady-state latency and reconfiguration downtime.}  A
configuration is \emph{functional}\/ if it was activated and all its members are
correct.  A configuration is \emph{stable}\/ if it is functional and no
configuration with a higher epoch is introduced. The \emph{steady-state
  latency}\/ is the maximum number of message delays it takes from the moment
the leader $p_i$ of a stable configuration receives a broadcast request for a
message $m$ and until $m$ is delivered by $p_i$. It is easy to see that
our protocol has the steady-state latency of $2$ (assuming self-addressed messages are received 
instantaneously), which is optimal~\cite{lower-bound}.

The system may be reconfigured not only in response to a failure, but also to
make changes to a functional configuration: e.g., to move replicas from highly
loaded machines to lightly loaded ones, or to change the number of machines
replicating the service~\cite{smart,kafka-book,matchmaker}. As modern online
services have stringent availability requirements, it is important to minimize
the period of time when a service is unavailable due to an ongoing
reconfiguration.
More precisely, suppose the system is being reconfigured from a functional
configuration $C$ to a stable configuration $C'$.  The reconfiguration
\emph{downtime}\/ is the maximum number of message delays it takes from the
moment $C$ is disabled and until the leader of $C'$ is ready to broadcast
application messages in the new configuration.

As we argue in \S\ref{sec:rw}, existing vertical solutions for atomic broadcast
stop the system as the first step of reconfiguration~\cite{ken-book}, resulting
in the reconfiguration downtime of at least $4$~($2$ message delays to disable
the latest functional configuration plus at least $2$ message delays to reach
consensus on the next configuration and propagate the decision).  In contrast,
our protocol achieves the downtime of $0$ by keeping the latest functional
configuration active while the probing of past configurations and agreement on a
new one is in progress.

\begin{theorem}\label{thm:zero-downtime}
  The VAB protocol reconfigures a functional configuration with $0$ downtime.
\end{theorem}
\begin{proof}
Suppose that the current configuration $C$ with an epoch $\eid$
is functional. Note that the normal path of our protocol is guarded by
preconditions $\epoch = e$, so that $C$ can broadcast and deliver application messages as
long as this holds at all its members (lines~\ref{vp_accept_pre},
\ref{vp_pre_quorum_of_accept_acks} and \ref{vp_commit_pre}). Assume now that a
process $p_r$ starts reconfiguring the system to a new configuration $C'$ with
epoch $\eid+1$. The process $p_r$ will send $\PROBE$ messages to the members of
$C$ and, since $C$ is functional, $p_r$ will only get replies
$\PROBEACK(\true, \eid+1)$. Handling a $\PROBE$ message only modifies the
$\newepoch$ variable, not $\epoch$. Therefore, $C$ can continue processing
broadcasts while $p_r$ is probing its members, storing $C'$ in the configuration
service, and sending $\NEWCONFIG(\eid+1,\_)$ to the leader $p_i$ of $C'$. When
the new leader $p_i$ handles $\NEWCONFIG(\eid+1,\_)$, it will set
$\epoch=\eid+1$, disabling the old configuration. However, the leader will
at once be ready to broadcast messages in the new configuration,
as required.
\end{proof}

\subparagraph{Correctness.} Our protocol achieves the above $0$-downtime
guarantee without violating correctness. Informally, this is because it always
chooses the leader of the new configuration from among the members of the latest
activated configuration, and a message can only be delivered in this
configuration after having been replicated to all its members. Hence, the new
leader will immediately know about all previously delivered messages, including
those delivered during preliminary reconfiguration steps. The following theorem
(proved in \tr{\ref{sec:correctness}}{\ncorr}) states the correctness of our
protocol.
\begin{theorem}\label{thm:correctness}
  The VAB protocol correctly implements reconfigurable atomic broadcast as
  defined in Figure~\ref{fig:rpob-properties}.
\end{theorem}

%% file: po.tex
\section{Passive Replication}
\label{sec:po}

The protocol presented in the previous section can be used to build
reconfigurable fault-tolerant services via \emph{active} (aka state-machine)
replication~\cite{smr}. Here a service is defined by a deterministic state
machine and is executed on several replicas, each maintaining a copy of the
machine. All replicas execute all client commands, which they receive via atomic
broadcast. Together with the state machine's determinism, this ensures that each
command yields the same result at all replicas, thus maintaining an illusion of
a centralized fault-tolerant service.

In the rest of the paper, we focus on an alternative approach of building
reconfigurable fault-tolerant services via \emph{passive} (aka primary-backup)
replication~\cite{budhiraja1993primary}. Here commands are only executed by the
leader, which propagates the state changes induced by the commands to the other
replicas. This allows replicating services with non-deterministic operations,
e.g., those depending on timeouts or interrupts.

Formally, we consider services with a set of states $\State$ that accept a set
of commands $\Command$. A command $\cmd \in \Command$ can be executed using a
call $\execute(\cmd)$, which produces its return value. Command execution may be
non-deterministic. To deal with this, the effect of executing a command $\cmd$
on a state $\Sigma \in \State$ is defined by transition relation
$\Sigma \stackrel{\cmd}{\rightarrow} \langle r, \delta \rangle$, which produces
a possible return value $r$ of $\cmd$ and a \emph{state update} $\delta$
performed by the command. The latter can be applied to any state $\Sigma'$ using
a function $\app(\Sigma', \delta)$, which produces a new state. For example, a
command
$\textbf{if}\ x \,{=}\, 0\ \textbf{then}\ y \,{\leftarrow}\, 1\ \textbf{else}\ y
\,{\leftarrow}\, 0$ produces a state update $y \,{\leftarrow}\, 1$ when executed
in a state with $x \,{=}\, 0$. A command assigning $x$ to a random number may
produce an update $x \,{\leftarrow}\, 42$ if the random generator returned
$42$ when the leader executed the command.

We would like to implement a service over a set of fault-prone replicas that is
linearizable~\cite{linearizability} with respect to a service that atomically
executes commands on a single non-failing copy of the state machine. The latter
applies each state update to the machine state $\Sigma$ immediately
after generating it, as shown in Figure~\ref{fig:passive-spec}. Informally, this
means that commands appear to clients as if produced by a single copy of the
state machine in Figure~\ref{fig:passive-spec} in an order consistent with the
\emph{real-time order}, i.e., the order of non-overlapping command invocations.

\subsection{Passive Replication vs Atomic Broadcast}
\label{sec:po-bcast}

As observed in~\cite{ken-book, zookeeper, junqueira2013barriers}, implementing
passive replication requires propagating updates from the leader to the
followers using a stronger primitive than atomic broadcast. To illustrate why,
Figure~\ref{fig:app-layer} gives an incorrect attempt to simulate the
specification in Figure~\ref{fig:passive-spec} using our reconfigurable atomic
broadcast %
(ignore the code in blue for now). This attempt serves as a strawman for a
correct solution we present later.
Each process keeps track of the epoch it belongs to in $\curepoch$ and the
leader of this epoch in $\curleader$. To execute a command
(line~\ref{passive:client-execute}), a process sends the command, tagged by a
unique identifier, to the leader. It then waits until it hears back about the
result.

A process keeps two copies of the service state -- a \emph{committed} state
$\pstate$ and a \emph{speculative} state $\tstate$; the latter is only used when
the process is the leader.  When the leader receives a command $\cmd$
(line~\ref{passive:accept}), it executes $\cmd$ on its speculative state
$\tstate$, producing a return value $r$ and a state update $\delta$. The leader
immediately applies $\delta$ to $\tstate$ and
distributes the triple of the command identifier, its return value and the state
update via atomic broadcast. When a process (including the leader) delivers such
a triple (line~\ref{passive:deliv}), it applies the update to its committed
state $\pstate$ and sends the return value to the process the command originated
at, determined from the command identifier.  When a process receives a
$\confchng$ upcall (line~\ref{passive:change}), it stores the information
received in $\curepoch$ and $\curleader$. If the process is the leader of the
new epoch, it also initializes its speculative state $\tstate$ to the committed
state $\pstate$.

In passive replication, a state update is incremental with respect to the state
it was generated in. Thus, to simulate the specification in
Figure~\ref{fig:passive-spec},
it is crucial that the committed state $\Sigma$ at a process delivering a state
update (line~\ref{passive:apply}) be the same as the speculative state $\Theta$
from which this state update was originally derived
(line~\ref{passive:exec}). This is captured by the following invariant. Let
$\Sigma_i(k)$ denote the value of $\Sigma$ at process $p_i$ before the $k$-th
action in the history (and similarly for $\Theta$).
\begin{invariant}
\label{inv:main-po}
Let $h$ be a history of the algorithm in Figure~\ref{fig:app-layer}.
If $h_k=\deliv_i(m)$, then
there exist $j$ and $l<k$ such that $h_l = \bcast_j(m)$ and
$\Sigma_i(k) = \Theta_j(l)$.
\end{invariant}

Unfortunately, if we use atomic broadcast to disseminate state updates in
Figure~\ref{fig:app-layer}, we may violate Invariant~\ref{inv:main-po}. We next
present two examples showing how this happens and how this leads to violating
linearizability. The examples consider a replicated counter 
$x$ with two commands -- an increment ($x\leftarrow x + 1$) and a read
($\return\ x$). Initially $x=0$, and then two clients execute two increments.

\subparagraph{Example 1.}
The two increments are executed by the same leader. The first one generates an
update $\delta_1=(x \leftarrow 1)$ and a speculative state $\Theta = 1$. Then
the second generates %
$\delta_2=(x \leftarrow 2)$. Atomic broadcast allows processes to deliver the
updates in the reverse order, with $\delta_1$ applied to a committed state
$\Sigma = 2$. This violates Invariant~\ref{inv:main-po}. Assume now that after
the increments complete we change the configuration to move the leader to a
different process.  This process will initialize its speculative state $\Theta$
to the %
committed state $\Sigma=1$. If the new leader now receives a read command, it
will return $1$, violating the linearizability with respect to
Figure~\ref{fig:passive-spec}.

\subparagraph{Example 2.} 
The first increment is executed by the leader of an epoch $e$, which generates
$\delta_1=(x \leftarrow 1)$. The second increment is executed by
the leader of an epoch $e' > e$ before it delivers $\delta_1$ and, thus, in a
speculative state $\Theta = 0$. This generates %
$\delta_2=(x \leftarrow 1)$. Finally, the leader of $e'$ delivers $\delta_1$ and
then $\delta_2$, with the latter applied to a committed state $\Sigma = 1$.
This %
is allowed by atomic broadcast yet
violates
Invariant~\ref{inv:main-po}. It also violates linearizability similarly to
Example 1: if now the leader of $e'$ receives a read, %
it will incorrectly return $1$.

\begin{figure}[t]
\small
\begin{algorithm}[H]
	\setcounter{AlgoLine}{0}
        \DontPrintSemicolon
	$\assign{\pstate}{\pstate_0\in\State}$\\ 

	\smallskip
       
        \Function{${\tt execute}(\cmd)$}{
		$\pstate \stackrel{\cmd}{\rightarrow} \langle r, \delta \rangle$\;
		$\assign{\pstate}{\app(\pstate,\effid)}$\;
		\return $\rid$\;
	}
      \end{algorithm}
\caption{Passive replication specification.}
\label{fig:passive-spec}
\end{figure}

\begin{figure}[t]
	\begin{tabular}{@{}l@{\ \ }|@{\quad\ \ \ }l@{}}
        \scalebox{0.96}{%
	\begin{minipage}[t]{6.5cm}
	\vspace*{-12pt}
	\begin{algorithm}[H]
          \setcounter{AlgoLine}{0}
          \DontPrintSemicolon
	$\curepoch \in \epochset$\;
	$\curleader \in \procset$\;
	$\assign{\pstate}{\pstate_0\in \State}$ ~// committed state\\ 
	$\assign{\tstate}{\tstate_0\in \State}$ ~// speculative state\\  
	\smallskip

	\Function{$\execute(\cmd)$\label{passive:client-execute}}{
    $\assign{\id}{{\tt get\_unique\_id}()}$\;
    \send $\EXECUTE(\id, \cmd)$ \To $\curleader$\;
    \textbf{wait until receive} $\RESULT(\id, r)$\;
    \textbf{return} $r$\;
	}

	\smallskip

	\SubAlgo{\onreceive $\EXECUTE(\id, \cmd)$\label{passive:accept}}{
		\precond $\curleader = \pid$\; \label{line:curleader-pre}
		$\tstate \stackrel{\cmd}{\rightarrow} \langle r, \delta \rangle$\; \label{passive:exec}
		$\assign{\tstate}{\app(\tstate,\delta)}$\; \label{pr:code:state-assign}
		$\bcast(\langle\id, r, \delta \rangle)$\;\label{rep_bcast}
	}
      \end{algorithm}
	\vspace*{-5pt}
\end{minipage}}
&
        \scalebox{0.96}{%
	\begin{minipage}[t]{7.5cm}
	\vspace*{-12pt}
	\begin{algorithm}[H]
          \DontPrintSemicolon

         \SubAlgo{\textbf{upon} $\deliv(\langle\id, r, \delta \rangle)$\label{passive:deliv}}{
		$\assign{\pstate}{\app(\pstate, \delta)}$\; \label{passive:apply}
    \send $\RESULT(\id, r)$ \To $\client(\id)$\;
	}

	\smallskip

          	\SubAlgo{\textbf{upon} $\confchng(\langle e, M,
                  p_j\rangle)$\label{passive:change}\\
                {\color{blue}{\phantom{\hspace{14.5pt}}// $\confchng(\langle e,
                    M, p_j\rangle, \sigma)$}} \label{passive:change2}}{
		$\assign{\curepoch}{e}$\;
		$\assign{\curleader}{p_j}$\; \label{line:curleader-set}
		\If{$p_i = p_j$}{
			$\assign{\tstate}{\pstate}$\; \label{pr:code:tstate-assign}
      {\color{blue}{// $\assign{\langle \_, \_, \delta_1 \rangle \ldots \langle \_, \_, \delta_k \rangle}{\sigma}$}}\label{pr:code:apply1}\;
			 {\color{blue}{// \textbf{forall} $l = 1..k$ \textbf{do}}
                         $\assign{\tstate}{\app(\tstate, \delta_l)}$}}} \label{pr:code:state-assign2}
\end{algorithm}
\vspace*{-5pt}
\end{minipage}}
\end{tabular}
\caption{Passive replication on top of broadcast: code at process $p_i$.}
\label{fig:app-layer}
\end{figure}

\subsection{Primary-Order Atomic Broadcast}
\label{sec:poabcast}

To address the above problem, Junqueira et al. proposed \emph{primary-order
  atomic broadcast (POabcast)}~\cite{zab,junqueira2013barriers}, which
strengthens the classical atomic broadcast. %
We now briefly review POabcast and highlight its drawbacks, which motivates an
alternative proposal we present in the next section.  In our framework we can
define POabcast by adding the properties over histories $h$ in
\figref{fig:pocast-spec} to those of \figref{fig:rpob-properties}. This yields a
reconfigurable variant of %
POabcast that we call \emph{reconfigurable primary-order atomic broadcast
  (RPOabcast)}. RPOabcast also modifies the interface of reconfigurable atomic
broadcast (\S\ref{sec:spec}) by only allowing a process to call $\bcast$ if it
is the leader of its current configuration.

Property~\ref{prop:localord} (Local Order) restricts the delivery order of
messages broadcast in the same epoch: they must be delivered in the
order the leader broadcast them. Property~\ref{prop:globalord} (Global
Order) restricts the delivery order of messages broadcast in different epochs:
they must be delivered in the order of the epochs they were broadcast
in. Finally, Property~\ref{prop:primaryint} (Primary Integrity) ensures that the
leader of an epoch $e'$ does not miss relevant messages from previous epochs:
each message broadcast in an epoch $e < e'$ either has to be delivered by the
leader before entering $e'$, or can never be delivered at all.
Local and Global Order trivially imply Property~\ref{prop:total} (Total Order),
so we could omit it from the specification. POabcast is stronger than plain
atomic broadcast: the latter can be implemented from the former
if each process forwards messages to be broadcast to the leader of its
configuration.
\begin{proposition}
  Reconfigurable atomic broadcast can be implemented from RPOabcast.
\end{proposition}

\begin{figure}[p]
{\small
\begin{enumerate} 
 \setcounter{enumi}{5}
\item \label{prop:localord}
\textbf{Local Order.}
If the leader of some epoch $\eid$ receives ${\bcast}(\vid[_1])$  
before receiving $\bcast(\vid[_2])$, then any process that delivers $\vid[_2]$
must also deliver $\vid[_1]$ before $\vid[_2]$:
\[
  \begin{array}{@{}l@{}}
                            \forall \vid[_1], \vid[_2], i, j, k, l, l'. \
                            h_k = \bcast_i(\vid[_1]) \ \wedge\ h_l = \bcast_i(\vid[_2]) \ \wedge \ 
                            k < l \ \wedge\  {}
                            \\[1pt] 
\epochOf(k) = \epochOf(l)\  \wedge\ 
    h_{l'}=\deliv_j(\vid[_2]) \implies \exists k'.\ h_{k'} = \deliv_j(\vid_1)\ \wedge\ k' < l'
                          \end{array}
\]
\item \label{prop:globalord}
\textbf{Global Order.}
Assume the leaders of $\eid$ and $\eid'>\eid$ receive ${\bcast}(\vid[_1])$
and ${\bcast}(\vid[_2])$ respectively.  
If a process $\pid$ delivers $\vid[_1]$ and $\vid[_2]$, then it must deliver 
$\vid[_1]$ before $\vid[_2]$:
\[
\begin{array}{@{}l@{}}
\forall \vid[_1], \vid[_2], i, k, k', l, l'.\
  h_k = \bcast(\vid[_1]) \ \wedge\ h_l = \bcast(\vid[_2]) \ \wedge \ {}
\\[1pt]
  \epochOf(k) < \epochOf(l)\ \wedge\  h_{k'} = \deliv_i(\vid_1)\ \wedge\
h_{l'}=\deliv_i(\vid[_2])  \implies  k' < l'
\end{array}
\]
\item \label{prop:primaryint}
  \textbf{Primary Integrity.} Assume some
  process delivers an application message $\vid$ originally broadcast in an
  epoch $\eid$.  If any process $\pid$ joins an epoch $\eid'> \eid$, then
$\pid$ must deliver $\vid$ before joining $\eid'$:
\[
\begin{array}{l}
\forall \vid, i,  k, l, l', e, e'.\,
h_k = \bcast(\vid) \ \wedge\
  \epochOf(k)=\eid \ \wedge\
h_l = \deliv(\vid)\ \wedge\  {} \\[1pt]
  h_{l'} = \confchng_i(\langle \eid', \_, \_ \rangle)\  \wedge\
  \eid < \eid'   \implies
\exists k'.\ h_{k'} = \deliv_i(\vid)\ \wedge\ k' < l'
\end{array}
\]
\end{enumerate}
}
\caption{Properties of reconfigurable primary-order atomic broadcast over a history $h$.}
\label{fig:pocast-spec}
\end{figure}

\begin{figure}[p]
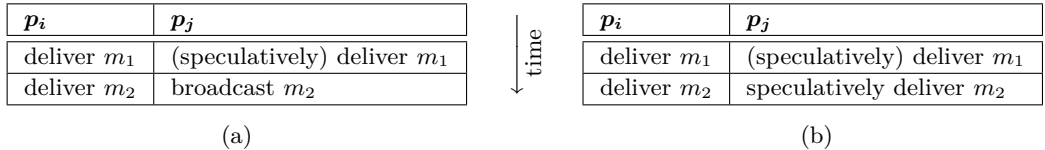

  {\small
\begin{enumerate} 
 \setcounter{enumi}{8}
\item \textbf{Basic Speculative Delivery Properties.}
\label{prop:sdeliva}
A process $p_i$ can speculatively deliver a given application message $\vid$ 
at most once in a given epoch and only if $p_i$ is the leader of the epoch,
$\vid$ has previously been broadcast, and $\vid$ has not yet been delivered by $p_i$:
\[
\begin{array}{@{}l@{}}
\forall i, j, k, \sigma.\, h_k \,{=}\, \confchng_i(\langle \_, \_, p_j \rangle, \sigma) \implies
(\sigma \,{\not=}\, \bot {\implies} p_i \,{=}\, p_j) \, \wedge\,
(\forall m_1, m_2 \in \sigma.\, m_1 \,{\not=}\, m_2) 
\\[1pt]
{} \wedge\ 
  (\forall m \in \sigma.\ (\exists l.\ h_l = \bcast(\vid)\ \wedge\ l < k)
\ \wedge\ 
(\neg\exists l.\ h_l = \deliv_i(\vid) \ \wedge\  l < k))
\end{array}
\]
\item \textbf{\Consistency{}.}
\label{prop:consistency}
\begin{enumerate}
\item \label{prop:sdelivc} Consider
  $\vid[_1]$ and $\vid[_2]$ broadcast in different epochs. Assume that a process $\pid$
  delivers $\vid[_2]$, and a process $\pid[j]$ broadcasts $\vid[_2]$ in an epoch
  $e'$. Then $\pid$ delivers $\vid[_1]$ before $\vid[_2]$ iff $\pid[j]$ delivers
  $\vid[_1]$ before joining $e'$ or speculatively delivers $\vid[_1]$ when
  joining $e'$:
\[
\begin{array}{l}
\forall \vid[_1], \vid[_2], i, j, k_0, l_0, k, l, l', \sigma, e'.\
  h_{k_0} = \bcast(\vid[_1])\ \wedge\ h_{l_0} = \bcast(\vid[_2])\ \wedge\ {}
\\[1pt]
  \epochOf(k_0) \not= \epochOf(l_0)\ \wedge\
  h_k = \deliv_i(\vid[_2]) \ \wedge\   h_l = \bcast_j(\vid[_2])\ \wedge\ {}
  \\[1pt]
  h_{l'}=\confchng_j(\langle e', \_, p_j \rangle, \sigma)\ \wedge\ \epochOf(l) = e'
 \implies 
  \\[1pt]
  ((\exists k'.\, h_{k'} = \deliv_i(\vid[_1]) \, \wedge\,  k' < k)
  \iff
  ((\exists l''.\, h_{l''} = \deliv_j(\vid[_1]) \, \wedge\,  l'' < l')\, \vee\,
    \vid_1 \in \sigma) \!\!\!\!
\end{array}
\]
\item \label{prop:sdelivb} Consider
  $\vid[_1]$ and $\vid[_2]$ broadcast in different epochs. Assume that a process $\pid$
  delivers $\vid[_2]$, and a process $\pid[j]$ speculatively delivers $\vid[_2]$
  when joining an epoch $e'$. Then $\pid$ delivers $\vid[_1]$ before $\vid[_2]$
  iff $\pid[j]$ delivers $\vid[_1]$ before joining $e'$ or speculatively delivers
  $\vid[_1]$ before $\vid[_2]$ when joining $e'$:
\[
\begin{array}{l}
\forall \vid[_1], \vid[_2], i, j, k_0, l_0, k, l, \sigma, e'.\
  h_{k_0} = \bcast(\vid[_1])\ \wedge\ h_{l_0} = \bcast(\vid[_2])\ \wedge\ {}
  \\[1pt]
  \epochOf(k_0) \not= \epochOf(l_0)\ \wedge\
  h_k = \deliv_i(\vid[_2]) \ \wedge\ h_{l}=\confchng_j(\langle e' , \_, p_j \rangle, \sigma)
\  \wedge {} \!\!\!\!\!
  \\[1pt]
  \vid_2 \in \sigma \implies
 ((\exists k'.\ h_{k'} = \deliv_i(\vid[_1]) \ \wedge\  k' < k)
  \iff
  \\[1pt]
  ((\exists l'.\ h_{l'} = \deliv_j(\vid[_1]) \ \wedge \  l' < l)\ \vee\
\sigma = \_\vid_1\_\vid_2\_)) 
\end{array}
\]
\end{enumerate}
\end{enumerate}

\smallskip

\centerline{
\begin{tabular}{@{}c@{\quad\ \ }c@{\quad\ \ }c@{}}
  \begin{tabular}{@{}|@{\ \ }l@{\ \ }|@{\ \ }l@{\ \ }|@{}}
   \hline
    \boldmath $p_i$ & \boldmath $p_j$
    \\
   \hline
    \hline
deliver $m_1$ & (speculatively) deliver $m_1$
    \\
    \hline
deliver $m_2$ & broadcast $m_2$
    \\
    \hline
  \end{tabular}%
&
    $\Bigg\downarrow
  \rotatebox[origin=c]{90}{time}$
&
  \begin{tabular}{@{}|@{\ \ }l@{\ \ }|@{\ \ }l@{\ \ }|@{}}
   \hline
    \boldmath $p_i$ & \boldmath $p_j$
    \\
   \hline
    \hline
deliver $m_1$ & (speculatively) deliver $m_1$
    \\
    \hline
deliver $m_2$ & speculatively deliver $m_2$
    \\
    \hline
  \end{tabular}
  \smallskip
  \smallskip
  \\
  (a) && (b)
\end{tabular}}
}
\caption{Properties of speculative primary-order atomic broadcast over a history
  $h$. Property~\ref{prop:consistency} replaces Property~\ref{prop:primaryint}
  from Figure~\ref{fig:pocast-spec}. The tables
  summarize its action orderings: the actions at the top happen before the
  actions at the bottom.
}
\label{fig:spo-spec}
\end{figure}

When the passive replication protocol in Figure~\ref{fig:app-layer} is used with
POabcast instead of plain atomic broadcast, Invariant~\ref{inv:main-po} holds,
and the protocol yields a service linearizable with respect to the specification
in Figure~\ref{fig:passive-spec}~\cite{junqueira2013barriers}. In particular,
Local Order disallows Example~1 from \S\ref{sec:po-bcast}, and Primary Integrity
disallows Example~2 (which does not violate either Local or Global Order).
POabcast can be obtained from our Vertical Atomic Broadcast (VAB) 
algorithm in \S\ref{sec:vp_protocol} as follows.
First, VAB already guarantees both Local and Global Order: e.g., this is the case for Local Order because
processes are connected by reliable FIFO channels.
\begin{theorem}
  VAB guarantees Local and Global order.
\end{theorem}

Second, to ensure Primary Integrity, neither the new leader nor the
followers invoke $\confchng$ upon receiving $\NEWCONFIG$
(line~\ref{vp_ready_confchg}) or $\NEWSTATE$
(line~\ref{vp_new_state_confchg}). Instead, the leader first waits until it
receives $\NEWSTATEACK$ messages from all followers (line
\ref{vp_quorum_of_new_state_acks}) and tells the processes to deliver all
application messages from the previous epoch via $\COMMIT$ messages. Only once a
process delivers all these application messages does it invoke $\confchng$ for
the new configuration (and if the process is the leader, starts broadcasting).

Deferring the invocation of $\confchng$ at the leader is the key to guarantee
Primary Integrity. On the one hand, it ensures that, before the newly elected
leader of an epoch $e'$ generates $\confchng$, it has delivered all application
messages that could have been delivered in previous epochs:
Invariant~\ref{inv:prefix-of-higher-epoch} from \S\ref{sec:vp_protocol}
guarantees that the leader's initial log includes all such messages. On the
other hand, the leader can also be sure that any message broadcast in an epoch
$<e'$ but not yet delivered can \emph{never} be delivered by any process. This is
because, by the time the leader generates $\confchng$, all followers in $e'$
have overwritten their log with that of the new leader. 

Since deferring $\confchng$ results in deferring the start of broadcasting by 
the leader, the modified VAB protocol has a reconfiguration downtime
of 2 messages delays. This cost is inherent: the lower bound of
Friedman and van Renesse~\cite{friedman-lower} on the latency of 
Strong Virtually Synchronous broadcast (a variant of POabcast)
implies that any solution must have a non-zero downtime.
In the next section we circumvent this limitation by introducing a weaker
variant of POabcast, which we show sufficient for passive replication.

\section{Speculative Primary-Order Atomic Broadcast}
\label{sec:srpob}

We now introduce \emph{speculative primary-order atomic broadcast} (SPOabcast),
a weaker variant of POabcast that allows implementing passive replication with
minimal downtime. During reconfiguration, SPOabcast allows the new leader to
deliver messages from previous epochs \emph{speculatively} -- without waiting for
them to become durable -- and start broadcast right away.

\subparagraph{SPOabcast specification.}
SPOabcast modifies the interface of reconfigurable atomic broadcast
(\S\ref{sec:spec}) in two ways. First, like in POabcast, a process can call
$\bcast$ only if it is the leader.  Second, the $\confchng$ upcall for a
configuration $C$ carries an additional argument $\sigma$:
$\confchng(C, \sigma)$. When the upcall is invoked at the leader of $C$,
$\sigma$ is a sequence of messages \emph{speculatively delivered} to the leader
($\sigma$ is not used at followers).
SPOabcast is defined by replacing Primary Integrity in the definition of
POabcast by the properties in \figref{fig:spo-spec}. Property~\ref{prop:sdeliva} is
self-explanatory. Property~\ref{prop:consistency} (\Consistency{}) constrains
how speculative deliveries are ordered with respect to ordinary deliveries and
broadcasts. For the ease of understanding, in Figure~\ref{fig:spo-spec} we
summarize these orderings in tables.

Part (a)/``only if'' of \Consistency{} is a weaker form of Primary
Integrity. Assume that the leader $p_j$ of an epoch $e'$ broadcasts a message
$m_2$. The property ensures that for any message $m_1$ delivered before $m_2$ at
some process $p_i$, the leader $p_j$ has to either deliver $m_1$ before joining
$e'$ or \emph{speculatively deliver $m_1$ when joining $e'$}. As we demonstrate
shortly, the latter option, absent in Primary Integrity, allows our
implementation of SPOabcast to avoid extra downtime during reconfiguration. Part
(a)/``if'' conversely ensures that, if the leader $p_j$ speculatively delivers
$m_1$ before broadcasting $m_2$, then $m_1$ must always be delivered before
$m_2$. This ensures that the speculation performed by the leader $p_j$ is
correct if any of the messages it broadcasts (e.g., $m_2$) are ever delivered
at any process.
Part (b) of \Consistency{} ensures that the order of messages in a sequence
speculatively delivered at a $\confchng$ upcall cannot contradict the order of
ordinary delivery.

Speculative delivery provides weaker guarantees than ordinary delivery, since it
does not imply durability. In particular, we allow a message to be speculatively
delivered at a process $p$ but never delivered anywhere, e.g., because $p$
crashed. However, in this case Part (a)/``if'' of \Consistency{} ensures that
all messages $p$ broadcast after such a non-durable speculative delivery will
also be lost. As we show next, this allows us to use SPOabcast to correctly
implement passive replication without undermining its durability guarantees.

\subparagraph{Passive replication using SPOabcast.}  The passive replication
protocol in Figure~\ref{fig:app-layer} requires minimal changes to be used with
SPOabcast, highlighted in blue.
When the leader of an epoch $e$ receives a $\confchng$ upcall for $e$
(line~\ref{passive:change2}), in addition to setting the speculative state
$\tstate$ to the committed state $\pstate$, the leader also applies the state
updates speculatively delivered via $\confchng$ to $\tstate$
(lines~\ref{pr:code:apply1}-\ref{pr:code:state-assign2}). The leader can then
immediately use the resulting speculative state to execute new commands
(line~\ref{passive:accept}). We prove the following in \tr{\ref{sec:pr:correct}}{\nprproof}.
\begin{theorem}\label{thm:correctness-passive}
  The version of the protocol in Figure~\ref{fig:app-layer} that uses SPOabcast
  satisfies Invariant~\ref{inv:main-po} and implements a service linearizable
  with respect to the specification in Figure~\ref{fig:passive-spec}.
\end{theorem}
In particular, part (a)/``only if'' of \Consistency{} disallows Example~2 from
\S\ref{sec:po-bcast}: it ensures that the leader broadcasting $\delta_2$ will be
aware of $\delta_1$, either via ordinary or speculative delivery. More
generally, part (a) ensures that, if a process $p_i$ delivers a state update
$\delta_2$ broadcast by a leader $p_j$, then at the corresponding points in the
execution, $p_i$ and $p_j$ are aware of the same set of updates (cf. the table
in Figure~\ref{fig:spo-spec}). Part (b) of \Consistency{} furthermore ensures
that the two processes apply these updates in the same order. This contributes
to validating Invariant~\ref{inv:main-po} and, thus, the specification in
Figure~\ref{fig:passive-spec}.

\subparagraph{Implementing SPOabcast.}
To implement SPOabcast we modify the Vertical Broadcast Protocol in
Figures~\ref{vp_failure_free}-\ref{vp_reconfiguration} as follows.  First, since
$\bcast$ can only be called at the leader, we replace
lines~\ref{vp_broadcast}-\ref{vp_forward} by line~\ref{po_broadcast}. Thus, the
leader handles $\bcast$ calls in the same way it previously handled $\FORWARD$
messages. Second, we augment $\confchng$ upcalls with speculative deliveries,
replacing line~\ref{vp_ready_confchg} by line~\ref{po_newconfig_sdeliv}, and
line~\ref{vp_new_state_confchg} by line~\ref{vp_new_state_confchg2}. Thus, the
$\confchng$ upcall at the leader speculatively delivers all application messages
in its log that have not yet been (non-speculatively) delivered.
It is easy to check that these modifications do not change the $0$-downtime
guarantee of Vertical Atomic Broadcast.
\begin{theorem}\label{thm:correctness2}
  The primary-order version of the Vertical Atomic Broadcast protocol is a
  correct implementation of speculative primary-order atomic broadcast.
\end{theorem}

Thus, Theorems~\ref{thm:correctness-passive} and~\ref{thm:correctness2} allow us
to use VAB to replicate even non-deterministic services
while minimizing the downtime from routine reconfigurations, e.g., those for
load balancing.

We prove Theorem~\ref{thm:correctness2} in \tr{\ref{sec:proof-po}}{\npoproof}.
Here we informally explain why the above protocol validates the key part
(a)/``only if'' of \Consistency{}, weakening Primary Integrity
(cf. the explanations we gave regarding the latter
at the end of \S\ref{sec:poabcast}).
On the one hand, as in the ordinary VAB,
Invariant~\ref{inv:prefix-of-higher-epoch} from \S\ref{sec:vp_protocol}
guarantees that the log of a newly elected leader of an epoch $e'$ contains all
application messages $m_1$ that could have been delivered in epochs $<e'$. The
new leader will either deliver or speculatively deliver all such messages before
broadcasting anything (line~\ref{po_newconfig_sdeliv}). On the other hand, if
the leader broadcasts a message $m_2$, then a follower will only accept it after
having overwritten its log with the leader's initial one, received in
$\NEWSTATE$ (line~\ref{vp_new_state}). This can be used to show that, if $m_2$
is ever delivered, then any message broadcast in an epoch $< e'$ that was not in
$\NEWSTATE$ will never get delivered.

%% file: related.tex
\section{Related Work}
\label{sec:rw}

The \emph{vertical} paradigm of implementing reconfigurable services by
delegating agreement on configuration changes to a separate component was first
introduced by Lynch and Shvartsman~\cite{rambo} for emulating
dynamic atomic registers. It was further applied by Lamport et
al.~\cite{vertical-paxos} to solve reconfigurable single-shot consensus,
yielding the Vertical Paxos family of protocols. Vertical Paxos and its
follow-ups~\cite{cheappaxos,niobe,ken-book,farm,podc19} require prior
configurations to be disabled (``wedged'') at the start of reconfiguration. In
contrast, our VAB protocol allows the latest functional configuration to
continue processing messages while the agreement on the next configuration is in
progress. This results in the downtime of $0$ when reconfiguring from a
functional configuration. This feature is particularly desirable for atomic
broadcast, where we want to keep producing new decisions when reconfiguration is
triggered for load balancing rather than to handle failures.

To achieve the minimal downtime, the VAB protocol uses different epoch variables
to guard the normal operation ($\epoch$) and reconfiguration ($\newepoch$). By
not modifying the $\epoch$ variable during the preliminary reconfiguration
steps, the protocol allows the old configuration to operate normally while the
reconfiguration is in progress (cf. the proof of Theorem~\ref{thm:zero-downtime}
in \S\ref{sec:vp_protocol}). In contrast, Vertical Paxos uses a single epoch
variable ($\mathsf{maxBallot}$) for both purposes, thus 
disabling the current configuration at the start of reconfiguration.
Our protocol for SPOabcast further extends the minimal downtime guarantee to the
case of passive replication.

Both our VAB and SPOabcast protocols achieve an optimal steady-state latency of
two message delays~\cite{lower-bound}. Although Junqueira et
al.~\cite{junqueira2013barriers} show that no POabcast protocol can guarantee
optimal steady-state latency if it relies on black-box consensus to order
messages, our SPOabcast implementation is not subject to this impossibility
result, as it does not use consensus in this manner.

Although the vertical approach has been widely used in
practice~\cite{chain-replication,corfu,spinnaker,farm,bigtable}, prior systems
have mainly focused on engineering aspects of directly implementing a replicated
state machine for a desired service rather than basing it on a generic atomic
broadcast layer.
Our treatment of Vertical Atomic Broadcast develops a formal foundation that
sheds light on the algorithmic core of these systems.  This can be reused for
designing future solutions that are provably correct and efficient.

Most reconfiguration algorithms that do not rely on an auxiliary configuration
service
can be traced back to the original technique of Paxos~\cite{paxos}, which
intersperses reconfigurations within the stream of normal command agreement
instances. The examples of practical systems that follow this approach include
SMART~\cite{smart}, Raft~\cite{raft}, and Zookeeper~\cite{zab-reconfig}.
\ag{Zab reconfiguration kind of does our speculative broadcast, even though they
  don't formalize it. So we may want to be more honest and more careful with
  saying things like ``our key insight'' elsewhere in the paper.}  Other
non-vertical algorithms~\cite{reconfiguring} implement reconfiguration by
spawning a separate non-reconfigurable state machine for each newly introduced
configuration. In the absence of an auxiliary configuration service, these
protocols require at least $2f+1$ processes in each
configuration~\cite{lower-bound}, in contrast to $f+1$ in our atomic broadcast
protocols.

The fault-masking protocols of Birman et al.~\cite{ken-book} and a recently
proposed MongoDB reconfiguration protocol~\cite{mongodb} separate the message
log from the configuration state, but nevertheless replicate them at the same
set of processes. As in non-vertical solutions, these algorithms require $2f+1$
replicas. They also follow the Vertical Paxos approach to implement
reconfiguration, and as a result, may wedge the system prematurely as we explain
above.

A variant of Primary Integrity, known as
Strong Virtual Synchrony (or Sending View Delivery~\cite{gcs-survey}), 
was originally proposed by 
Friedman and van Renesse~\cite{friedman-lower} who also studied its inherent costs.
Our SPOabcast abstraction is a relaxation of Strong Virtually Synchrony %
and primary-order atomic broadcast (POabcast) of Junqueira et al.~\cite{zab,junqueira2013barriers}.
Keidar and Dolev~\cite{idit-colours} proposed Consistent
Object Replication Layer (COReL) in which every delivered message is assigned a
color such that a message is ``yellow'' if it was received and acknowledged by a
member of an operational quorum, and ``green'' if it was
acknowledged by all members of an operational quorum. While the COReL's yellow
messages are similar to our speculative messages, Keidar and Dolev did not
consider their potential applications, in particular, their utility for
minimizing the latency of passive replication.

%% file: corr.tex
\section{Proof of Correctness for Vertical Atomic Broadcast
  (Theorem~\ref{thm:correctness})}
\label{sec:correctness}

\begin{figure}[h]
	{\small		
	\begin{enumerate}
        \setcounter{enumi}{2}
        
      \item \label{inv:epochprobe} Let $\eid$ be an epoch such that all its
	members at some point set their $\epoch$ variable to $\eid$. Then the
	leader $\pid[j]$ of any epoch $\eid'>\eid$ has $\epoch\geq \eid$ right
	before receiving $\NEWCONFIG(\eid',\_)$.

      \item \label{inv:origin} 
	If the leader of an epoch $\eid$ sends
	$\NEWSTATE(\eid, \msgi, \_)$ such that $\msgi[k] = \vid$,
	then there exists an epoch $\eid'<\eid$ such that the leader of
	$\eid'$ has previously sent $\ACCEPT(\eid', k , \vid)$.
	\end{enumerate}
	}
	\caption{Invariants maintained by Vertical Atomic Broadcast.}
	\label{fig:vp_invariants}
\end{figure}

Our proof of Theorem~\ref{thm:correctness} relies on auxiliary
Invariants~\ref{inv:epochprobe} and \ref{inv:origin} in
Figure~\ref{fig:vp_invariants}, in addition to Invariant~\ref{inv:prefix-of-higher-epoch} from
\S\ref{sec:vp_protocol}. \invref{inv:epochprobe}
ensures that if an epoch $\eid$ has been activated, then the leader of any epoch
$\eid' > \eid$ must have been initialized at an epoch $\geq \eid$. Informally,
this holds because, when probing to look for a leader of $\eid'$, we only skip
epochs that are not activated and will never be
(\S\ref{sec:vp_protocol}). Hence, probing cannot go lower than an activated
epoch $\eid$. \invref{inv:origin} gives the key argument for Integrity: it
ensures that any message in the $\msg$ array of a process has been broadcast
before. We next prove the key \invref{inv:prefix-of-higher-epoch}.

\begin{proof}[Proof of {\confchgnot}.]
\hfill
\begin{itemize}
\item Property (\ref{prop:wf-1}).  Line~\ref{vp_compare_and_swap} guarantees
  that there exists a single $\NEWCONFIG(\eid, \vm)$ message. From this the
  required follows trivially.

\item Property (\ref{prop:wf-3}). Let $p_i$ be a process that calls
  $\confchng(\eid, \vm, \_)$ for an epoch $\eid$. Assume first that $p_i$ calls
  it when handling $\NEWCONFIG(\eid, \vm)$. This message is only sent to the
  leader of $\eid$, which is guaranteed to be included in $\vm$ by the
  constraint on ${\tt compute\_membership}()$. Assume now that $p_i$ calls
  $\confchng(\eid, \vm, \_)$ when handling a $\NEWSTATE(\eid,\_, \_)$
  message. Then the property follows from line~\ref{vp_send_new_state}. Hence,
  any process that calls $\confchng(\eid, \vm, \_)$ is member of $\vm$.

\item Property (\ref{prop:wf-4}). By lines~\ref{vp_probe_pre}
  and~\ref{vp_new_state_pre}, the $\epoch$ variable at a process does not
  decrease. Thus, $\eid_1 \leq \eid_2$. Furthermore, by line
  \ref{vp_compare_and_swap} there exists a single $\NEWCONFIG$ message per
  epoch. Therefore, each process calls $\confchng$ at most once per epoch and
  $\eid_1 < \eid_2$, as required.

\item Property (\ref{prop:wf-5}).  Consider a configuration
  $\confid = \aconf[@][\_][{\pid[j]}]$ and assume that a process $\pid$ invokes
  $\confchng(\confid)$. Then the leader $\pid[j]$ must have earlier received
  $\NEWCONFIG(\eid, \_)$. By line \ref{vp_compare_and_swap}, this occurs when
  ${\tt compare\_and\_swap}(\_, \confid)$ succeeds. Therefore, $\intro(\confid)$
  is invoked before $\pid$ invokes $\confchng(\confid)$. Moreover, line
  \ref{vp_compare_and_swap} guarantees that there exists a single
  $\intro(\confid)$ action.
\end{itemize}
\end{proof}

\begin{proof}[Proof of \invref{inv:prefix-of-higher-epoch}]
We prove the invariant by induction on $\eid'$. Assume the invariant holds for each
$\eid' < \eid''$.  We now show that it holds for $\eid' = \eid''$ by induction on the length
of the protocol execution. We only consider the most interesting transition in
line~\ref{vp_new_state}, where a process $\pid$ receives $\NEWSTATE(\eid'', \vmsg'',\_)$
from the leader of $\eid''$ and sets $\msg$ to $\vmsg''$ (line~\ref{vp_newstate_msg_assignment}).
We show that $\vmsg''[k]=\vid$, so that after this transition 
$\pid$ has $\msg[k] = \vid$, as required.
Let $\pid[j]$ be the leader of $\eid''$, i.e., the process that sends
$\NEWSTATE(\eid'', \vmsg'',\_)$, and let $\eid_0<\eid''$ be the value of $\epoch$ at
$\pid[j]$ right before receiving the $\NEWCONFIG(\eid'', \_)$ message at line \ref{vp_new_config}.
By the premise of the invariant, some process sends the message
$\COMMIT(\eid, k)$. Then at some point all members of $\eid$ set their
$\epoch$ variable to $\eid$, so by \invref{inv:epochprobe} we get $\eid \leq \eid_0$.

Consider first the case when $\eid<\eid_0$. Since $\eid_0<\eid''$, by
the induction hypothesis, $\pid[j]$ has $\msg[k]=\vid$ when it sends
$\NEWSTATE(\eid'', \msg, \_)$, which implies the required. 
Assume now that $\eid = \eid_0$. There are two possibilities: either
$\pid[j]$ is the leader of $\eid$, or it is a follower in $\eid$.
We only consider the latter case, since the former is analogous. 
By the premise of the invariant, some process sends $\COMMIT(\eid, k)$ while having
$\msg[k]=\vid$. Then either each follower in $\eid$ receives
$\ACCEPT(\eid, k, \vid)$ and replies with $\ACCEPTACK(\eid, k)$, or
each follower in $\eid$ receives $\NEWSTATE(\eid, \vmsg', \_)$ with
$\vmsg'[k] = m$ and replies with $\NEWSTATEACK(\eid)$. In particular, this is
the case for $\pid[j]$, and by lines~\ref{vp_accept_pre} and~\ref{vp_new_state_pre},
$\pid[j]$ must handle one of these messages before it
handles $\NEWCONFIG(\eid', \_)$. After handling
one of the above messages and while $\epoch=\eid$,
$\pid[j]$ does not overwrite $\msg[k]$, so that $\msg[k]=\vid$. Then $\pid[j]$ has
$\msg[k]=\vid$ before handling $\NEWCONFIG(\eid', \_)$. Hence, $\pid[j]$ has
$\msg[k]=\vid$ when it sends $\NEWSTATE(\eid'', \msg, \_)$, which implies the
required.
\end{proof}

\begin{proof}[Proof of \invref{inv:epochprobe}.]
Since $\pid[j]$ was chosen as the leader of $\eid'$, 
we know that this process replied with $\PROBEACK(\true, \eid')$
to a message $\PROBE(\eid', \eid^{*})$ from a process executing $\reconf()$, where $\eid^{*}<\eid'$.
Therefore, $\pid[j]$ was a member at the epoch $\eid^*$ that was being
probed.

When a process handles $\PROBE(\eid',\_)$, it sets $\newepoch$
to $\eid'$. Then, by line~\ref{vp_new_state_pre}, a process cannot set $\epoch$ to an epoch
$<\eid'$ after handling $\PROBE(\eid',\_)$. By the premise of the
invariant, at some point all members of $\eid$ set their $\epoch$
variable to $\eid$. Since $\eid'>\eid$, if a member of $\eid$ handles
$\PROBE(\eid',\_)$, then it must have set
$\epoch=\eid$ before. Thus, no member of $\eid$
can reply with $\PROBEACK(\false, \eid')$.

Probing ends when at least one process sends a $\PROBEACK(\true, \eid')$ message.
Thus, the probing phase could have not gone beyond $\eid$, so that
$\eid \leq \eid^* < \eid'$.
Let $\eid_0<\eid'$ be the value of $\epoch$ at $\pid[j]$ right before receiving 
the $\NEWCONFIG(\eid', \_)$ message at line \ref{vp_new_config}.
By the protocol, we know that the transition $\NEWCONFIG(\eid', \_)$ is triggered after $\pid[j]$ was probed in epoch $\eid^{*}\geq \eid$.
Furthermore, by line \ref{vp_send_probe_ack} when $\pid[j]$ handled $\PROBE(\eid', \eid^{*})$, it had $\epoch\geq \eid^{*}$.
Then $\eid \leq \eid^{*} \leq \eid_0$, as required. 
\end{proof}

\begin{proof}[Proof of \invref{inv:origin}.] 
Let $\pid$ be the leader of $\eid$. 
We prove the invariant by induction on $\eid$.
Assume that the invariant holds for all $\eid < \eid^{*}$.
We now show it for $\eid = \eid^{*}$.
According to the protocol, 
the process $\pid$ sends $\NEWSTATE(\eid, \msgi, \_)$ 
with $\msgi[k] = \vid$
when it handles the message $\NEWCONFIG(\eid, \_)$ at line \ref{vp_new_config}.
Let $\eid''<\eid$ be the epoch at $\pid$ just before handling $\NEWCONFIG(\eid, \_)$.
At this moment we have $\msg[k]=m$ at $\pid$.
Then either $\pid$ has received a message $\NEWSTATE(\eid'', \msgi'', \_)$ with
$\msgi''[k] = \vid$ or a message $\ACCEPT(\eid'', k, \vid)$ has been sent.
The former case is immediate by the induction hypothesis. 
The latter gives us what we wanted to prove.
\end{proof}

\begin{lemma}
\label{thm:commit-msg}
Assume that the leader $p_i$ of an epoch $e_1$ sends $\COMMIT(e_1, k)$ while
having $\msg[k] = m_1$, and the leader $p_j$ an epoch $e_2$ sends
$\COMMIT(e_2, k)$ while having $\msg[k] = m_2$. Then $m_1 = m_2$.
\end{lemma}
\begin{proof}
  Without loss of generality, assume $e_1 \le e_2$. The lemma trivially holds if
  $e_1 = e_2$. Assume now that $e_1 < e_2$. By
  \invref{inv:prefix-of-higher-epoch}, when $p_j$ sends $\COMMIT(e_2, k)$, it
  has $\msg[k] = m_1$. But then $m_2 = \msg[k] = m_1$, as required.
\end{proof}

\begin{lemma}
  \label{thm:commit-pos}
  \label{inv:commit_ord_new}
  Assume that the leader of an epoch $e_1$ sends $\COMMIT(e_1, k_1)$ while
  having $\msg[k_1] = m$, and the leader of an epoch $e_2$ sends
  $\COMMIT(e_2, k_2)$ while having $\msg[k_2] = m$. Then $k_1 = k_2$.
\end{lemma}
\begin{proof}
  Since the leader $p_i$ of $e_1$ sends $\COMMIT(e_1, k_1)$ while having
  $\msg[k_1] = m$, it must previously send either $\ACCEPT(\eid_1, k_1, m)$ or
  $\NEWSTATE(\eid, \msgi, \_)$ with $\msgi[k]=\vid$. In the latter case, by
  \invref{inv:origin}, for some epoch $e'_1 < e_1$ the leader of $e'_1$ sends
  $\ACCEPT(\eid'_1, k_1, m)$. Thus, in both cases for some epoch $e''_1 \le e_1$
  the leader of $e''_1$ sends $\ACCEPT(\eid''_1, k_1, m)$. We analogously
  establish that for some epoch $e''_2 \le e_2$ the leader of $e''_2$ sends
  $\ACCEPT(\eid''_2, k_2, m)$. Either leader sends the $\ACCEPT$ message upon
  receiving a $\FORWARD(m)$ message, generated when a process calls
  $\bcast(m)$. By~(\ref{prop:env-1}), there is a unique $\bcast(m)$ call, and
  thus a unique $\FORWARD(m)$ message. This implies $\eid''_1 = \eid''_2$ and
  $k_1 = k_2$, as required.
\end{proof}

\begin{proof}[Proof of Integrity.]
  We first prove that $\vid$ is broadcast before its delivered.  Assume that a
  process $p_i$ delivers $m$. Then $\pid$ handles a $\COMMIT(\eid, k)$ message
  for some epoch $\eid$ and has $\msg[k]=\vid$ at that moment. Consider first
  the case when the leader of $\eid$ sends $\COMMIT(\eid, k)$ because it
  receives an $\ACCEPTACK(\eid, k)$ message from every follower.  This implies
  that the leader $\pid[j]$ has sent an $\ACCEPT(\eid, k, \vid')$ message to all
  followers.  After $\pid$ handles $\ACCEPT(\eid, k, \vid')$ (or sends it if
  $\pid=\pid[j])$ and while $\epoch=\eid$, $\pid$ does not overwrite $\msg[k]$,
  and hence, it has $\msg[k] = \vid'$. Since $\pid$ has $\msg[k]=\vid$ when it
  handles $\COMMIT(\eid, k)$, then $\vid=\vid'$. Therefore, $\pid[j]$ sends
  $\ACCEPT(\eid, k, \vid)$. Then $\pid[j]$ has received a $\FORWARD(m)$ message
  from some process, which must have previously called $\bcast(m)$. The latter
  process thus broadcast $m$ before $p_i$ delivered it, as required.

  Consider now the case when the leader of $\eid$ sends $\COMMIT(\eid, k)$
  because it receives a $\NEWSTATEACK$ message from every follower.  Let
  $\NEWSTATE(\eid, \msgi, \_)$ be the message sent by the leader of $\eid$. We
  have $\msgi[k]=\vid$. By \invref{inv:origin} there exists an epoch
  $\eid'<\eid$ in which its leader sends $\ACCEPT(\eid', k, \vid)$. The leader
  does this upon receiving a $\FORWARD(m)$ message from some process, which must
  have previously called $\bcast(m)$. The latter process thus broadcast $m$
  before $p_i$ delivered it, as required.
  
  Finally, we prove that $\pid$ delivers $\vid$ at most once. Indeed, if $\pid$
  delivered $\vid$ twice, then it would receive $\COMMIT(\vid, k_1)$ and
  $\COMMIT(\vid, k_2)$ for some $k_1 \not= k_2$. But this contradicts
  Lemma~\ref{thm:commit-pos}.
\end{proof}

\begin{proof}[Proof of Total Order.]
  Assume that a process $\pid[i]$ delivers $m_1$ before $m_2$, and that a
  process $\pid[j]$ delivers $m_2$. Then: 
\begin{itemize}
\item $\pid[i]$ receives a message $\COMMIT(e_1, k_1)$ sent by the leader of an
  epoch $e_1$, and $\pid[i]$ has $\msg[k_1] = m_1$ when receiving this message;
\item $\pid[i]$ receives a message $\COMMIT(e_2, k_2)$ sent by the leader of an
  epoch $e_2$, and $\pid[i]$ has $\msg[k_2] = m_2$ when receiving this message; and 
\item $\pid[j]$ receives a message $\COMMIT(e'_2, k'_2)$ sent by the leader of an
  epoch $e'_2$, and $\pid[j]$ has $\msg[k_2] = m'_2$ when receiving this message.
\end{itemize}
It is easy to see that when the leaders of $e_1$, $e_2$ and $e'_2$ send the
$\COMMIT$ messages, they must have $\msg[k_1] = m_1$, $\msg[k_2] = m_2$ and
$\msg[k'_2] = m'_2$, respectively. Then by Lemma~\ref{thm:commit-pos},
$k'_2 = k_2$. Since $k_1 < k_2$, before delivering $m_2$ the process $\pid[j]$
has to deliver a message $m'_1$ at position $k_1$. Then:
\begin{itemize}
\item $\pid[j]$ receives a message $\COMMIT(e'_1, k_1)$ from the leader of an
  epoch $e'_1$, and $\pid[j]$ has $\msg[k_1] = m'_1$ when receiving this message.
\end{itemize}
Then the leader of $e'_1$ has $\msg[k_1] = m'_1$ when sending the $\COMMIT$
message. Hence, by Lemma~\ref{thm:commit-msg}, $m'_1 = m_1$. We have thus shown
that $\pid[j]$ delivers $m_1$ before $m_2$, as required.
\end{proof}

\begin{proof}[Proof of Agreement.]
  Assume that a process $\pid[i]$ delivers $m_1$ and a process $\pid[j]$
  delivers $m_2$. Then:
\begin{itemize}
\item $\pid[i]$ receives a message $\COMMIT(e_1, k_1)$ sent by the leader of an
  epoch $e_1$, and $\pid[i]$ has $\msg[k_1] = m_1$ when receiving this message;
\item $\pid[j]$ receives a message $\COMMIT(e_2, k_2)$ sent by the leader of an
  epoch $e_2$, and $\pid[j]$ has $\msg[k_2] = m_2$ when receiving this message.
\end{itemize}
It is easy to see that when the leaders of $e_1$ and $e_2$ send the $\COMMIT$
messages, they must have $\msg[k_1] = m_1$ and $\msg[k_2] = m_2$, respectively.
Without loss of generality, assume that $k_1 \le k_2$. If $k_1 = k_2$, then by
Lemma~\ref{thm:commit-msg}, $m_1 = m_2$, which trivially implies the required.
Assume now that $k_1 < k_2$. Then before delivering $m_2$ at position $k_2$,
process $\pid[j]$ must deliver a message $m'_1$ at position $k_1$. Then:
\begin{itemize}
\item $\pid[j]$ receives a message $\COMMIT(e'_1, k_1)$ sent by the leader of an
  epoch $e'_1$, and $\pid[j]$ has $\msg[k_1] = m'_1$ when receiving this message.
\end{itemize}
Then the leader of $e'_1$ has $\msg[k_1] = m'_1$ when sending the $\COMMIT$
message. Hence, by Lemma~\ref{thm:commit-msg}, $m'_1 = m_1$. We have thus shown
that $\pid[j]$ delivers $m_1$, as required.
\end{proof}

We next prove the first part of the Property~\ref{prop:liveness}a, ensuring that
a new configuration will be introduced. This is the most interesting part due to
its use of \asmref{asm:liveness-condition} (Availability).
\begin{theorem}\label{thm:livenessreconfiguration}
  Consider an execution with finitely many reconfiguration requests, and let $r$
  be the last reconfiguration request to be invoked. Suppose that $r$ is invoked
  by a correct process and no other reconfiguration requests take steps after
  $r$ is invoked. Then $r$ succeeds to introduce a configuration.
\end{theorem}
\begin{proof}
  The process $p_r$ executing $r$ starts by querying the configuration service
  to find the latest introduced configuration. Let $\eid$ be the epoch of this
  configuration. We first prove that the probing by $p_r$ eventually
  terminates. The probing procedure proceeds by iterations in epoch-descending
  order, starting by probing the members of $\eid$. The process $\pid[r]$ only
  moves to the next iteration after receiving at least one reply from a member
  of the epoch being probed, provided no process has replied with
  $\PROBEACK(\true,\eid+1)$. Consider an arbitrary epoch $\eid'$ such that
  $\eid'\leq \eid$. If $\pid[r]$ is probing the members of $\eid'$, then $p_r$
  has received a $\PROBEACK(\false, \eid+1)$ from at least one member of each
  epoch $e^*$ such that $e'<e^*\leq e$. Furthermore, because of
  line~\ref{vp_probe_set_newepoch} and the check in line~\ref{vp_new_state_pre},
  none of configurations with these epochs will ever become activated. Since no
  reconfiguration request other than $r$ takes steps after $r$ is invoked, no
  configuration with an epoch $>e$ can be introduced, and thus,
  activated. Then by \asmref{asm:liveness-condition} (Availability),
  at least one member $p_j$ of $\eid'$ is guaranteed to receive the
  $\PROBE(\eid+1, \eid')$ message sent by $\pid[r]$. Furthermore, since no
  configuration with an epoch $>e$ is introduced, no $\PROBE$ message for an epoch $>e+1$
  could have been issued. Then $p_j$ has $\newepoch\leq
  e+1$ when it receives $\PROBE(\eid+1, \eid')$. Therefore, by
  line~\ref{vp_probe_pre}, $p_j$ handles $\PROBE(\eid+1, \eid')$ and replies to $p_r$. Hence, for each epoch
  $\eid'$ that $p_r$ probes, either $p_r$ will eventually move to probe the
  previous epoch, or the probing phase will terminate. If the probing phase does
  not terminate for any epoch $>0$, then $p_r$ will eventually reach the initial
  epoch $0$. We have shown that $p_r$ will receive at least one reply from a
  member of $0$. The condition at line~\ref{vp_send_probe_ack} trivially holds
  for this member, and thus it will reply to $\pid[r]$ with
  $\PROBEACK(\true, \eid+1)$. Hence, the probing procedure is guaranteed to
  finish. After probing finishes, $\pid[r]$ attempts to store the new
  configuration into the configuration service. Since no reconfiguration request
  other than $r$ takes steps after $r$ is invoked, no configuration could have
  been introduced while $p_r$ was probing, and its {\tt compare\_and\_swap} will
  succeed, as required.
\end{proof}

\begin{proof}[Proof of Liveness.]
	Let $r$ be the last reconfiguration request such that
	$r$ is invoked by a correct process $\pid[r]$, and every other
        reconfiguration request does not take any steps after the $r$’s
        invocation.  
	By \thrmref{thm:livenessreconfiguration}, $\pid[r]$
        introduces a configuration $\confid = \aconf$. Assume that all
        processes in $\vm$ are correct. We first prove
        Property~\ref{prop:liveness-a} of Liveness.
	\begin{enumerate}[5a.]
		\item  	
                After introducing $\confid$, $\pid[r]$ sends $\NEWCONFIG(\eid,
                M)$ to $\pid$, the leader of $\confid$. 
                We assumed that no reconfiguration request but
                $r$ takes any steps after $r$’s invocation. Thus,
                no reconfiguration request may produce a
                $\PROBE(\eid',\_)$ such that $\eid'>\eid$. This
                implies that all the processes in $M$ have
                $\newepoch\leq \eid$ at any time. Furthermore, the fact that $\pid$ is
                picked as the leader guarantees that $\pid$ has
                handled $\PROBE(\eid,\_)$. Therefore, $\pid$ sets 
                $\newepoch=\eid$ before  $\pid[r]$ introduces
                $\confid$. Line~\ref{vp_probe_pre} guarantees that $\pid$ has
                $\newepoch=\eid$ when it receives $\NEWCONFIG(\eid,
                M)$.  Then $\pid$ invokes $\confchng(\confid)$ and
                sends $\NEWSTATE(\eid, \_, \_)$ to its 
                followers in $M$. We have established that all the processes in $M$ have
                $\newepoch\leq \eid$ at any time. Then when a follower
                receives a $\NEWSTATE(\eid, \_, \_)$ message, 
                the precondition in line~\ref{vp_new_state_pre} is satisfied and the follower 
                invokes $\confchng(\confid)$, as required.
       \end{enumerate}
        By Property~\ref{prop:liveness-a}, all processes
                  in $\vm$ deliver $\confchng(\confid)$. Therefore,
                  every process in $\vm$ sets $\epoch$ to $\eid$.
                  We assumed that no reconfiguration request but
                $r$ takes any steps after $r$’s invocation. Then no configuration with
                epoch $>\eid$ is ever introduced. Thus:
                (*) any process in $M$ has $\epoch=\eid$ at any moment
                after delivering $\confchng(\confid)$.
                We now use this fact to prove
                Properties~\ref{prop:liveness-b}
                and~\ref{prop:liveness-c} of Liveness.
       \begin{enumerate}[5a.]
         \setcounter{enumi}{1}
		\item Let $k_1=\max\{k\mid\msg[k]\neq\bot\}$ at the leader
                of $\eid$ when it invokes $\confchng(\confid)$. The
                leader of $\eid$ sends a $\COMMIT(\eid, k)$ to all
                members such that $k\leq k_1$ after it receives a
                 $\NEWSTATEACK(\eid)$ message from every follower.
                 Now assume that a process from $M$ broadcasts an application
                 message $m$ while in epoch $e$. This happens after the process
                 invokes $\confchng(\confid)$ and thus after it sets
                 $\leadervar = p_i$. Hence, the leader $\pid$ of
                 $\eid$ receives a $\FORWARD(\vid)$ message.
                 Let $k'$ be the value 
		of $\next$ at $\pid$ at the moment when it receives
                this message. Then $p_i$ sends $\ACCEPT(\eid, k',
                \vid)$. When processes in $M$ receive this message, they have
                already invoked $\confchng(\confid)$. Then by (*) and the fact that all
                processes in $M$ are correct, they reply with $\ACCEPTACK$s and 
                thus the leader eventually sends $\COMMIT(\eid, k')$. Furthermore, by
                lines~\ref{vp_broadcast_assign},~\ref{vp_broadcast_next}
                and~\ref{vp_quorum_of_new_state_acks_next_opt}, the leader $\pid$
                sends also $\ACCEPT(\eid, k'', \_)$ such that
                $k_1<k''<k'$ before sending $\ACCEPT(\eid, k',
                \vid)$. By (*) and the fact that all
                processes in $M$ are correct, the leader $\pid$
                eventually sends a
                $\COMMIT(\eid, k'', \_)$ message for all $k''$ such
                that $k_1<k''<k'$. Thus, we have established that the
                leader sends a $\COMMIT$ messages for all positions
                between $0$ and $k'$. Fix an arbitrary process $\pid[j] \in M$.

                After receiving the above $\COMMIT$ messages, by (*), $\pid[j]$
                handles $\COMMIT(\eid, k')$ and delivers the message stored at
                $\msg[k']$. If the leader $\pid$ sends $\COMMIT(\eid,
                k')$, then all followers handle  $\ACCEPT(\eid, k',
                \vid)$ before. Thus, $\pid[j]$ sets $\msg[k']=\vid$
                before handling $\COMMIT(\eid, k')$ either because it
                is the leader of $\eid$ and sends $\ACCEPT(\eid, k',
                \vid)$, or because it handles $\ACCEPT(\eid, k',
                \vid)$. The process $\pid[j]$ does not overwrite $\msg[k']$
                while $\epoch = \eid$, so that $\pid[j]$ has $\msg[k']=\vid$ when it handles
                $\COMMIT(\eid, k')$. Therefore, it delivers
                $\vid$. Since $\pid[j]$ was picked arbitrarily, all members of $M$ deliver $\vid$ as required.
                
              \item Assume that a process $\pid[k]$ delivers $\vid$ at an epoch
                $\eid'\leq\eid$. Then $\pid[k]$ handles $\COMMIT(\eid', k)$
                while having $\msg[k]=\vid$. This implies that the leader of
                $\eid'$ has $\msg[k]=\vid$ when it sends $\COMMIT(\eid', k)$.

                Consider an arbitrary process $\pid[j] \in M$. Let
                $\NEWSTATE(\eid, \msgi, \_)$ be the $\NEWSTATE$ message sent by
                the leader $\pid$ of $e$ and let $k_1=\length(\msgi)$. If
                $\eid'<\eid$, then by \invref{inv:prefix-of-higher-epoch},
                $\pid[j]$ has $\msg[k]=\vid$ while in $\eid$. Then
                $\msgi[k]=\vid$ and $k_1\geq k$. Assume now that
                $\eid'=\eid$. Since $\pid[k]$ delivers $\vid$ at an epoch
                $\eid'$, this process has $\msg[k]=\vid$ when it handles
                $\COMMIT(\eid, k)$. The case when $\pid[k]$ sets $\msg[k]$ to
                $\vid$ due to sending or receiving $\ACCEPT(\eid, k, \vid)$ is
                handled as in Property~\ref{prop:liveness-b}. Hence, we can
                assume that $\pid[k]$ sends or receives
                $\NEWSTATE(\eid, \msgi, \_)$ such that $\msgi[k]=\vid$. In this
                case we have $k_1\geq k$. Thus, from now on we can assume that
                $\msgi[k]=\vid$ and $k_1\geq k$.

                The leader of $\eid$ sends a $\COMMIT$ message to all members of
                $M$ for all positions $k_1\geq k$. Thus, by (*), $\pid[j]$
                eventually delivers the message at position $k$ in its log. If
                this message $\vid'$ is different from $\vid$, then $\pid[j]$
                handles $\COMMIT(\eid^*, k)$ for some $\eid^*$ while having
                $\msg[k] = \vid'$. Then the leader of $\eid^*$ must have had
                $\msg[k] = \vid'$ when it sent the $\COMMIT$ message. But then
                by Lemma~\ref{thm:commit-msg}, $\vid=\vid'$. Thus, $\pid[j]$
                eventually delivers $\vid$, as required.
	\end{enumerate}
\end{proof}

%% file: pr-proof.tex
\section{Proof of Correctness for the Passive Replication Implementation
  (Theorem~\ref{thm:correctness-passive})}
\label{sec:pr:correct}

A state $S'$ is {\em derived}\/ from a state $S$ using
a command $c$, denoted $S \stackrel{c}{\leadsto} S'$, 
iff $\exists \delta.\, S' = \app(S, \delta) \wedge 
S \stackrel{c}{\rightarrow} \langle \_, \delta \rangle$.
Given a sequence of commands $\sigma = \sigma_1\ldots\sigma_l$,
a state $S'$ is derived from a state $S$ using $\sigma$, 
denoted $S \stackrel{\sigma}{\leadsto} S'$, iff
$S=S_0 \stackrel{\sigma_1}{\leadsto} S_1 \stackrel{\sigma_2}{\leadsto} 
\dots \stackrel{\sigma_l}{\leadsto} S_{l} = S'$ for some
states $S_1\ldots S_{l-1}$. A state $S'$ is derived from 
a state $S$, denoted $S \leadsto S'$, iff 
$S \stackrel{\sigma}{\leadsto} S'$ for 
some sequence of commands $\sigma$. %
For two sequences $\sigma$ and $\sigma'$, we write $\sigma \preceq \sigma'$
if $\sigma$ is a prefix of $\sigma'$. 

In the following, we will use the term \emph{time} to refer to the index
at which a particular event occurs in a history.
The lemma below asserts the key invariants satisfied
by the passive replication protocol in Figure~\ref{fig:app-layer}:
\begin{lemma}\label{thm:pr}
Let $h$ be a finite history of the algorithm in Figure~\ref{fig:app-layer}.
Then, there exists a sequence of commands $\sigma_h$ over the set $\{c \mid \execute(c) \in h \}$
and a non-decreasing function
$f^j_h: \{0..|h|\} \rightarrow \{0..|h|\}$ for each processes $p_j$
such that $|\sigma_h|=\max_j\{f_j(|h|)\}$, and
the following holds for all times $t \in \{0..|h|\}$:
\begin{enumerate}

\item \label{thm:pr:main}
For all processes $p_i$,
$\Sigma_i(t)$ is derived from $\Sigma_0$ using
the prefix of $\sigma_h$ of size $f_i(t)$:\\
$\forall p_i.\, \Sigma_0 \stackrel{\prefix{\sigma_h}{k}}{\leadsto} \Sigma_i(t) \where k=f_h^i(t)$.

\item \label{thm:pr:mnum}
For all processes $p_i$, the total number of messages delivered by $p_i$
at or before $t$ is equal $f_i(t)$, and if $t > 0$, 
then $f_i(t-1) \le f_i(t) \le f_i(t-1)+1$.

\item \label{thm:pr:context} For all processes $p_i$,
if $p_i$ delivers a message $m$ at $t$, 
then at the time when $m$ was broadcast, the value of the speculative state $\Theta$ 
at the sender of $m$ is the same as 
the value of the committed state $\Sigma$ at $p_i$ at $t$:\\
 $\forall p_i, m.\, h_t = \deliv_i(m) \implies
 \exists p_j, t'.\, t' < t \wedge h_{t'}=\bcast_j(m) \wedge \Sigma_i(t) = \Theta_j(t')$.

\item \label{thm:pr:context21} Consider a message $m$ that was originally
  broadcast by a process $p_j$ in an epoch $e_j$. Suppose that a process $p_k$
  joins an epoch $e_k > e_j$ at $t$, and $p_k$ is the leader of $e_k$. Suppose further
  that $p_k$ does not speculatively deliver any messages when joining $e_k$ and
  $m$ is the last message delivered by $p_k$ before joining $e_k$. Then
  $\Theta_k(t) = \Theta_j(t_j)$.

\item \label{thm:pr:context22} Consider a message $m$ that was originally
  broadcast by a process $p_j$ in an epoch $e_j$, and assume that some process delivers
  $m$. Suppose that a process $p_k$ joins an epoch $e_k > e_j$ at $t$, and $p_k$ is the
  leader of $e_k$. Suppose further that $m$ is the last message speculatively
  delivered by $p_k$ when joining $e_k$. Then $\Theta_k(t) = \Theta_j(t_j)$.

\end{enumerate}
\end{lemma}

The following proposition is immediate from the passive replication code:
\begin{proposition}
\begin{enumerate}
\item
For all processes $p_i$, times $t>0$, and messages $m = \langle \_, \_, \delta \rangle$,
if $p_i$ delivers $m$ at time $t$, then $\Sigma_i(t) = \app(\Sigma_i(t-1), \delta)$.

\item
For all processes $p_i$ and times $t>0$, if $\Sigma_i(t) \neq \Sigma_i(t-1)$, then
$p_i$ delivers a message at $t$.
\end{enumerate}
\label{prop:state-deliver}
\end{proposition}

\begin{lemma}
  Consider a process $p_i$ that delivers $m$ before $m'$ and does not deliver
  any message in between the two. Suppose that $m$ is broadcast in an epoch $e$
  and $m'$ is broadcast in an epoch $e' > e$. Let $p_k$ be the process that
  broadcast $m'$ and assume that $p_k$ delivers $m$ before joining $e'$. Then
\begin{enumerate}

\item \label{noinbetween-deliv} 
$p_k$ does not deliver any messages after delivering $m$ and before
broadcasting $m'$; and

\item \label{noinbetween-spec}
$p_k$ does not speculatively deliver any messages when joining $e'$.

\end{enumerate}
\label{lem:aux-deliver}
\end{lemma}

\begin{proof}
Let $t_k$ be the time at which $p_k$ broadcasts $m'$, and $s_k < t_k$ be the time 
at which $p_k$ delivers $m$.

Assume by contradiction that either 
(\ref{noinbetween-deliv}) or (\ref{noinbetween-spec}) is violated.
Suppose first that (\ref{noinbetween-deliv}) does not hold: i.e., 
there exists a time $s_k < t_k'' < t_k$ such that
$p_k$ delivers a message $m''$ at time $t_k''$.
Since $p_k$ broadcasts $m'$ and
$p_i$ delivers $m'$, by \Consistency{} ((a)/if), $p_i$
must deliver $m''$ before delivering $m'$. 
Since $s_k < t_k''$, by
Total Order, $p_i$ must also deliver $m''$ after delivering
$m$. However, since $p_i$ does not deliver any messages after
delivering $m$ and before delivering $m'$, this is a contradiction.

Suppose next that (\ref{noinbetween-spec}) does not hold. Let
$\confchng_k(\langle e', \_, \_ \rangle, \sigma)$ be
the configuration change event that causes $p_k$ to join
$e'$ and $s_k < t_k'' < t_k$ be the time at which $p_k$
joins $e'$. Then, there exists a message $m''$  such 
that $m''\in \sigma$. Since $p_i$ delivers $m'$, $p_k$ broadcasts $m'$, and
$m''\in \sigma$,
by \Consistency{} ((a)/if), $p_i$ must deliver $m''$ before
delivering $m'$. Thus, we get that $p_i$ delivers $m''$, $p_k$
speculatively delivers $m''$ when joining $e_k$, and
$p_k$ delivers $m$ before joining $e_k$. 
By \Consistency{} ((b)/if), this implies that 
$p_i$ delivers $m$ before delivering $m''$.
However, since $p_i$ does not deliver any messages
after delivering $m$ and before delivering $m'$, 
this is a contradiction.
\end{proof}

\begin{lemma}
  Consider a process $p_i$ that delivers $m$ before $m'$ and does not deliver
  any message in between the two.  Suppose that $m$ is broadcast in an epoch $e$
  and $m'$ is broadcast in an epoch $e' > e$. Let $p_k$ be the process that
  broadcasts $m'$ and assume that $p_k$ speculatively delivers $m$ when joining
  $e'$. Then $m$ is the last message speculatively delivered by $p_k$ when
  joining $e'$.
\label{lem:aux-deliver2}
\end{lemma}

\begin{proof}
Let $t'$ be the time at which $p_k$ joins $e'$, and 
$\confchng_k(\langle e', \_, \_ \rangle, m_1..m_l)$
be the configuration change event occurring at $t'$. 
Since $p_k$ speculatively delivers $m$ when joining $e'$,
there exists $1 \le r \le l$ such that $m_r = m$.
Assume by contradiction that $1 \le r < l$.
Then, there exists a message $m''$ such that $p_k$ speculatively
delivers $m$ before $m''$ when joining $e_k$. 
By \Consistency{} ((a)/if), $p_i$ delivers $m''$ before
$m'$, which by \Consistency{} ((b)/if), implies that $p_i$
delivers $m$ before delivering $m''$. Thus, $p_i$
delivers $m$ followed by $m''$ followed by $m'$, which 
is a contradiction to the assumption that $p_i$
does not deliver any messages after delivering $m$
and before delivering $m'$. 
\end{proof}

\begin{lemma}
  Consider a process $p_i$ that delivers $m$ before $m'$ and does not deliver
  any message in between the two.  Suppose that $m$ is broadcast in an epoch $e$
  and $m'$ is broadcast in an epoch $e' > e$. Let $p_k$ be the process that
  broadcasts $m'$. Then $p_k$ does not broadcast any messages after
  joining $e'$ and before broadcasting $m'$.
\label{lem:aux-deliver3}
\end{lemma}

\begin{proof}
Assume by contradiction that
there exists a message $m''$ such that $p_k$ broadcasts $m''$
after joining $e'$ and before broadcasting $m'$. By Local Order,
$p_i$ must deliver $m''$ before delivering $m'$. Since $p_i$
does not deliver any messages after delivering $m$ and
before delivering $m'$, $p_i$ must deliver $m''$ before
delivering $m$. By \Consistency{} ((b)/only if), this implies that
$p_k$ delivers $m''$ before joining $e'$ or speculatively delivers
$m''$ before $m$ when joining $e'$. However, this means that $m''$
is delivered or speculatively delivered by $p_k$ before it is broadcast
by $p_k$, which is a contradiction. 
\end{proof}

\begin{proof}[Proof of Lemma~\ref{thm:pr}]
By induction on the length of an execution. Consider first an execution $h$
such that $|h|=0$. Then, all processes are in their initial states.
Let $\sigma_h$ be the empty sequence, and $f_h^j: \{0\} \rightarrow \{0\}$
for each process $p_j$. Then, 
$|\sigma_h|=\max_j\{f_j(|h|)\}=0$, and
for all processes $p_i$, $\Sigma_i=\Sigma_0$. Since  
$\Sigma_0$ is derived from itself using the empty sequence of commands,
and the empty sequence is the prefix of $\sigma_h$ 
of size $f^i_h(0)=0$, Lemma~\ref{thm:pr}.\ref{thm:pr:main} holds.
Furthermore, since no process delivers a message in $h$ and $t=0$,
Lemma~\ref{thm:pr}.\ref{thm:pr:mnum} is vacuously true.
Finally, since no process broadcasts a message or
joins an epoch in the history of size $0$, the claims asserted by 
Lemmas~\ref{thm:pr}.\ref{thm:pr:context}-\ref{thm:pr}.\ref{thm:pr:context22}
are vacuously true as well.

For the inductive step, assume that the lemma holds for
all histories $h$ of length $k$, and consider a history $h'$ of length $k+1$.
If the last event in $h'$ is the delivery of a message 
$m=\langle \_, \_, \delta \rangle$ and the process $p_i$
that delivers $m$ has the longest prefix of messages delivered in $h$, then
we let
\begin{align}
&\sigma_{h'} \triangleq \sigma_h \cdot c \ \ \ \suchthat \exists S.\, \exists c.\,
\execute(c) \in h \wedge 
S \stackrel{c}{\rightarrow} \langle \_, \delta \rangle \wedge
\Sigma_0 \stackrel{\sigma_h}{\leadsto} S \text{,~and} \label{eq:sigma-grow}\\
&f_{h'}^i \triangleq f_h^i \cup \{(|h|+1, |\sigma_h|+1)\} \wedge
\forall p_j\neq p_i.\, f_{h'}^j \triangleq  f_h^j \cup \{(|h|+1, f^j_h(|h|))\}. \label{eq:f-grow}
\end{align}
If the last event in $h'$ is the delivery of a message 
$m=\langle \_, \_, \delta \rangle$ by a process $p_i$, and the 
length of the message prefix delivered by $p_i$ in $h$ is not the longest, then
we let:
\begin{align}
& \sigma_{h'} \triangleq \sigma_h \text{,~and} \label{eq:sigma-same}\\
&  f_{h'}^i \triangleq f_h^i \cup \{(|h|+1, f_h^i(|h|)+1)\} \wedge
\forall p_j\neq p_i.\, f_{h'}^j \triangleq f_h^j \cup \{(|h|+1, f^j_h(|h|))\}. \label{eq:f-same}
\end{align}
In all the remaining cases, we let:
\begin{align}
& \sigma_{h'} \triangleq \sigma_h \text{,~and} \label{eq:sigma-same2}\\
& \forall p_j.\, f_{h'}^j \triangleq f_h^j \cup \{(|h|+1, f^j_h(|h|))\}. \label{eq:f-same2}
\end{align}
We first show that
\begin{equation}
|\sigma_{h'}|=\max_j\{f_{h'}^j(|h'|)\}.
\label{eq:len-bound}
\end{equation}
If the last event in $h'$ is not the delivery of a message, then 
by~(\ref{eq:sigma-same2}), (\ref{eq:f-same2}), and the induction hypothesis:
$$
|\sigma_{h'}|=|\sigma_h|=\max_j\{f_{h}^j(|h|)\}=\max_j\{f_{h'}^j(|h'|)\},
$$
as needed. Next, assume that the last event in $h'$ is the
delivery of a message by process $p_i$, which has the longest
prefix of messages delivered in $h$.
Then by~(\ref{eq:f-grow}) and the induction hypothesis
$$
f_{h'}^i(|h'|)=|\sigma_h|+1=\max_j\{f_{h}^j(|h|)\}+1 \wedge 
\forall p_j\neq p_i.\, f_{h'}^j(|h'|)=f_{h}^j(|h|)\le \max_j\{f_{h}^j(|h|)\},
$$
which implies that
$$
f_{h'}^i(|h'|)=\max_j\{f_{h'}^j(|h'|)\}.
$$
Thus, by~(\ref{eq:sigma-grow}), we have 
$$
|\sigma_{h'}|=|\sigma_h|+1=f_{h'}^i(|h'|)=\max_j\{f_{h'}^j(|h'|)\},
$$
as needed.
Last, assume that the last event in $h'$ is the
delivery of a message by process $p_i$ which does not have the longest
prefix of messages delivered in $h$.
Let $l_h^k$ be the length of the message prefix delivered by process $p_k$
in $h$. By the induction hypothesis for Lemma~\ref{thm:pr}.\ref{thm:pr:mnum},
$$
\forall p_j.\, f_h^j(|h|)=l_h^j,
$$
and therefore, 
$$
f_h^i(|h|)=l_h^i < \max_j\{l_h^j\}.
$$
By~(\ref{eq:sigma-same}), and the induction hypothesis,
$$
|\sigma_{h'}|=|\sigma_h|=\max_j\{l_h^j\}=\max_j\{f_h^j(|h|)\}>l_h^i=f_h^i(|h|).
$$
Thus by~(\ref{eq:f-same}), we have
$$
f_{h'}^i(|h'|)=f_h^i(|h|)+1\le \max_{j\neq i}\{f_h^j(|h|)\}=\max_{j\neq i}\{f_{h'}^j(|h'|)\}.
$$
The above implies that
$$
\max_j\{f_{h'}^j(|h'|)\}=\max_j\{f_{h}^j(|h|)\}=|\sigma_h|=|\sigma_{h'}|,
$$
as needed.
We next prove the individual statements of the lemma assuming 
$t=k+1$ (i.e., $t$ is the time of the last event in $h'$).

\smallskip
\smallskip

\textit{Proof of Lemma~\ref{thm:pr}.\ref{thm:pr:context21}}. 
Suppose that at $t$, a process $p_k$ joins an epoch $e_k > e_j$, and $p_k$
is the leader of $e_k$. Let $m$ be a message that was originally broadcast 
by a process $p_j$ in an epoch $e_j$ at time $t_j$.  Suppose further that $p_k$ does not speculatively
deliver any messages when joining $e_k$ and $m$ is the last message delivered by
$p_k$ before joining $e_k$. We prove that $\Theta_k(t) = \Theta_j(t_j)$.

Let $s_k$ be the time at which $p_k$ delivers $m$. Since $m$ is the last message
delivered by $p_k$ before joining $e_k$, $\Sigma_k$ does not change after it was
updated by the code in line~\ref{passive:apply} and before $p_k$ joins $e_k$.
By the $\confchng$ handler code, $\Theta_k(t)$ is computed by first assigning
it the value of $\Sigma_k(t)$ (line~\ref{pr:code:tstate-assign}), and then
applying the state updates speculatively delivered in the $\confchng$
upcall. Since by assumption, no state updates are speculatively delivered at
$t$, $\Sigma_k(t) = \Sigma_k(s_k)$, and so we have
$\Theta_k(t) = \Sigma_k(s_k)$.  By the induction hypothesis for
Lemma~\ref{thm:pr}.\ref{thm:pr:context}, $\Sigma_k(s_k) = \Theta_j(t_j)$, and
therefore, $\Theta_k(t) = \Theta_j(t_j)$, as needed. \qed

\smallskip
\smallskip

\textit{Proof of Lemma~\ref{thm:pr}.\ref{thm:pr:context22}}. 
Consider a message $m$
that was originally broadcast by a process $p_j$ at time $t_j$
in an epoch $e_j$. We consider two cases:

\renewcommand{\theenumi}{(\alph{enumi})}
\begin{enumerate}
  \item \label{context22:1}
Suppose first that at $t$, there exists a process process $p_k$ that joins an epoch
$e_k > e_j$, and $p_k$ is the leader of $e_k$. Suppose further that $m$ is the
last message speculatively delivered by $p_k$ when joining $e_k$. If no process delivers
$m$ before $t$, then the required holds by the induction hypothesis. Suppose that
there exists a process $p_i$ that delivers $m$ at time  $t_i < t$. We prove that
$\Theta_k(t) = \Theta_j(t_j)$. 

Let $\confchng_k(\langle e_k, \_, \_ \rangle, m_1..m_l)$ be the configuration
change action executed by $p_k$ at $t$.  Since $m$ is the last message
speculatively delivered by $p_k$ when joining $e_k$ (i.e., $m_l = m$), 
by \Consistency{} ((b)/if),
\begin{multline}
\forall r.\, 
1 \le r \le l \implies p_i \text{~delivers~} m_r\ \wedge\\
1 < r \le l \implies p_i \text{~delivers~} m_{r-1} \text{~before~} m_r\ \wedge\\
\forall m'.\, p_i \text{~does not deliver~} m' \text{~after~} m_{r-1}
\text{~and~before~} m_r.
\label{eq:spec-msg}
\end{multline}
Let $t_1$ be the
time at which $p_i$ delivers $m_1 = \langle \_, \_, \delta_1 \rangle$, 
and $\Sigma_i(t_1)^-$ be the value of $\Sigma_i$ to which
$\delta_1$ is applied in line~\ref{passive:apply}. 
We show that 
\begin{equation}
\Sigma_i(t_1)^- = \Sigma_k(t).
\label{eq:equal-states}
\end{equation}
Assume by contradiction that $\Sigma_i(t_1)^- \neq \Sigma_k(t)$.  Then, by the
induction hypothesis for Lemma~\ref{thm:pr}.\ref{thm:pr:main},
there exists a message $m'$ such that
either \emph{(i)} $p_i$ delivers $m'$ before delivering $m_1$ and $p_k$ does not
deliver $m'$ before joining $e_k$, or \emph{(ii)} $p_k$ delivers $m'$ before
joining $e_k$ and $p_i$ does not deliver $m'$ before delivering $m_1$.  Suppose
that $\emph{(i)}$ holds. Since $p_k$ does not deliver $m'$ before joining $e_k$,
by \Consistency{} ((b)/only if), $p_k$ speculatively delivers $m'$ before $m_1$
when joining $e_k$. However, since $m_1$ is the first message speculatively
delivered by $p_k$ when joining $e_k$, this is a contradiction.  Suppose that
$\emph{(ii)}$ holds. Then, by \Consistency{} ((b)/if), $p_i$ must deliver $m'$
before delivering $m_1$ contradicting the assumption that $p_i$ does not deliver
$m'$ before delivering $m_1$.  Thus,~(\ref{eq:equal-states}) holds, and
from~(\ref{eq:equal-states}) and~(\ref{eq:spec-msg}), and the code in
line~\ref{pr:code:state-assign}, we get $\Theta_k(t) = \Sigma_i(t_i)$.  By the
induction hypothesis for Lemma~\ref{thm:pr}.\ref{thm:pr:context},
$\Sigma_i(t_i) = \Theta_j(t_j)$, and therefore, $\Theta_k(t) = \Theta_j(t_j)$,
as needed.

\item
Suppose next that at $t$, there exists a process $p_i$ which delivers a message $m$
that was originally broadcast by a process $p_j$ in epoch $e_j$. If $m$ is not
the last message speculatively delivered by any process when joining an epoch $>e_j$ before $t$,
then the required holds by the induction hypothesis. Otherwise, let $p_k$ be a process
that speculatively delivers $m$ when joining an epoch $e_k > e_j$ at a time $t_k < t$, and
assume that $m$ is the last message delivered by $p_k$ when joining $e_k$. Then instantiating
the argument in~\ref{context22:1} with $t=t_k$, we obtain $\Theta_k(t_k) = \Theta_j(t_j)$,
as needed.
\end{enumerate}
\qed

\smallskip
\smallskip

\textit{Proof of Lemma~\ref{thm:pr}.\ref{thm:pr:context}}.
Let $m'$ be the message delivered by $p_i$ at $t$.
Consider the latest time $s \le T$ at which $p_i$ delivers a message, and let
$m$ be the message delivered by $p_i$ at $s$. Thus,
\begin{equation}
\Sigma_i(t) = \app(\Sigma_i(s), \delta) \wedge m' = \langle \_, \_, \delta \rangle.
\label{eq:sigma-after}
\end{equation}
Let $p_k$ be the processes that broadcast $m'$ at time $t_k < t$. We prove that
$$
\Sigma_i(t) = \Theta_k(t_k).
$$
By the induction hypothesis for Lemma~\ref{thm:pr}.\ref{thm:pr:context},
there exists a process $p_j$ which broadcast $m$ at time $t_j < s$ and
\begin{equation}
\Sigma_i(s) = \Theta_j(t_j).
\label{eq:pr:context:ic}
\end{equation}
Let $e_j = \epochOf(\bcast_j(m))$ and 
$e_k=\epochOf(\bcast_k(m'))$. 
By definition of $\epochOf$, 
$\confchng_j(\langle e_j, \_, \_ \rangle,\_)$ 
(respectively, $\confchng_j(\langle e_k, \_, \_ \rangle,\_)$) is the 
latest configuration change event preceding $\bcast_j(m)$ (respectively, $\bcast_k(m')$)
at $p_j$ (respectively, $p_k$). 
By line~\ref{line:curleader-set}, $\curleader_j(t_j) = \primary(e_j)$
and $\curleader_k(t_k) = \primary(e_k)$. Since a process can only broadcast
a message if the precondition in line~\ref{line:curleader-pre} holds, 
we have $\curleader_j(t_j) = p_j$ and $\curleader_k(t_k) = p_k$. Hence, 
$$
p_j = \primary(e_j) \wedge p_k = \primary(e_k).
$$
We consider two cases. Suppose first that $e_j = e_k$. Then
$$
p_j=p_k=\primary(e_j)=\primary(e_k). 
$$
By Local Order, $t_j < t_k$ and 
for all times $t_j < t_j' < t_k$, $p_j$ does not broadcast any messages 
at $t_j'$. Thus, 
by the code in 
lines~\ref{passive:exec}--\ref{rep_bcast}, we have
$$
\exists c.\, \Theta_k(t_k) = \app(\Theta_j(t_j), \delta) \wedge 
\Theta_j(t_j) \stackrel{c}{\rightarrow} \langle \_, \delta \rangle \wedge
m' = \langle \_, \_, \delta \rangle.
$$
By~(\ref{eq:pr:context:ic}), the above implies 
$$
\Theta_k(t_k) = \app(\Sigma_i(s), \delta) \wedge 
m' = \langle \_, \_, \delta \rangle.
$$
Thus, by~(\ref{eq:sigma-after}), we have
$$
\Sigma_i(t) = \Theta_k(t_k),
$$
as needed.

Suppose next that $e_j \neq e_k$. Since $p_i$ delivers $m$ before $m'$,
by \Consistency{} ((a)/only if),
either $p_k$ delivers $m$ before joining $e_k$, or 
$p_k$ speculatively delivers $m$ when joining $e_k$.
Let $t_k'$ be the time at which $p_k$ joins $e_k$.
Suppose first that $\deliv_k(m)$ occurs at time
$s_k < t_k'$. 
Since by Global Order, $e_j < e_k$, by Lemma~\ref{lem:aux-deliver},
$p_k$ does not deliver any messages after delivering $m$ and
before joining $e_k$
and does not speculatively deliver any messages 
when joining $e_k$. Thus, by the induction hypothesis 
for Lemma~\ref{thm:pr}.\ref{thm:pr:context21},
\begin{equation}\label{ind-hyp-context21}
\Theta_k(t_k') = \Theta_j(t_j).
\end{equation}
Since $p_i$ does not deliver any messages after delivering
$m$ and before delivering $m'$, by Local Order, $p_k$
does not broadcast any messages in-between $t_k'$ and $t_k$.
Together with~(\ref{ind-hyp-context21}), this implies 
$$
\exists c.\, \Theta_k(t_k) = \app(\Theta_j(t_j), \delta) \wedge 
\Theta_j(t_j) \stackrel{c}{\rightarrow} \langle \_, \delta \rangle
\wedge m' = \langle \_, \_, \delta \rangle.
$$
By~(\ref{eq:pr:context:ic}), the above implies 
$$
\Theta_k(t_k) = \app(\Sigma_i(s), \delta) \wedge 
m' = \langle \_, \_, \delta \rangle.
$$
Thus, by~(\ref{eq:sigma-after}), we obtain,
$$
\Sigma_i(t) = \Theta_k(t_k),
$$ 
as needed.

Suppose that $p_k$ speculatively delivers $m$ when joining $e_k$. 
Then, by Lemma~\ref{lem:aux-deliver2}, 
$m$ is the last message speculatively delivered by $p_k$
when joining $e_k$. Hence, by the induction hypothesis 
for Lemma~\ref{thm:pr}.\ref{thm:pr:context22},
\begin{equation*}
\Theta_k(t_k') = \Theta_j(t_j).
\label{eq:specdel}
\end{equation*}
Since by Lemma~\ref{lem:aux-deliver3},
$p_k$ does not broadcast any messages after joining
$e_k$ and before broadcasting $m'$, the above implies 
$$
\exists c.\, \Theta_k(t_k) = \app(\Theta_j(t_j), \delta) \wedge 
\Theta_j(t_j) \stackrel{c}{\rightarrow} \langle \_, \delta \rangle
\wedge m' = \langle \_, \_, \delta \rangle.
$$
By~(\ref{eq:pr:context:ic}), the above implies 
$$
\Theta_k(t_k) = \app(\Sigma_i(s), \delta) \wedge 
m' = \langle \_, \_, \delta \rangle.
$$
Thus, by~(\ref{eq:sigma-after}), we obtain,
$$
\Sigma_i(t) = \Theta_k(t_k),
$$ 
as needed. \qed

\smallskip
\smallskip

\textit{Proof of Lemma~\ref{thm:pr}.\ref{thm:pr:mnum}}
Let $l_j$ denote the length of the message prefix delivered by 
the process $p_j$. By~(\ref{eq:f-grow}), (\ref{eq:f-same}), and (\ref{eq:f-same2}),
for all processes $p_j$ and times $t' < t$, $f^j_{h'}(t')=f^j_h(t')$. Since by the induction
hypothesis, 
$$
\forall t' < t.\, l_j(t')=f^j_h(t') \wedge \forall 0<t'<t.\, f_h^j(t'-1) \le f^j_h(t') \le f_h^j(t'-1)+1, 
$$
we have
$$
\forall t' < t.\, \forall p_j.\, f^j_{h'}(t')=l_j(t').
$$
Thus, it remains to show that
\begin{equation}
\forall p_j.\, f^j_{h'}(t)=l_j(t) \wedge f^j_{h'}(t-1) \le f^j_{h'}(t) \le f^j_{h'}(t-1)+1.
\label{eq:mnum-to-show}
\end{equation}
We consider the following three cases. First, if the last even in $h'$ is not a delivery
of a message, then $\forall p_j.\, l_j(t-1)=l_j(t)$. Thus, by~(\ref{eq:f-same}) and 
the induction hypothesis 
$$
\forall p_j.\, f^j_{h'}(t)=f_h^j(t-1)=l_j(t-1)=l_j(t), 
$$
and therefore,~(\ref{eq:mnum-to-show}) holds.
Second, suppose that the last event in $h'$ is the delivery of a message $m$ by process
$p_i$ such that the following holds:
$$
l_i(t-1)=\max_j\{l_j(t-1)\}.
$$
Then, 
$$
l_i(t)=l_i(t-1)+1 \wedge \forall p_j\neq p_i.\, l_j(t-1)=l_j(t).
$$
By~(\ref{eq:f-grow}),
$$
\forall p_j\neq p_i.\, f_{h'}^j(t)=f_h^j(t-1)=l_j(t-1)=l_j(t).
$$
Hence,~(\ref{eq:mnum-to-show}) holds for all $p_j\neq p_i$.
Otherwise, by~(\ref{eq:f-grow}),
$$
f_{h'}^i(t)=|\sigma_h|+1,
$$
which by the induction hypothesis for~(\ref{eq:len-bound}) implies that
$$
f_{h'}^i(t)=\max_j\{f_{h}^j(|h|)\}+1=\max_j\{l_j(t-1)\}+1=l_i(t-1)+1=l_i(t).
$$
In addition, the above and the induction hypothesis also yield
$$
f_{h'}^i(t)=l_i(t-1)+1=f_h^i(t-1)+1,
$$
and so~(\ref{eq:mnum-to-show}) follows.
Finally, suppose that the last event in $h'$ is the delivery of a message $m$ by process
$p_i$ such that the following holds:
$$
l_i(t-1)<\max_j\{l_j(t-1)\}.
$$
By~(\ref{eq:f-same}),
$$
\forall p_j\neq p_i.\, f_{h'}^j(t)=f_h^j(t-1)=l_j(t-1)=l_j(t).
$$
Otherwise,
$$
f_{h'}^i(t)=f_h^i(t-1)+1=l_i(t-1)+1=l_i(t),
$$
as needed.

\smallskip
\smallskip

\textit{Proof of Lemma~\ref{thm:pr}.\ref{thm:pr:main}}. 
Let $m'$ be the message delivered by $p_i$ at $t$. Assume that $m'$
is the $l$th message delivered by $p_i$. Let
$m_1..m_l=m'$ be the sequence of the first $l$ messages delivered
by $p_i$, and $t_j^k$ denote the time at which the $k$th
message is delivered by $p_j$. (Note that $t = t_i^l$.) 
Let 
$$
m_l=\langle\_,\_,\delta_l \rangle.
$$
Suppose first that for all processes $p_j \neq p_i$,
the number of messages delivered by $p_j$ before $t$ is $<l$.
Then, we have to show that the following holds provided
$\sigma_{h'}$ and the functions $f_{h'}^j$ for each process $p_j$
are chosen as shown in (\ref{eq:sigma-grow})~and~(\ref{eq:f-grow}):
\begin{align}
&\forall t' \le t.\, \forall p_j.\, \exists k.\,
\Sigma_0 \stackrel{\prefix{\sigma_{h'}}{k}}{\leadsto} \Sigma_j(t') \wedge k=f_{h'}^j(t').
\label{eq:to-prove}
\end{align}
By the induction hypothesis, we have
\begin{align}
&\forall t' < t.\, \forall p_j.\, \exists k.\,
\Sigma_0 \stackrel{\prefix{\sigma_{h}}{k}}{\leadsto} \Sigma_j(t') \wedge k=f_{h}^j(t').
\label{eq:ih1}
\end{align}
Since
$$
\forall c.\, \forall k.\, k \le |\sigma_h| \implies \prefix{\sigma_h}{k} = \prefix{(\sigma_h \cdot c)}{k} = 
\prefix{\sigma_{h'}}{k},
$$
and from~(\ref{eq:f-grow}),
$$
\forall t' < t.\, f_{h'}(t')=f_h(t') \wedge f_{h'}(t') \le |\sigma_h|,
$$
we get
$$
\forall t' < t.\, \forall c.\, \exists k.\, (\prefix{\sigma_h}{k} = \prefix{(\sigma_h \cdot c)}{k} \wedge k=f_{h'}(t')).
$$
Thus,~(\ref{eq:ih1}) implies
\begin{align*}
&\forall t' < t.\, \exists k.\, \Sigma_0 \stackrel{\prefix{\sigma_{h'}}{k}}{\leadsto} \Sigma_j(t') 
\wedge k=f_{h'}^j(t').
\end{align*}
Furthermore, since
$$
\forall p_j.\, p_j \neq p_i \implies \Sigma_j(t-1)=\Sigma_j(t),
$$
we also have
\begin{align*}
&\forall t' \le t.\, \forall p_j.\, (p_j \neq p_i \implies
\exists k.\, \Sigma_0 \stackrel{\prefix{\sigma_{h'}}{k}}{\leadsto} \Sigma_j(t') \wedge k=f_{h'}^j(t'))\ \wedge\\
&\forall t' < t.\, \exists k.\, \Sigma_0 \stackrel{\prefix{\sigma_{h'}}{k}}{\leadsto} \Sigma_i(t') 
\wedge k=f_{h'}^j(t').
\end{align*}
Thus, it remains to show that
\begin{align*}
&\Sigma_i(t) = \app(\Sigma_i(t-1), \delta) \wedge
\exists c.\, \Sigma_i(t-1) \stackrel{c}{\rightarrow} \langle \_, \delta_l \rangle,
\end{align*}
which by Proposition~\ref{prop:state-deliver} can be rewritten as
\begin{equation}
\Sigma_i(t_i^l) = \app(\Sigma_i(t_i^{l-1}), \delta_l) \wedge
\exists c.\, \Sigma_i(t_i^{l-1}) \stackrel{c}{\rightarrow} \langle \_, \delta_l \rangle.
\label{eq:find-c0}
\end{equation}
Since $\Sigma_i(t_i^l) = \app(\Sigma_i(t_i^{l-1}), \delta)$, we only need to 
show that
\begin{equation}
\exists c.\, \Sigma_i(t_i^{l-1}) \stackrel{c}{\rightarrow} \langle \_, \delta_l \rangle.
\label{eq:find-c}
\end{equation}
Let $p_k$ be the process that broadcast $m_{l-1}$, and $s_k$ and $e_k$
be the time and the epoch at which this happened, respectively. 
By the induction hypothesis for Lemma~\ref{thm:pr}.\ref{thm:pr:context},
\begin{equation}
\Sigma_i(t_i^{l-1}) = \Theta_k(s_k).
\label{eq:sigma-theta}
\end{equation}
Let $s_k'$ and $e_k'$ be respectively the time and the epoch when $m_l$ is broadcast.
Suppose first that $e_k = e_k'$. Since the leader of $e_k$
is the only process that can broadcast messages in $e_k$,
it also broadcasts $m_{l-1}$. Since $m_{l-1}$ is delivered by $p_i$ before $m_l$ and no
messages are delivered by $p_i$ in-between $m_{l-1}$ and $m_l$, by Local Order,
$s_k < s_k'$ and $p_k$ does not broadcast any messages between $s_k$ and $s_k'$.
Thus, by the code of the $\EXECUTE$ handler, we have
$$
\exists c.\, \Theta_k(s_k) \stackrel{c}{\rightarrow} \langle \_, \delta_l \rangle, 
$$
which by~(\ref{eq:sigma-theta}) implies
$$
\Sigma_i(t_i^{l-1}) \stackrel{c}{\rightarrow} \langle \_, \delta_l \rangle,
$$
where $c$ is the command used to compute $\delta$ by the sender $p_k$
of $m_l$, as needed. 

Suppose next that $e_k \neq e_k'$, let $p_r$ be the leader
of $e_k'$, and $s_r$ be the time when $p_r$ joins $e_k'$. 
Since $p_i$ delivers $m_{l-1}$ before $m_l$,
by \Consistency{} ((b/only if)),
either $p_r$ delivers $m_{l-1}$ before joining $e_k'$, or 
$p_r$ speculatively delivers $m_{l-1}$ when joining $e_k'$.
Suppose first that $\deliv_k(m_{l-1})$ occurs 
before $p_r$ joins $e_k'$. 
Since by Global Order, $e_k < e_k'$,
by Lemma~\ref{lem:aux-deliver}, 
$p_r$ does not deliver any messages after delivering $m_{l-1}$ and before joining $e_k'$
and does not speculatively deliver any messages 
when joining $e_k'$. Thus, by the induction hypothesis 
for Lemma~\ref{thm:pr}.\ref{thm:pr:context21},
$$
\Theta_{r}(s_r) = \Theta_k(s_k).
$$
Since $p_i$ does not deliver any messages after delivering
$m_{l-1}$ and before delivering $m_l$, by Lemma~\ref{lem:aux-deliver3}, $p_r$
does not broadcast any messages in-between $s_r$ and $s_k'$.
Thus, by the code of the $\EXECUTE$ handler,
$$
\exists c.\, \Theta_r(s_r) \stackrel{c}{\rightarrow} \langle \_, \delta_l \rangle
\wedge m_{l} = \langle \_, \_, \delta_l \rangle, 
$$
which implies
$$
\exists c.\, \Theta_k(s_k) \stackrel{c}{\rightarrow} \langle \_, \delta_l \rangle
\wedge m_{l} = \langle \_, \_, \delta_l \rangle.
$$
Thus, by~(\ref{eq:sigma-theta}) we get~(\ref{eq:find-c}), as needed. 
Next, suppose that $p_r$ speculatively delivers $m_{l-1}$ when joining $e_k'$.
Then, by Lemma~\ref{lem:aux-deliver2}, 
$m_{l-1}$ is the last message speculatively delivered by $p_r$
when joining $e_k'$. Hence, by the induction hypothesis 
for Lemma~\ref{thm:pr}.\ref{thm:pr:context22},
$$
\Theta_r(s_r) = \Theta_k(s_k).
$$
Thus, applying the same reasoning as above, we again get~(\ref{eq:find-c}), as needed. 

Finally, suppose that there exists a process $p_r \neq p_i$ such that
$p_r$ delivers $>l$ messages before $t$.
Thus, we have to show that~(\ref{eq:to-prove}) holds 
given that $\sigma_{h'}$ and $f_{h'}^j$ for each process $p_j$
are chosen as shown in~(\ref{eq:sigma-same}) and~(\ref{eq:f-same}).
Then,~(\ref{eq:to-prove}) can be re-written as follows
\begin{align}
&\forall t' \le t.\, \forall p_j.\, \exists k.\,
\Sigma_0 \stackrel{\prefix{\sigma_{h}}{k}}{\leadsto} \Sigma_j(t') \wedge k=f_{h'}^j(t').
\label{eq:to-prove2}
\end{align}
By the induction hypothesis and~(\ref{eq:f-same})
\begin{align*}
&\forall t' \le t.\, \forall p_j.\, p_j \neq p_i \implies \exists k.\,
\Sigma_0 \stackrel{\prefix{\sigma_{h}}{k}}{\leadsto} \Sigma_j(t') \wedge k=f_{h'}^j(t')\ \wedge\\
&\forall t' < t.\, \exists k.\,
\Sigma_0 \stackrel{\prefix{\sigma_{h}}{k}}{\leadsto} \Sigma_i(t') \wedge k=f_{h'}^i(t').
\label{eq:to-prove2}
\end{align*}
It remains to show that
\begin{equation}
\exists k.\, \Sigma_0 \stackrel{\prefix{\sigma_{h}}{k}}{\leadsto} \Sigma_i(t) \wedge k=f_h^i(t-1)+1.
\label{eq:last-to-show}
\end{equation}
Let $\ell_k(s)$ denote the length of the message prefix delivered by process $p_k$ at or before
time $s$. By the induction hypotheses for 
Lemma~\ref{thm:pr}.\ref{thm:pr:mnum}, we get
$$
f_h^r(t-1)=\ell_h^r(t-1) > \ell_h^i(t-1)=f_h^i(t-1).
$$
Since by the induction hypothesis for Lemma~\ref{thm:pr}.\ref{thm:pr:mnum},
$f_h^r$ never skips a value, there exists a time $t'' \le t-1$ such that
$$
f_h^r(t'')=\ell_h^r(t'') = \ell_h^i(t-1)=f_h^i(t-1).
$$
Let $s_r$ and $s_i$ be the earliest times for which the above holds, i.e., 
$$
f_h^r(t'')=f_h^r(s_r)=\ell_h^r(s_r) = \ell_h^i(s_i)=f_h^i(s_i)=f_h^i(t-1).
$$
Then, by the induction hypothesis, we get
$$
\exists k.\, \Sigma_0 \stackrel{\prefix{k}{\sigma_h}}{\leadsto} \Sigma_r(s_r) \wedge
\Sigma_0 \stackrel{\prefix{k}{\sigma_h}}{\leadsto} \Sigma_i(s_i) \wedge k=f_h^r(s_r).
$$
Thus,
$$
\Sigma_r(s_r)=\Sigma_i(s_i).
$$
By the induction hypothesis for Lemma~\ref{thm:pr}.\ref{thm:pr:mnum}, 
the value of $f_h^r$ remains the same until the time $s_r'>s_r$ at which 
$p_r$ delivers the next message. Hence,
\begin{equation}
\Sigma_r(s_r'-1)=\Sigma_r(s_r)=\Sigma_i(s_i).
\label{eq:equal-states2}
\end{equation}
By the induction hypothesis
$$
\exists c.\, \sigma_h[f_h^r(s_r)+1]=c \wedge 
\exists \delta.\,\Sigma_r(s_r')=\app(\Sigma_r(s_r'-1), \delta) \wedge \Sigma_r(s_r'-1) \stackrel{c}{\rightarrow} 
\langle \_,\delta \rangle.
$$
Since by the code, $\Sigma_r(s_r'-1)$ can only transition to $\Sigma_r(s_r')$ by applying
the delta included in the payload of the message delivered at $s_r'$.
By Agreement and Total Order, this message
must be the same as the one $p_i$ delivers at $t$, i.e., $m_l=\langle \_,\_,\delta_l \rangle$.
Hence, we get $\delta=\delta_l$, and therefore,
\begin{equation}
\exists c.\, \sigma_h[f_h^r(s_r)+1]=c \wedge 
\Sigma_r(s_r')=\app(\Sigma_r(s_r'-1), \delta_l) \wedge \Sigma_r(s_r'-1) \stackrel{c}{\rightarrow} 
\langle \_,\delta_l \rangle.
\label{eq:derive-state}
\end{equation}
Since $f_h^i(s_i)=f_h^i(t-1)$, by the induction hypothesis for Lemma~\ref{thm:pr}.\ref{thm:pr:mnum}, 
the value of $f_h^i$ remains the same until the time $t>s_i$ at which 
$p_r$ delivers $m_l$. Hence, by~(\ref{eq:equal-states2}), 
\begin{equation}
\Sigma_r(s_r'-1)=\Sigma_r(s_r)=\Sigma_i(s_i)=\Sigma_i(t-1).
\label{eq:more-equal-states}
\end{equation}
Since at $t$, $p_i$ delivers $m_l=\langle \_,\_,\delta_l$ and applies it to
$\Sigma_i(t-1)$, from~(\ref{eq:derive-state}) and~(\ref{eq:more-equal-states}), we get
$$
\exists c.\, \sigma_h[f_h^r(s_r)+1]=c \wedge 
\Sigma_i(t)=\app(\Sigma_r(t-1), \delta) \wedge \Sigma_t(t-1) \stackrel{c}{\rightarrow} 
\langle \_,\delta_l \rangle.
$$
Since $f_h^r(s_r)=f_h^i(t-1)$, the above implies
$$
\exists c.\, \sigma_h[f_h^i(t-1)+1]=c \wedge 
\Sigma_i(t)=\app(\Sigma_r(t-1), \delta) \wedge \Sigma_t(t-1) \stackrel{c}{\rightarrow}, 
\langle \_,\delta_l \rangle.
$$
which combined with the induction hypothesis yields~(\ref{eq:last-to-show}),
as needed.
\end{proof}
We are now ready to prove the main result:
\begin{proof}[Proof of Theorem~\ref{thm:correctness-passive}]
Consider a history of the passive replication algorithm in Figure~\ref{fig:app-layer}.
By Lemma~\ref{thm:pr}, there exists a sequence of commands $\sigma = \sigma_1..\sigma_l$ such that
for all $1 \le k \le l$, $\execute(\sigma_k)$ is an event in $h$, and the conditions
given by Lemmas~\ref{thm:pr:main}--\ref{thm:pr:context22} hold for 
all times $t$ in $h$. By Lemma~\ref{thm:pr}.\ref{thm:pr:main}, for all $k$,
there exists a committed state $\Sigma_k$ and a response $r_k$ such that
$\Sigma_k \stackrel{\sigma_k}{\rightarrow} \langle r_k, \_ \rangle$. 
Let $\pi=\sigma_1,r_1,..,\sigma_l,r_l$. Then, by induction on the length of $\pi$,
it follows that $\pi$ is a history of the specification in Figure~\ref{fig:passive-spec},
which implies the required.
\end{proof}

%% file: app-po-proof.tex
\section{Proof of Correctness for the Primary-Order Broadcast Protocol
  (Theorem~\ref{thm:correctness2})} 
\label{sec:proof-po}

\figref{fig:vp_invariants-rem} summarizes additional invariants that, together
with the invariants listed in \figref{fig:vp_invariants}, are used to prove the
correctness of the protocol.
The proof that the
Primary-Order Broadcast protocol satisfies the invariants in
\figref{fig:vp_invariants} and the properties in \figref{fig:rpob-properties}
(except for Property~\ref{prop:total}, which is not required) is virtually
identical to the one for the atomic broadcast protocol, and therefore omitted.

\begin{figure}[h!]
	{\small	
	\begin{enumerate}
		
		  \setcounter{enumi}{4}
				
      \item \label{inv:new_state_prefix}
	Assume the leader of an epoch $\eid$ sends $\ACCEPT(\eid, k, \vid)$ while
	having $\msg = \msgi$. Whenever any process $\pid$ has $\epoch = \eid' > \eid$ 
	and $\msg[k] = \vid$, it also has $\prefixmsg = \msgi$.

                       \item \label{inv:prefix-of-leader}
                  If a process $\pid$ receives
          $\ACCEPT(\eid, k, \vid)$ and replies with $\ACCEPTACK$, then after
          this and while $\epoch = \eid$ at $\pid$ we have
          $\prefixmsg = \prefixmsg[k][\msgi]$, where $\msgi$ was the value of
          the $\msg$ array at the leader of $\eid$ when it sent
          $\ACCEPT(\eid, k, \vid)$.
	
		\item \label{inv:uniqueness_accept} For any messages $\ACCEPT(e, k, m_1)$ and
		$\ACCEPT(e, k, m_2)$, we have $m_1 = m_2$.

\item \label{inv:accept-commit}
	Assume the leader of an epoch $\eid$ sends $\ACCEPT(\eid, k,
        \vid)$. If a process $p_i$ delivers $\vid$ at an epoch
        $\eid'\geq \eid$ after handling $\COMMIT(\eid', k')$, then $k=k'$.

	\end{enumerate}
}
	\caption{Remaining invariants. }
	\label{fig:vp_invariants-rem}
	
\end{figure}

\begin{proof}[Proof of \invref{inv:prefix-of-leader}.]
Assume that a process $\pid$ at $\eid$ receives $\ACCEPT(\eid, k,
\vid)$ and replies with the message $\ACCEPTACK$. We prove that, after
the transition and while $\epoch = \eid$, $\pid$ has $\prefixmsg =
\prefixmsg[k][\msgi]$ where $\msgi$ is the value of the array $\msg$
at the leader of $\eid$ when it sent $\ACCEPT(\eid, k,s
\vid)$. 

By the protocol, since $\pid$ handles $\ACCEPT(\eid, k, \vid)$,
by line \ref{vp_accept_pre}, $\pid$ has $\epoch=\eid$. This
implies that $\pid$ is a follower in $\eid$ and that $\pid$ has handled
the $\NEWSTATE(\eid, \msgi')$ message sent by the leader $\pid[j]$ of
$\eid$. After handling this message, $\pid$ has 
$\msg = \msgi'$ where $\msgi'$ is the value of $\msg$ at $\pid[j]$
when it sent the $\NEWSTATE(\eid, \msgi')$
message. Let $k_1=\length(\msgi')$.  
We know that $\pid[j]$ does not overwrite
any position $\leq k_1$ from the moment it sends 
$\NEWSTATE(\eid,
\msgi')$ while in $\eid$. Then
$\prefixmsg[k_1][\msgi]=\msgi'$. The same holds of $\pid$ from the moment 
 it handles $\NEWSTATE(\eid, \msgi')$ and while it stays in $\eid$. Thus, (*) $\pid$ has
$\prefixmsg[k_1][\msg]=\prefixmsg[k_1][\msgi]$ after handling
$\NEWSTATE(\eid, \msgi')$ and while in $\eid$.
We now prove that $\pid$ receives an $\ACCEPT$ message for every
position $k^{*}$ such that $k_1<k^{*}<k$ before receiving
$\ACCEPT(\eid, k, \_)$. We prove it by induction on $k$. Assume that
the invariant holds for $k_2$ such that $k_1<k_2<k$. We now show it for
$k_2=k$. Consider the case when
 $k_1<k_2-1$; otherwise the required holds trivially.
If $\pid$ receives $\ACCEPT(\eid, k_2, \vid)$, then the
leader $\pid[j]$ has $\next=k_2$ when it sent it. The leader $\pid[j]$
sets $\next=k_1$ at line~\ref{vp_quorum_of_new_state_acks_next_opt} before
setting $\status=\LEADER$. After that, $\pid[j]$ only updates $\next$
at line~\ref{vp_broadcast_next} by increasing it by one. Therefore,
$\pid[j]$ has
$\next=k_2-1$ at some point before sending $\ACCEPT(\eid, k_2,
\vid)$. Since  $k_1<k_2-1$,  the leader $\pid[j]$ sends $\ACCEPT(\eid,
k_2-1, \_)$ before sending $\ACCEPT(\eid, k_2, \vid)$. Since channels are FIFO,
$\pid$ receives $\ACCEPT(\eid, k_2-1, \_)$ 
before receiving $\ACCEPT(\eid,
k_2, \vid)$. Then, by this and the induction hypothesis, $\pid$
receives an $\ACCEPT$ message for all positions $<k_2$ (and $>k_1$)
before receiving $\ACCEPT(\eid, k_2, \_)$, as required.
Hence, $\pid$ receives an $\ACCEPT$ message for every
position $k^{*}$ such that $k_1<k^{*}<k$ before receiving
$\ACCEPT(\eid, k, \_)$. When $\pid[j]$ sends $\ACCEPT(\eid,
k^{*}, \vid^{*})$, it sets $\msg[k^{*}]=\vid^{*}$. 
While $\epoch = \eid$, $\pid[j]$ does not overwrite $\msg[k^*]$, so that
$\msgi[k^{*}]=\vid$. When $\pid$ handles $\ACCEPT(\eid,
k^{*}, \vid^{*})$, it sets $\msg[k^{*}]=\vid^{*}$.  
Since $\pid$ does not overwrite $\msg[k^{*}]$ while $\epoch = \eid$,
after handling $\ACCEPT(\eid, k^{*}, \vid^{*})$, 
$\pid$ has $\msg[k^{*}]=\msgi[k^{*}]$ while in
$\eid$. Then by (*), $\pid$ has $\prefixmsg[k][\msg] =
\prefixmsg[k][\msgi]$ after handling $\ACCEPT(\eid, k, \vid)$
and while in $\eid$, as required.
\end{proof}

\begin{lemma}
  \label{lem:epochbefore}
Assume that the leader of an epoch $\eid$ sends $\ACCEPT(\eid, k,
m)$. Let $\eid'>\eid$ be an epoch and $p_i$ its leader. If $p_i$ sends
$\NEWSTATE(\eid',msg',\_)$ such that $\msg'[k]=m$, then $p_i$ has
$\epoch\geq \eid$ right before it sends $\NEWSTATE(\eid',msg',\_)$.
\end{lemma}
\begin{proof}
  Let $e_0$ be the value of $\epoch$ at $p_i$ right before handling
  $\NEWCONFIG(\eid', \_)$. The process $p_i$ has $\msg[k]=\vid$ when it sends
  $\NEWSTATE(\eid',msg',\_)$. Then $p_i$ sets $\msg[k]=\vid$ after handling
  $\ACCEPT(e^*, k, m)$ or $\NEWSTATE(e^*, msg_0, \_)$ such that $\msg_0[k] = m$
  in an epoch $e^*\le e_0$. In the former case, by \equref{prop:env-1},
  $\eid=e^* \le e_0$, which implies required.  In the latter case, we prove by
  contradiction that $e^*\geq \eid$. Assume that $e^*<\eid$. By
  Invariant~\ref{inv:origin}, there exists an epoch $<\eid$ in which $m$ is
  broadcast. But by \equref{prop:env-1}, this is impossible. Thus,
  $e_0 \geq e^*\geq \eid$, as required.
\end{proof}
  
\begin{proof}[Proof of \invref{inv:new_state_prefix}.]
We prove the invariant by induction on $\eid'$. Assume the invariant holds for each
$\eid' < \eid''$.  We now show that it holds for $\eid' = \eid''$ by induction on the length
of the protocol execution. We only consider the most interesting transition in
line~\ref{vp_new_state}, where a process $\pid$ receives $\NEWSTATE(\eid'',
\vmsg'')$ from the leader of $\eid''$ and sets $\msg$ to $\vmsg''$
(line~\ref{vp_newstate_msg_assignment}). Assume that
after the transition $p_i$ has $\msg''[k]=\vid$; then $\vmsg''[k]=\vid$. To show
that after the transition $p_i$ has $\prefixmsg = \vmsg$, it is sufficient to
prove that $\prefixmsg[@][\msgi''] = \vmsg$. Let $\pid[j]$ be the leader of
$\eid''$, i.e., the process that sends 
$\NEWSTATE(\eid'', \vmsg'')$, and let $\eid_0<\eid''$ be the value of $\epoch$ at
$\pid[j]$ right before receiving  
the $\NEWCONFIG(\eid'', \_)$ message at line \ref{vp_new_config}.
By Lemma~\ref{lem:epochbefore}, $\eid \leq \eid_0$.

Consider first the case when $\eid<\eid_0$. We have that $\pid[j]$ has
$\epoch=\eid_0$ at some point. Therefore, $\pid[j]$ sends or handles
$\NEWSTATE(\eid_0, \msgi_0)$.
Right after this $\pid[j]$ has
$\msg[k]=\msgi_0[k]$, and 
$\pid[j]$ does not overwrite $\msg[k]$ while $\epoch = \eid_0$.
Then, when $\pid[j]$ handles 
$\NEWCONFIG(\eid'', \_)$, it has $\msg[k] =\msgi_0[k]=\msgi''[k] = m$.
Since $\eid_0<\eid''$, by the induction hypothesis, from the moment
when $\pid[j]$ sends or handles $\NEWSTATE(\eid_0, \msgi_0)$
and while $p_j$ is in $\eid_0$, it has $\prefixmsg[@][\msg] = \msgi$.
Then $\pid[j]$ has $\prefixmsg[@][\msg] = \msgi$ right before
handling $\NEWCONFIG(\eid'', \_)$. Hence, $\prefixmsg[@][\msgi''] = \msgi$, as required. 

Consider now the case when $\eid = \eid_0$. There are two possibilities: either
$\pid[j]$ is the leader of $\eid$, or it is a follower in
$\eid$. We only consider the latter case, since the former is analogous. 
We know that $\pid$ receives $\NEWSTATE(\eid'', \msgi'')$ with
$\msgi''[k]=\vid$ from $\pid[j]$.  Therefore, $\pid[j]$ has
$\msg[k]=\vid$ before sending $\NEWSTATE(\eid'', \msgi'')$. The process
$\pid[j]$ sets $\msg[k]=\vid$ when handling $\ACCEPT(\eid, k, \vid)$ or
$\NEWSTATE(\eid, \msgi''')$ such that $\msgi'''[k]=\vid$. By \invref{inv:origin} and
\equref{prop:env-1}, the latter is impossible. Thus, $\pid[j]$ handles
$\ACCEPT(\eid, k, \vid)$ and replies with $\ACCEPTACK$.
Hence, by \invref{inv:prefix-of-leader}, $\pid[j]$ has
$\prefixmsg[@][\msg] = \msgi$ right before sending $\NEWSTATE(\eid'', \msgi'')$.
Thus, $\prefixmsg[@][\msgi''] = \msgi$, as required.  
\end{proof}

\begin{proof}[Proof of \invref{inv:uniqueness_accept}.]
We have proved that the protocol satisfies \confchgnot and, in particular, 
Property~\ref{prop:wf-1}.
Hence, for any epoch number there is at most one leader. 
Let $\pid$ be the leader of $\eid$.
Then, since $\ACCEPT(\eid, k, \vid[_1])$ and $\ACCEPT(\eid, k, \vid[_2])$ are sent in $\eid$, then
$\pid$ is the only process that can send these messages.

By contradiction, assume $\vid[_1] \neq \vid[_2]$. Then, $\pid$ must have
received $\bcast(\vid[_1])$ and $\bcast(\vid[_2])$.  Let us assume without loss
of generality that it receives $\bcast(\vid[_1])$ before $\bcast(\vid[_2])$.
When $\pid$ handles $\bcast(\vid[_1])$, it sends the
$\ACCEPT(\eid, k, \vid[_1])$ message to the followers. We note that $k$ is the
value stored in the $\next$ variable.  Furthermore, $\pid$ increases $\next$ by
one (line \ref{vp_broadcast_next}) for each application message to
broadcast. Hence, when $\pid$ has to handle $\bcast(\vid[_2])$ while in $\eid$,
it will use a slot number $>k$, which contradicts our assumption.
\end{proof}

\begin{proof}[Proof of Invariant~\ref{inv:accept-commit}.]
  Let $\NEWSTATE(\eid', msg, \_)$ be the $\NEWSTATE$
message sent by the leader of $\eid'$. The process $p_i$ has  $\msg[k']=\vid$ when it handles 
$\COMMIT(\eid', k')$. Then either it has received $\ACCEPT(\eid', k',
\vid)$ or  $msg[k']=\vid$. In the latter case, by Invariant~\ref{inv:origin}, there
exists an epoch $\eid''<\eid'$ such that its leader sends
$\ACCEPT(\eid'', k' , \vid)$. Thus, in both cases, there exists an
epoch $\eid^*$ in which its leader sends $\ACCEPT(\eid^*, k',
\vid)$. Furthermore, we
have that the leader of $\eid$ sends $\ACCEPT(\eid, k,
\vid)$. By \equref{prop:env-1}, $\eid=\eid^*$ and $k=k'$, as required. 
\end{proof}

\begin{proof}[Basic Speculative Delivery Properties]
Let $\eid$ be an epoch at which $p_i$ is the leader. A
  process speculatively delivers an application message when it receives a $\NEWCONFIG$ message. 
  Line~\ref{vp_send_new_config} guarantees that $\NEWCONFIG$ is only sent to the
  leader of $\eid$. Therefore, only $p_i$ can speculatively deliver
  messages in $\eid$, as required.

  A process speculatively delivers an application message in its $\msg$ array at
  a position $>\lastdelivered$ when it receives a $\NEWCONFIG$ message. Assume
  that $p_i$ speculatively delivers $m$ in $\eid$. Then it receives
  $\NEWCONFIG(\eid, \vm)$ and there is a $k$ such that $\msg[k]=m$ and
  $k>\lastdelivered$ at $p_i$ at that moment. The process $p_i$ sets $\msg[k]=m$
  either after handling $\ACCEPT(\eid', k, m)$ or $\NEWSTATE(\eid', msg, \_)$
  such that $msg[k]=m$. In the latter case, by Invariant~\ref{inv:origin}, there
  exists an epoch $\eid''<\eid'$ such that the leader of $\eid''$ has previously
  sent $\ACCEPT(\eid'', k, m)$. Then, in either case, there exists an
  $\ACCEPT(\eid^*, k, m)$ message such that $\eid^*< \eid$. Thus, $m$ has been
  previously broadcast, as required. We now prove that (i) $p_i$ speculatively
  delivers $m$ at most once in $\eid$, and (ii) $p_i$ does not deliver $m$
  before handling $\NEWCONFIG(\eid, \vm)$. We prove both by contradiction.
  \begin{enumerate}[(i)]
  \item  
  Assume that $p_i$ has $\msg[k']=m$ such that $k'\neq k$ when it
  receives $\NEWCONFIG(\eid, \vm)$. Following the same reasoning as before, there exists an
  $\ACCEPT(\eid_1, k', m)$ message such that $\eid_1<\eid$. By
  \equref{prop:env-1}, an application message is only broadcast once, so that $\eid_1=\eid^*$
  and $k'=k$, which contradicts our assumption.
  \item Assume that
  $p_i$ delivers $m$ before handling 
  $\NEWCONFIG(\eid, \vm)$. Then $p_i$ handles a $\COMMIT(\eid_1, k')$
  message such that $\eid_1<\eid$ while $\msg[k']=m$. When $p_i$
  handles $\COMMIT(\eid_1, k')$, it sets $\lastdelivered=k'$. Thus,
  $p_i$ has $\lastdelivered\geq k'$ when it handles
  $\NEWCONFIG(\eid, \vm)$. Therefore, $k'<k$. If $p_i$ has
  $\msg[k']=m$ when it handles $\COMMIT(\eid_1, k')$, then it has handled
  an $\ACCEPT(\eid_1, k', m)$ message or an $\NEWSTATE(\eid_1, msg_1, \_)$
  message such that
  $msg_1[k']=m$.  In the latter case, by Invariant~\ref{inv:origin},
  there exists an epoch $\eid_2<\eid_1$ such that the leader of
  $\eid_2$ has previously sent $\ACCEPT(\eid_2, k', m)$. Then, in
  either case, there
  exists an $\ACCEPT(\eid_3, k', m)$ message such that $\eid_3\le
  \eid_1$. By
  \equref{prop:env-1}, an application message is only broadcast once, then $\eid_3=\eid^*$
  and $k'=k$, which contradicts our assumption.
    \end{enumerate}

  \end{proof}

\begin{lemma}
  \label{thm:integrity}
  If a process delivers an application message $\vid$ while in an epoch $e'$,
  then $\vid$ has been previously broadcast in some epoch $e \le e'$.
\end{lemma}
\begin{proof}
  The proof is analogous to the proof of Integrity for atomic broadcast in
  \S\ref{sec:correctness}.
\end{proof}

\begin{proof}[Proof of Local Order.]
  Assume that the leader $\pid$ of an epoch $\eid$ receives $\bcast(\vid[_1])$
  followed by $\bcast(\vid[_2])$ at $\eid$ and that some process $\pid[j]$
  delivers $\vid[_2]$ at $\eid'$. By Lemma~\ref{thm:integrity} and
  \equref{prop:env-1}, $e \le \eid'$. We show that $\pid[j]$ delivers $\vid[_1]$
  before delivering $\vid[_2]$.

 \begin{itemize}
 \item Assume first that $\eid'=\eid$. Let:
   \begin{itemize}
     \item $k = \length(\msgi)$, where $\msgi$ is the array of
       messages sent by the leader of $\eid$ in  $\NEWSTATE(\eid,
       \msgi)$;
       \item $\ACCEPT(\eid, k', \vid[_1])$ and $\ACCEPT(\eid, k'',
         \vid[_2])$ be the $\ACCEPT$ messages sent by the leader of $\eid$
         to broadcast $\vid[_1]$ and $\vid[_2]$.
       \end{itemize}
       Since the leader of $\eid$ broadcast $\vid[_1]$ before $\vid[_2]$,
       then $k'<k''$. Moreover, by
          lines \ref{vp_broadcast_next} and
          \ref{vp_quorum_of_new_state_acks_next_opt}, $k'>k$. 

          By line~\ref{vp_commit_pre}, $\pid[j]$ has
          $\lastdelivered=k''-1$ right before handling $\COMMIT(\eid,
          k'')$. Hence, the messages stored at  
          slot numbers between $1$ and $k''-1$, must have been delivered.
          Consequently, since $k'<k''$, we know that $\pid[j]$ handles
          $\COMMIT(\eid^{*}, k')$ before handling $\COMMIT(\eid, k'')$. If
          $\eid^{*}<\eid$, then by
          \invref{inv:prefix-of-higher-epoch}, when $\pid$ became the
          leader of $\eid$, it had $\msg[k']\neq\bot$, so that $k'<k$. But we have established that
          $k'>k$. Hence, we must have $\eid^{*}=\eid$. Consider the case when the leader of
          $\eid$ sends $\COMMIT(\eid, k')$ because it receives a
          $\NEWSTATEACK$ message from every follower. This is
          impossible: in this case $\pid$ sends $\COMMIT$ messages for the
          entries between $1$ and $k$, but $k'>k$. Thus, the leader
          of $\eid$ sends $\COMMIT(\eid, k')$ because it  receives an
          $\ACCEPTACK(\eid, k')$ message from every follower. This
          implies that the leader has sent an $\ACCEPT(\eid, k',
          \vid[_3])$ message to its followers for some $m_3$. Since we know that
          $\pid$ also sends $\ACCEPT(\eid, k', \vid[_1])$, by
          \invref{inv:uniqueness_accept} we have
          $\vid[_3]=\vid[_1]$. When handling this message, $\pid[j]$ sets
        $\msg[k]=\vid[_1]$. Since $\pid[j]$ does not overwrite $\msg[k]$ while $\epoch = \eid$, $\pid[j]$ has
        $\msg[k]=\vid[_1]$ when it handles $\COMMIT(\eid,
	k')$. Hence, $\pid[j]$  delivers $\vid[_1]$ before
        $\vid[_2]$, as required.

      \item Assume now that $\eid'>\eid$. Let $\COMMIT(\eid', k)$ be
        the commit message that makes $p_j$ deliver $\vid[_2]$.
        By \equref{prop:env-1}, it is easy to
 see that the leader of $\eid'$ sends $\COMMIT(\eid', k)$ at
line~\ref{vp_send_commit_2} after receiving a $\NEWSTATEACK$ message from
every follower. 
Let $\NEWSTATE(\eid', \msgi')$ be the $\NEWSTATE$ message sent by the leader
of $\eid'$. Then $\length(\msgi')\ge k$. By \invref{inv:origin}, there exists an epoch
$\eid''<\eid'$ whose leader sends $\ACCEPT(\eid'', k, \vid[_2])$. We know that
$\pid$ broadcast $\vid[_2]$ in $\eid$, and thus sent
$\ACCEPT(\eid, \_, \vid[_2])$. Then, by \equref{prop:env-1},
$\eid''=\eid$. Let $\msgi$ be the value of the
array $\msg$ at $\pid$ at the time when it sent $\ACCEPT(\eid,
k,\vid[_2])$; then $\msgi[k]=m_2$.
We know that $\pid$ broadcast $\vid[_1]$ before $\vid[_2]$ in
$\eid$. Then there exists a $k'<k$ such that $\msgi[k']=\vid[_1]$.

We have that $\pid$ sends $\ACCEPT(\eid, k, \vid[_2])$, $\msgi'[k]=\vid[_2]$ and
$\eid'>\eid$. Therefore, by \invref{inv:new_state_prefix}, $\pid[j]$ has
$\prefixmsg[k] = \msgi$ after handling $\NEWSTATE(\eid', \msgi')$ and while in
$\eid'$. In particular, $\msg[k'] = m_1$. We know that $\pid[j]$ delivers $m_2$
after handling $\COMMIT(\eid', k)$. Therefore, before this $\pid[j]$ delivers
all values in positions $1$ to $k-1$. Let $k_0$ be the value of $\lastdelivered$
at $\pid[j]$ after it handles $\NEWSTATE(\eid', \msgi')$. Assume first that
$k_0<k'$. Then $\pid[j]$ delivers the value at position $k'$ while in
$\eid'$. We have established that $\pid[j]$ has $\msg[k']=\vid[_1]$ after
handling $\NEWSTATE(\eid', \msgi')$ and while in $\eid'$. Thus, $\pid[j]$
delivers $m_1$, as required. Assume now that $k_0\ge k'$. Then, $\pid[j]$
handled a $\COMMIT$ message for position $k'$ in an epoch $< e'$ and delivered a
message $m'$. By \invref{inv:prefix-of-higher-epoch}, $\pid[j]$ has
$\msg[k']=m'$ at $\eid'$. We have already established that $\pid[j]$ has
$\msg[k']=\vid[_1]$ after handling $\NEWSTATE(\eid', \msg')$ and while in
$\eid'$. Thus, $m'=\vid[_1]$ and $\pid[j]$ delivers $m_1$, as required.
\end{itemize}
\end{proof}

\begin{proof}[Proof of Global Order.]
 We have that:
  \begin{itemize}
\item the leader of $\eid$ sends $\ACCEPT(\eid, k_1, \vid[_1])$ while
  in $\eid$,
  \item the leader of $\eid'$ sends $\ACCEPT(\eid', k_2, \vid[_2])$ while
  in $\eid'$,
\end{itemize}
and $e < e'$. Then by \equref{prop:env-1}, $m_1 \not= m_2$.

Let $\eid_1$ and $\eid_2$ be the epochs at which $p_i$ delivers $\vid[_1]$ and
$\vid[_2]$ respectively. By Lemma~\ref{thm:integrity} and \equref{prop:env-1},
$\eid \le \eid_1$ and $\eid' \le \eid_2$.  The process $p_i$ delivers $\vid[_1]$
when it receives a message $\COMMIT(\eid_1, k)$ from the leader of $\eid_1$
while having $\msg[k]=\vid[_1]$, and $\vid[_2]$ when it receives a message
$\COMMIT(\eid_2, k')$ from the leader of $\eid_2$ while having
$\msg[k']=\vid[_2]$. By Invariant~\ref{inv:accept-commit}, $k=k_1$ and
$k'=k_2$. Furthermore, since $m_1 \not= m_2$, by 
Lemma~\ref{thm:commit-msg} we have $k_1\neq k_2$.

Lines~\ref{vp_commit_pre}-\ref{vp_commit_last_delivered} and the fact that
a process never decrements $\lastdelivered$ guarantee that a process delivers
each position of its $\msg$ array only once and in position order. Then, if
$k_1<k_2$, the process delivers $\vid[_1]$ before $\vid[_2]$, as required. We
now show that $k_1>k_2$ is impossible. Assume the contrary, so that
$k_1>k_2$. As argued above, a process delivers each position of its $\msg$ array
only once and in order. Since $k_1>k_2$ , then $p_i$ has to deliver $\vid[_2]$
before $\vid[_1]$, so that $\eid_2\leq \eid_1$.
\begin{itemize}
\item Consider first the case when $\eid_1>\eid_2$. We have
  $\eid<\eid' \le \eid_2 < \eid_1$. The process $p_i$ has $\msg[k_1]=\vid[_1]$
  when it handles $\COMMIT(\eid_1, k_1)$. Then $p_i$ either handles an
  $\ACCEPT(\eid_1, k_1, \vid[_1])$ message or the process has
  $\msg[k_1]=\vid[_1]$ when it enters $\eid_1$, i.e., when it sets
  $\epoch=\eid_1$. The application message $\vid[_1]$ is broadcast in $\eid$ and
  $\eid<\eid_1$. Then by \equref{prop:env-1}, $\vid[_1]$ cannot be broadcast in
  $\eid_1$. Thus, $p_i$ has $\msg[k_1]=\vid[_1]$ when it enters $\eid_1$. Since
  $\eid_1>\eid_2$, by Invariant~\ref{inv:prefix-of-higher-epoch}, $p_i$ has
  $\msg[k_2]=\vid[_2]$ when it enters $\eid_1$. By
  Invariant~\ref{inv:new_state_prefix}, $p_i$ also has $\prefixmsg[k_1] = \msgi$
  when it enters $\eid_1$, where $\msgi$ is the $\msg$ array that the leader of
  $\eid$ had when it sent $\ACCEPT(\eid, k_1, \vid[_1])$. Since $k_2<k_1$,
  $\msgi[k_2]=\vid[_2]$. But $\vid[_2]$ is broadcast at $\eid'$ and
  $\eid'>\eid$. Thus, by \equref{prop:env-1} and \invref{inv:origin},
  $\msgi[k_2]\neq\vid[_2]$, yielding a contradiction.
  
\item Consider now the case when $\eid_2=\eid_1$. We have
  $\eid<\eid' \le e_2 = \eid_1$. The application message $\vid[_1]$ is broadcast in
  $\eid$. Then by \equref{prop:env-1}, $\vid[_1]$ cannot be broadcast in
  $\eid_1$. Thus, $p_i$ has $\msg[k_1]=\vid[_1]$ when it enters $\eid_1$, i.e.,
  when it sets $\epoch=\eid_1$. By Invariant~\ref{inv:new_state_prefix}, $p_i$
  also has $\prefixmsg[k_1] = \msgi$ when it enters $\eid_1$, where $\msgi$ is
  the $\msg$ array that the leader of $\eid$ had when it sent
  $\ACCEPT(\eid, k_1, \vid[_1])$. After setting $\epoch=\eid_1$ and while
  $\epoch=\eid_1$, $p_i$ does not overwrite $\msg[k]$. Since $k_1>k_2$ and $p_i$
  has $\msg[k_2]=\vid[_2]$ when it receives a message $\COMMIT(\eid_1, k_2)$
  from the leader of $\eid_1$, we must have $\msgi[k_2] = \vid[_2]$. But
  $\vid[_2]$ is broadcast at $\eid'$ and $\eid'>\eid$. Thus, by
  \equref{prop:env-1} and \invref{inv:origin}, $\msgi[k_2]\neq\vid[_2]$,
  yielding a contradiction.
\end{itemize}
\end{proof}

\begin{lemma}
  \label{lem:setposition}
Assume that the leader of an epoch $e$ sends $\ACCEPT(e, k,
m)$. If a process $p_i$ has $\msg[k']=m$, then $k'=k$.
\end{lemma}
\begin{proof}
 The process $p_i$ sets
$\msg[k']=m$ either after handling an $\ACCEPT(e', k', m)$ message or
a $\NEWSTATE(e', msg', \_)$ message such that $msg'[k'] = m$. In the
latter case, by Invariant~\ref{inv:origin},  there exists an epoch $e'' < e'$ such
that the leader of $e''$ has previously sent $\ACCEPT(e'', k',
m)$. Thus, if $p_j$ sets $\msg[k']=m$, then there exists an $\ACCEPT(e^*, k', m)$
message. By \equref{prop:env-1}, a message is only broadcast
once. Hence, $k=k'$.
 \end{proof}

\begin{lemma}
  \label{lem:consistency1}
Assume that the leaders of epochs $e_1$ and $e_2$ such that $e_1\neq
e_2$ send messages $\ACCEPT(e_1, k_1, m_1)$ and $\ACCEPT(e_2, k_2,
m_2)$ respectively. Let $\msgi$ be the value of
 the $\msg$ array at the leader of $\eid_2$ when it sent
 $\ACCEPT(\eid_2, k_2, \vid_2)$. Assume that a process $\pid$
  delivers both, delivering $\vid[_1]$ before $\vid[_2]$, and that
  $\pid$ has
 $\prefixmsg[k_2][\msg] = \msgi$ after delivering $m_2$. Then
  $\msgi[k_1]=m_1$ and $k_1<k_2$.
\end{lemma}
\begin{proof}
  We have that:
  \begin{itemize}
  \item $p_i$ receives a message $\COMMIT(e_1',k_1')$ sent by the
         leader of an epoch $e_1'$, and $p_i$ has $\msg[k_1'] = m_1$
         when receiving this message;
   \item $p_i$ receives a message $\COMMIT(e_2',k_2')$ sent by the
         leader of an epoch $e_2'$, and $p_i$ has $\msg[k_2'] = m_2$
         when receiving this message. 
       \end{itemize}

By Lemma~\ref{thm:integrity} and \equref{prop:env-1},
 $e_2 \le e_2'$ and  $e_1 \le e_1'$. Then by Invariant~\ref{inv:accept-commit},
 $k_2=k_2'$ and $k_1=k_1'$. Furthermore, since $p_i$ delivers $m_1$ before $m_2$,
 then $e_1'\le e_2'$ and $k_1<k_2$, as required.

   Assume
that $e_1'=e_2'$. Since $p_i$ delivers both $m_1$ and $m_2$ in the
same epoch and $k_1<k_2$, $p_i$ has
$\msg[k_1]=m_1$ when it handles $\COMMIT(e_2',k_2)$. Assume now that 
$e_1'<e_2'$. By Invariant~\ref{inv:prefix-of-higher-epoch}, $p_i$ has
$\msg[k_1]=m_1$ when it handles $\COMMIT(e_2',k_2)$.
Thus, in both cases $p_i$ has
$\msg[k_1]=m_1$ when it handles $\COMMIT(e_2',k_2)$. Furthermore, we
have that $p_i$ has 
 $\prefixmsg[k_2][\msg] = \msgi$ after handling
 $\COMMIT(e_2',k_2)$. Hence, $\msgi[k_1]=m_1$, as required.      
\end{proof}

\begin{lemma}
\label{thm:commit-msg:po}
Assume that the leader $p_i$ of an epoch $e_1$ sends $\COMMIT(e_1, k)$ while
having $\msg[k] = m_1$, and the leader $p_j$ an epoch $e_2$ sends
$\COMMIT(e_2, k)$ while having $\msg[k] = m_2$. Then $m_1 = m_2$.
\end{lemma}
\begin{proof}
  The proof is analogous to the proof of Lemma~\ref{thm:commit-msg} for atomic
  broadcast in \S\ref{sec:correctness}.
\end{proof}

\begin{lemma}
  \label{lem:deliverorspec}
  Assume that a process $p_i$ has $\msg[k]=m$ when joining an epoch
  $e$ for which it is the leader. Then $p_i$ delivers $m$ before joining $e$ or speculatively delivers
  $m$ at position $k$ when joining $e$.
\end{lemma}
\begin{proof}
Let $k_0$ be the value of the
 variable $\lastdelivered$ at $p_i$ when it joins $e$.
Consider first the case when $k_0<k$. Then $p_i$
speculatively delivers $m$ when joining $e$. Furthermore, by
Property~\ref{prop:sdeliva} [Basic Speculative Delivery Properties],
$p_i$ does not deliver $m$ before joining $e$, as required.
Consider now the case
when $k_0\geq k$. Then $p_i$ has delivered position $k$ before
joining $e$. Thus, it has handled a $\COMMIT$ message for
position $k$ before joining $e$ while having $\msg[k]=m'$. By Lemma~\ref{thm:commit-msg:po}, $m'=m$. Hence, in
this case, $p_i$ delivers $m$ before joining $e$.  Furthermore, by
Property~\ref{prop:sdeliva} [Basic Speculative Delivery Properties],
$p_i$ does not speculatively deliver $m$ when joining $e$, as required.
\end{proof}

\begin{lemma}
  \label{lem:pim1m2}
  Assume that a process $p_i$ delivers $m_2$ after handling
  $\COMMIT(e, k_2)$ while having $\msg[k_1]=m_1$ such that
  $k_1<k_2$. Then $p_i$ delivers $m_1$ before $m_2$.
\end{lemma}
\begin{proof}
Since $k_1<k_2$, then $p_i$ has already delivered
position $k_1$ when it delivers $m_2$ at position $k_2$. Assume that $p_i$ delivers position $k_1$ after
handling $\COMMIT(e', k_1)$ while having $\msg[k_1]=m'$. We trivially
have that $e'\le e$. Consider the case when $e'=e$. After
handling  $\COMMIT(e, k_1)$, $p_i$ does not overwrites position
$k_1$ while in $e$. Since $p_i$ has $\msg[k_1]=m_1$ when it
handles $\COMMIT(e, k_2)$, then $m'=m_1$. Consider now the case
when $e'<e$. By Invariant~\ref{inv:prefix-of-higher-epoch}, $p_i$
has $\msg[k_1]=m'$ when it handles $\COMMIT(e, k_2)$. But, $p_i$ has $\msg[k_1]=m_1$ when it
handles $\COMMIT(e, k_2)$. Thus, $m'=m_1$. Therefore, in both
cases $p_i$ delivers $m_1$ before delivering $m_2$, as required.
\end{proof}

  \begin{proof}[Proof of \Consistency{} (a)]
  We have that:
  \begin{itemize}
\item the leader of $e_1$ sends $\ACCEPT(e_1, k_1, m_1)$ while
  in $e_1$;
  \item the leader of $e_2$ sends $\ACCEPT(e_2, k_2, m_2)$ while
    in $e_2$. Let $\msgi$ be the value of
 the $\msg$ array at the leader of $\eid_2$ when it sent
 $\ACCEPT(\eid_2, k_2, \vid_2)$;
\end{itemize}

The process $p_i$ delivers $m_2$ after handling $\COMMIT(e_2',
k_2')$ while having $\msg[k_2']=m_2$. By
Lemma~\ref{lem:setposition}, $k_2=k_2'$. Furthermore, by Lemma~\ref{thm:integrity},
$e_2\le e_2'$. Consider the case when $e_2=e_2'$. Then $p_i$ has
 handled $\ACCEPT(e_2, k_2, m_2)$. Thus, by
 Invariant~\ref{inv:prefix-of-leader}, $p_i$ has
 $\prefixmsg[k_2][\msg] = \msgi$ after handling
 $\COMMIT(e_2',k_2)$. Consider now the case when $e_2<e_2'$. Then by 
Invariant~\ref{inv:new_state_prefix}, we get the same. Hence, (*) $p_i$ has
 $\prefixmsg[k_2][\msg] = \msgi$ after handling
 $\COMMIT(e_2',k_2)$.

The process $p_j$ broadcast $m_2$ in $e'$. By
\equref{prop:env-1}, $p_j$ is the leader of $e_2$ and
$e_2=e'$. Therefore, $\msgi$ is the value of
 the $\msg$ array at $p_j$ when it sent
 $\ACCEPT(e', k_2, \vid_2)$.

\begin{enumerate}[i)]
\item Consider first the case when $p_i$ delivers $m_1$ before
  $m_2$. We prove that
  $p_j$ delivers $m_1$ before joining $e'$ or speculatively delivers
  $m_1$ before $m_2$ when joining $e'$. By
  (*) and Lemma~\ref{lem:consistency1}, $\msgi[k_1]=m_1$ and $k_1<k_2$. Thus,
  $p_j$ has $\msg[k_1]=m_1$ when it sends  $\ACCEPT(e', k_2,
  m_2)$. We have that $e_1\neq e'$. Then by \equref{prop:env-1},
  $m_1$ cannot be broadcast in $e'$. Thus, $p_j$ has $\msg[k_1]=m_1$
  when joining $e'$. Furthermore, by
  Lemma~\ref{lem:deliverorspec}, $p_j$ delivers $m_1$ before joining $e'$ or speculatively delivers
  $m_1$ at position $k_1$ when joining $e'$. Since $k_1<k_2$, in the
  case that $p_j$ speculatively delivers $m_1$, it does it before
  $m_2$, as required.
  
  \item  Consider now the case when  
  $p_j$ delivers $m_1$ before joining $e'$ or speculatively delivers
  $m_1$ before $m_2$ when joining $e'$. We prove that $p_i$ delivers
  $m_1$ before $m_2$. Assume that $p_j$ delivers $m_1$ before joining $e'$. Then $p_i$
  handles a $\COMMIT$ message for a position $k'$ while having
  $\epoch<e'$ and $\msg[k']=m_1$.
  By
  Invariant~\ref{inv:prefix-of-higher-epoch}, $p_j$ has $\msg[k']=m_1$
  when joining $e'$. Since $p_j$ broadcast $m_2$ in $e'$,
  then $k'<k_2$. Furthermore, by Invariant~\ref{inv:accept-commit},
  $k'=k_1$. Thus, $p_j$ has $\msg[k_1]=m_1$ when joining $e'$ and $k_1<k_2$.  Assume now that $p_j$ speculatively
  delivers $m_1$ when joining $e'$. Then $p_j$ has
  $\msg[k']=m_1$ when joining $e'$. Since $p_j$ broadcast $m_2$ in $e'$,
  then $k'<k_2$. Furthermore, by
Lemma~\ref{lem:setposition}, $k'=k_1$. Thus, $p_j$ has $\msg[k_1]=m_1$
when joining $e'$ and $k_1<k_2$ in both cases. This implies that $p_j$
has $\msg[k_1]=m_1$  when it sends
 $\ACCEPT(e', k_2, \vid_2)$. Therefore, 
$\msgi[k_1]=m_1$. Thus by (*), $p_i$ has $\msg[k_1]=m_1$ when
it handles $\COMMIT(e_2', k_2)$. By Lemma~\ref{lem:pim1m2}, $p_i$
delivers $m_1$ before $m_2$ as required.
\end{enumerate}
  \end{proof} 

 \begin{proof}[Proof of \Consistency{} (b)]
  We have that:
  \begin{itemize}
\item the leader of $e_1$ sends $\ACCEPT(e_1, k_1, m_1)$ while
  in $e_1$;
  \item the leader of $e_2$ sends $\ACCEPT(e_2, k_2, m_2)$ while
    in $e_2$. Let $\msgi$ be the value of
 the $\msg$ array at the leader of $\eid_2$ when it sent
 $\ACCEPT(\eid_2, k_2, \vid_2)$;
\end{itemize}

The process $p_i$ delivers $m_2$ after handling $\COMMIT(e_2',
k_2')$ while having $\msg[k_2']=m_2$. By
Lemma~\ref{lem:setposition}, $k_2=k_2'$. Furthermore, by Lemma~\ref{thm:integrity},
$e_2\le e_2'$. Consider the case when $e_2=e_2'$. Then $p_i$ has
 handled $\ACCEPT(e_2, k_2, m_2)$. Thus, by
 Invariant~\ref{inv:prefix-of-leader}, $p_i$ has
 $\prefixmsg[k_2][\msg] = \msgi$ after handling
 $\COMMIT(e_2',k_2)$. Consider now the case when $e_2<e_2'$. Then by 
Invariant~\ref{inv:new_state_prefix}, we get the same. Hence, (*) $p_i$ has
 $\prefixmsg[k_2][\msg] = \msgi$ after handling
 $\COMMIT(e_2',k_2)$.

The process $p_j$ speculatively delivers $m_2$ when joining $e'$. Then
$p_j$ has $\msg[k]=m_2$ when it sets $\epoch=e'$. By
Lemma~\ref{lem:setposition}, $k=k_2$. Then by Invariant~\ref{inv:new_state_prefix}, (**) $p_j$ has
$\prefixmsg[k_2][\msg] = \msgi$ when joining $e'$.

\begin{enumerate}[i)]
\item Consider first the case when $p_i$ delivers $m_1$ before
  $m_2$. We prove that
  $p_j$ delivers $m_1$ before joining $e'$ or speculatively delivers
  $m_1$ before $m_2$ when joining $e'$. By Lemma~\ref{lem:consistency1}
  and (*), $\msgi[k_1]=m_1$ and $k_1<k_2$. Then by (**) $p_j$ has
  $\msg[k_1]=m_1$ when joining $e'$. Furthermore, by 
  Lemma~\ref{lem:deliverorspec}, $p_j$ delivers $m_1$ before joining $e'$ or speculatively delivers
  $m_1$ at position $k_1$ when joining $e'$. Since $k_1<k_2$, in the
  case that $p_j$ speculatively delivers $m_1$, it does it before
  $m_2$, as required.
  
  \item Consider now the case when  
  $p_j$ delivers $m_1$ before joining $e'$ or speculatively delivers
  $m_1$ before $m_2$ when joining $e'$. We prove that $p_i$ delivers
  $m_1$ before $m_2$. Assume that $p_j$ delivers $m_1$ before joining $e'$. Then $p_i$
  handles a $\COMMIT$ message for a position $k'$ while having
  $\epoch<e'$ and $\msg[k']=m_1$.
  By
  Invariant~\ref{inv:prefix-of-higher-epoch}, $p_j$ has $\msg[k']=m_1$
  when joining $e'$ such that $k'<k_2$. Furthermore, by Invariant~\ref{inv:accept-commit},
  $k'=k_1$. Thus, $p_j$ has $\msg[k_1]=m_1$ when joining $e'$ and $k_1<k_2$.  Assume now that $p_j$ speculatively
  delivers $m_1$ before $m_2$ when joining $e'$. Then $p_j$ has
  $\msg[k']=m_1$ when joining $e'$ such that $k'<k_2$.  By
Lemma~\ref{lem:setposition}, $k'=k_1$. Thus, $p_j$ has $\msg[k_1]=m_1$
when joining $e'$ and $k_1<k_2$ in both cases. By (**), this implies that
$\msgi[k_1]=m_1$. Furthermore, by (*), $p_i$ has $\msg[k_1]=m_1$ when
it handles $\COMMIT(e_2', k_2)$. By Lemma~\ref{lem:pim1m2}, $p_i$
delivers $m_1$ before $m_2$ as required.
\end{enumerate}
\end{proof}